%% file: main.tex
\documentclass[11pt]{article}      
\usepackage[margin=1in,bottom=1in]{geometry}  
\usepackage{libertine}
\usepackage{amsthm}
\usepackage{amsmath}
\DeclareMathOperator{\symdiff}{\triangle}
\usepackage{amssymb}
\usepackage{stmaryrd}
\usepackage{mathrsfs} 
\usepackage{mathtools}
\usepackage{thm-restate}
\usepackage[hidelinks]{hyperref}
\usepackage[capitalize]{cleveref}
\usepackage{url}
\usepackage[size=tiny]{todonotes}
\usepackage{mathrsfs} 
\usepackage{soul}
\usepackage{relsize}
\usepackage{enumerate}
\usepackage[all,defaultlines=3]{nowidow}
\usepackage[mathlines]{lineno}
\usepackage{multicol}
\usepackage{enumitem}
\usepackage{textpos}
\setlist[itemize]{topsep=4pt,itemsep=3pt,parsep=0pt} 
\setlist[enumerate]{topsep=4pt,itemsep=3pt,parsep=0pt} 

\crefname{claim}{Claim}{Claims}
\crefname{figure}{Figure}{Figures}
\renewcommand{\preceq}{\preccurlyeq}

\renewcommand{\subset}{\subseteq}
\newtheorem{theorem}{Theorem}

\newtheorem{corollary}[theorem]{Corollary}
\newtheorem{conjecture}{Conjecture}

\newtheorem{observation}[theorem]{Observation}
\newtheorem{lemma}[theorem]{Lemma}

\newtheorem{claim}{Claim}
\newtheorem{fact}[theorem]{Fact}
\newtheorem{proposition}[theorem]{Proposition}
\theoremstyle{definition}
\newtheorem{definition}{Definition}
\theoremstyle{plain}

\theoremstyle{definition}

\newtheorem*{example*}{Example}

\newenvironment{claimproof}[1][\proofname]{%
  \begin{proof}[#1]%
}{%
  \end{proof}%
}

\newcommand{\Prob}{\mathbb{P}}
\newcommand{\Exp}{\mathbb{E}}

\def\phi{\varphi}
\newcommand{\N}{\mathbb{N}}

\def\CC{\mathscr{C}}

\def\atp{\textnormal{atp}}

\newcommand{\flipp}{\mathrm{flip}}
\newcommand{\flipped}{\mathrm{flipped}}

\newcommand{\Xx}{\mathcal{X}}
\newcommand{\Cc}{\mathscr{C}}

\newcommand{\Th}{\mathrm{Th}}
\newcommand{\Oof}{\mathcal{O}}

\newcommand{\str}[1]{\mathbf{#1}}

\newcommand{\Ex}{\mathbb{X}}

\newcommand{\pc}{\mathrm{predictedCol}}
\renewcommand{\check}{\mathrm{check}}
\newcommand{\lc}{\mathrm{lc}}

\newcommand{\otp}{\mathrm{otp}}
\newcommand{\Ball}{\mathrm{Ball}}
\newcommand{\dist}{\mathrm{dist}}
\newcommand{\Ff}{\mathcal{F}}
\newcommand{\Cols}{\LL}
\newcommand{\Tc}{\mathsf{T}}

\newcommand{\Cover}{\mathcal{K}}

\newcommand{\from}{\colon}

\newcommand{\DD}{\mathscr D}

\newcommand{\LL}{\mathcal L}
\newcommand{\BB}{\mathscr B}
\newcommand{\WW}{\mathscr W}
\newcommand{\C}{\mathscr{C}}
\newcommand{\Dd}{\mathscr D}
\def\N{\mathbb N} 
\def\epsilon{\varepsilon}
\def\eps{\varepsilon}
\newcommand{\Oh}{\Oof}

\newcommand{\PP}{\mathcal P}

\newcommand{\flipCol}{\lambda}

\newcommand{\Ramsey}{\mathcal{R}}
\newcommand{\RamseyGrid}{\Ramsey_{\mathrm{grid}}}
\newcommand{\RamseyBip}{\Ramsey_{\mathrm{bip}}}

\newcommand{\Inc}{\mathsf{Inc}}

\renewcommand{\emptyset}{\varnothing}

\newcommand{\subclique}[3][1.5]{
\begin{tikzpicture}[scale=#1]
  \foreach \i in {1,...,#3} {
    \node[circle, draw, fill=black, inner sep=0.8pt] (N\i) at ({360/#3 * (\i - 1)}:2) {};
  }
  \foreach \i in {1,...,#3} {
    \foreach \j in {\i,...,#3} {
      \draw (N\i.center) -- (N\j.center)
      \foreach \k in {1,...,#2} {
      node[circle, draw, fill=black, inner sep=1pt, pos=\k/(#2+1)] (M\i\j\k) {}
      };
    }
  }
\end{tikzpicture}
}
\newcommand{\spiderweb}[3][1.5]{
\begin{tikzpicture}[scale=#1]
  \foreach \i in {1,...,#3} {
    \node[circle, draw, fill=black, inner sep=0.8pt] (N\i) at ({360/#3 * (\i - 1)}:2) {};
  }
  \foreach \i in {1,...,#3} {
    \foreach \j in {\i,...,#3} {
      \draw (N\i.center) -- (N\j.center)
      \foreach \k in {1,...,#2} {
      node[circle, draw, fill=black, inner sep=1pt, pos=\k/(#2+1)] (M\i\j\k) {}
      };
    }
  }
  \foreach \i in {1,...,#3} {
    \foreach \j in {\i,...,#3} {
        \foreach \k in {\i,...,#3} {
            \draw (M\i\j1.center) -- (M\i\k1.center);
        }
        \foreach \k in {1,...,\i} {
            \draw (M\i\j1.center) -- (M\k\i#2.center);
        }
    }
    \foreach \j in {1,...,\i} {
        \foreach \k in {1,...,\i} {
            \draw (M\j\i#2.center) -- (M\k\i#2.center);
        }
    }
  }
\end{tikzpicture}
}
\newcommand{\biweb}[4][1.2]{
\begin{tikzpicture}[scale=#1]
  \foreach \i in {1,...,#3} {
    \node[circle, draw, fill=black, inner sep=0.6pt] (B\i) at (-0.8*#3-0.8+1.6*\i,0) {};    
  }
  \foreach \j in {1,...,#4} {
    \node[circle, draw, fill=black, inner sep=0.6pt] (T\j) at (-0.8*#4-0.8+1.6*\j,4) {};    
  }
  \foreach \i in {1,...,#3} {
    \foreach \j in {1,...,#4} {
      \draw (B\i.center) -- (T\j.center)
      \foreach \k in {1,...,#2} {
      node[circle, draw, fill=black, inner sep=0.6pt, pos=\k/(#2+1)] (M\i\j\k) {} 
      };
    }
    }
    \foreach \i in {1,...,#3} {
    \foreach \j in {1,...,#4} {
        \foreach \k in {\j,...,#4} {
            \pgfmathsetmacro\plus{\j+1}
            \ifnum\k=\plus
            {\draw  (M\i\j1.center) edge (M\i\k1.center);}
            \else
            {\draw  (M\i\j1.center) edge[bend left] (M\i\k1.center);}
            \fi
        }
        \foreach \k in {\i,...,#3} {
            \pgfmathsetmacro\aplus{\i+1}
            \ifnum\k=\aplus
            {\draw  (M\i\j#2.center) edge (M\k\j#2.center);}
            \else
            {\draw  (M\i\j#2.center) edge[bend right] (M\k\j#2.center);}
            \fi
        }
    }
    }
\end{tikzpicture}
}

\newcommand{\ioannis}[1]{\todo[color=white!40]{Ioannis: #1}}

\newcommand{\michal}[1]{\todo[color=pink!40]{Michal: #1}}

\renewcommand{\leq}{\leqslant}
\renewcommand{\geq}{\geqslant}
\renewcommand{\le}{\leq}
\renewcommand{\ge}{\geq}

\renewcommand{\setminus}{-}
\newcommand{\set}[1]{\{#1\}}
\newcommand{\setof}[2]{\{#1 \colon#2\}}

\begin{document}

\newcommand{\funding}{
I.E. was supported by a George and Marie Vergottis Scholarship awarded by Cambridge Trust, an Onassis Foundation Scholarship, and a Robert Sansom Studentship. 
N.M. was supported by the German Research Foundation (DFG) with grant agreement No. 444419611.
R.M. was supported by the
National Science Foundation (NSF) with grant No. DMS-2202961.
M.P. and S.T. were supported by the project BOBR that is funded from the European Research Council (ERC) under the European Union’s Horizon 2020 research and innovation programme with grant agreement No. 948057.
A part of this work was done during the 1st workshop on Twin-width, which was partially financed by the grant ANR ESIGMA (ANR-17-CE23-0010) of the French National Research Agency.}

\title{First-Order Model Checking on \\ Monadically Stable Graph Classes\thanks{\funding}}
\date{\today}
\author{
  Jan Dreier \\
  \small{TU Wien} \\
  \small{dreier@ac.tuwien.ac.at}
  \and
  Ioannis Eleftheriadis \\
  \small{University of Cambridge} \\
  \small{ie257@cam.ac.uk} 
  \and
  Nikolas M\"ahlmann \\
  \small{University of Bremen} \\
  \small{maehlmann@uni-bremen.de}
  \and
  Rose McCarty \\
  \small{Princeton University} \\
  \small{rm1850@math.princeton.edu}
  \and
  Michał Pilipczuk \\
  \small{University of Warsaw} \\
  \small{michal.pilipczuk@mimuw.edu.pl}
  \and
  Szymon Toru\'nczyk \\
  \small{University of Warsaw} \\
  \small{szymtor@mimuw.edu.pl}
}
\maketitle

\input{chapters/abstract.tex}

\begin{textblock}{20}(-1.75, 0.9)
\includegraphics[width=40px]{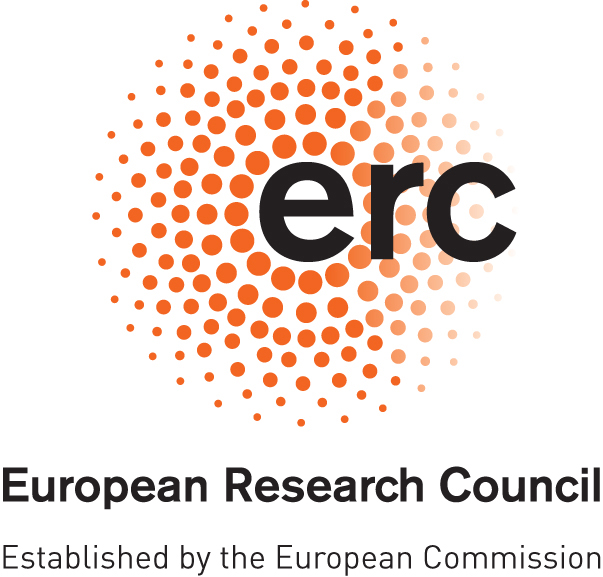}%
\end{textblock}
\begin{textblock}{20}(-1.75, 1.9)
\includegraphics[width=40px]{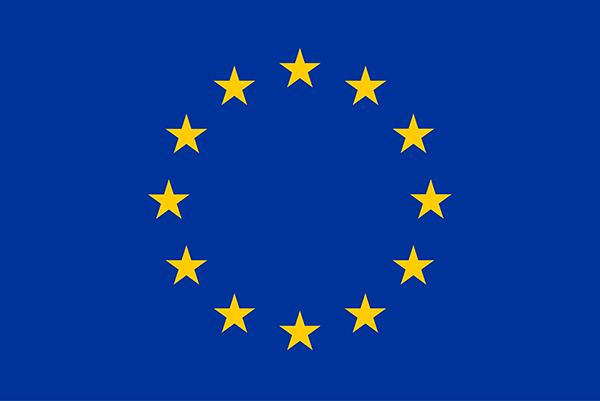}%
\end{textblock}

\thispagestyle{empty}

\newpage

\tableofcontents
\thispagestyle{empty}

\newpage

\clearpage
\setcounter{page}{1}

\input{chapters/intro.tex}

\input{chapters/prelims.tex}

\input{chapters/tractability.tex}

\input{chapters/hardness.tex}

\input{chapters/conclusions.tex}

\bibliographystyle{alpha}
\bibliography{references}

\appendix
\input{chapters/appendix.tex}

\input{chapters/tractability/appendix-sparse-covers.tex}
\input{chapters/hardness/sop.tex}

\end{document}

%% file: chapters/abstract.tex
\begin{abstract}
 A graph class $\C$ is called {\em{monadically stable}} if one cannot interpret, in first-order logic, arbitrary large linear orders in colored graphs from $\C$. We prove that the model checking problem for first-order logic is fixed-parameter tractable on every monadically stable graph class. This extends the results of [Grohe, Kreutzer, and Siebertz; J. ACM '17] for nowhere dense classes and of [Dreier, M\"ahlmann, and Siebertz; STOC '23] for structurally nowhere dense classes to all monadically stable classes.

 This result is complemented by a hardness result showing that monadic stability is precisely the dividing line between tractability and intractability of first-order model checking on hereditary classes that are {\em{edge-stable}}: exclude some half-graph as a semi-induced subgraph. Precisely, we prove that for every hereditary graph class $\C$ that is edge-stable but not monadically stable, first-order model checking is $\mathrm{AW}[*]$-hard on $\C$, and $\mathrm{W}[1]$-hard when restricted to existential sentences. This confirms, in the special case of edge-stable classes, an on-going conjecture that the notion of {\em{monadic NIP}} delimits the tractability of first-order model checking on hereditary classes of graphs.

 For our tractability result, we first prove that monadically stable graph classes have almost linear neighborhood complexity. We then use this result to construct sparse neighborhood covers for monadically stable graph classes, which provides the missing ingredient for the algorithm of [Dreier, M\"ahlmann, and Siebertz; STOC~2023]. The key component of this construction is the usage of orders with low crossing number~[Welzl; SoCG '88], a tool from the area of range queries.

 For our hardness result, we first prove a new characterization of monadically stable graph classes in terms of forbidden induced subgraphs. We then use this characterization to show that in hereditary classes that are edge-stable but not monadically stable, one can effectively interpret the class of all graphs using only existential formulas.
\end{abstract}

%
%

%% file: chapters/intro.tex
\section{Introduction}\label{sec:intro}

We study the model checking problem for first-order logic on graphs: given a first-order sentence $\varphi$ in the vocabulary of graphs (consisting of one binary adjacency relation $E(\cdot,\cdot)$) and a graph $G$, decide whether $\varphi$ holds in $G$. Arguably, this is the most fundamental problem of algorithmic model theory, and it can be used to model many concrete problems of interest. For instance, the following two sentences express that the graph in question contains a clique of size $k$ or a dominating set of size $k$,  respectively.

\[\exists x_1 \ldots \exists x_k \ \bigwedge_{1\leq i<j\leq k} E(x_i,x_j)\qquad\qquad\qquad \exists x_1 \ldots \exists x_k \forall y\ \bigvee_{1\leq i\leq k} (y=x_i)\vee E(y,x_i)\]

Clearly, first-order model checking can be solved by brute-force in time $n^{\Oh(\|\varphi\|)}$ on $n$-vertex graphs.
However, as the problem is \(\mathrm{AW}[*]\)-hard, we cannot expect a 
\emph{fixed-parameter tractable} algorithm, that is,
an algorithm with a running time of the form $f(\|\varphi\|)\cdot n^c$ for a computable function $f$ and constant $c$,
unless $\mathrm{FPT}=\mathrm{AW}[*]$ and the whole hierarchy of parameterized complexity classes collapses. 
Therefore, significant effort has been devoted to understanding on which restricted classes of graphs first-order model checking is fixed-parameter tractable, as well as finding the dividing line between tractability and intractability.
Despite decades of intensive studies, a comprehensive understanding of the answer to this question has not yet been achieved.

A landmark result in this area is the following theorem of Grohe, Kreutzer, and Siebertz~\cite{GroheKS17}: If a graph class $\C$ is {\em{nowhere dense}}, then the first-order model checking problem on $\Cc$ can be solved in time $f(\|\varphi\|,\eps)\cdot n^{1+\eps}$, for any $\eps>0$. Here, nowhere denseness is a structural property of a graph class that turns out to be precisely the dividing line for subgraph-closed classes: if a subgraph-closed class of graphs $\C$ is not nowhere dense, then model checking on $\C$ is as hard as on general graphs, that is, $\mathrm{AW}[*]$-hard.

While the result of Grohe, Kreutzer, and Siebertz completely explains the complexity of first-order model checking on subgraph-closed classes, there are natural classes of graphs that are not nowhere dense where the problem is still fixed-parameter tractable. Examples are classes of bounded cliquewidth~\cite{courcelle2000linear} and, more generally, classes of bounded twin-width\footnote{The caveat is that in the twin-width case, a suitable decomposition needs to be provided with the input.}~\cite{BonnetKTW22}. The explanation here is that these classes are not subgraph-closed, and therefore do not fall under the classification provided by~\cite{GroheKS17}. Finding a suitable dividing line for classes that are {\em{hereditary}} --- closed under taking induced subgraphs --- remains a central problem in the area.

Around 2016 the following conjecture has been formulated. Its resolution would completely explain the hereditary case.

\begin{conjecture}[see e.g.~\cite{Warwick16,BonnetGMSTT22,ssmc,GajarskyHOLR20}]\label{conj:main}
 Let $\C$ be a hereditary class of graphs. If $\C$ is monadically NIP, then first-order model checking is fixed-parameter tractable on $\C$. Otherwise, if $\C$ is not monadically NIP, then first-order model checking is $\mathrm{AW}[*]$-hard on $\C$.
\end{conjecture}

Here, {\em{monadic NIP}}\footnote{NIP stands for {\em{Not the Independence Property}}. Hence, an alternative name for monadic NIP is the term {\em{monadic dependence}}. In this work we use the abbreviation NIP for compatibility with the model-theoretic literature.} is a notion borrowed from model theory: $\C$ is monadically NIP if one cannot interpret, in first-order logic, the class of all graphs in vertex-colored graphs from $\C$. Thus, \cref{conj:main}, if true, would provide a magnificent link between model theory and complexity theory, showing that an important dividing line studied in one field also fundamentally manifests itself in the other. 

The hardness part of \Cref{conj:main} has very recently been proven~\cite{CharacterizingNIP}.
There are good reasons to believe that also the tractability part holds:
\begin{itemize}
 \item Adler and Adler~\cite{adler2014interpreting} observed that for subgraph-closed classes, monadic NIP collapses to nowhere denseness, which is the dividing line by the results of~\cite{GroheKS17}. A result of Dvo\v{r}\'ak~\cite{Dvorak18} implies that this collapse holds even for classes that are {\em{weakly sparse}}: exclude some biclique as a~subgraph.
 \item For classes of {\em{ordered graphs}} (graphs equipped with a linear order on the vertex set), monadic NIP is equivalent to the boundedness of twin-width and provides the sought dividing line~\cite{BonnetGMSTT22}.
\end{itemize}
However, a complete resolution of \cref{conj:main} remains out of reach for now, mostly due to the insufficient combinatorial understanding of monadically NIP classes.

It has been long postulated that to approach the tractability part of \cref{conj:main}, one should focus first on the simpler case of monadically stable classes. Here, we say that $\C$ is {\em{monadically stable}} if one cannot interpret arbitrarily long linear orders in vertex-colored graphs from $\C$. Intuitively, monadic NIP means that the class in question is not fully rich from the point of view of first-order logic, whereas graphs in monadically stable classes are even more restricted: they are {\em{orderless}}, in the sense that there is no apparent linear order on any sizeable portion of the graph that could be inferred from its structure. Notably, every nowhere dense graph class is monadically stable~\cite{adler2014interpreting,podewski1978stable}, so studying monadically stable classes already generalizes the setting considered by Grohe et al.~\cite{GroheKS17}.

Stability is a concept studied within model theory for decades, mostly in connection with Shelah's classification theory; see~\cite{shelah1990classification}, or~\cite[Chapter~8]{tent_ziegler} for a more basic introduction. Monadic stability was first studied by Baldwin and Shelah~\cite{baldwin1985second} already in 1985; in particular, their work provides a fairly good understanding of the structure of monadically stable classes from the model-theoretic perspective. However, the work on describing combinatorial properties of monadically stable classes has started only recently, and it immediately provided algorithmic corollaries. Dreier et al.~\cite{DreierMST23} and Gajarsk\'y et al.~\cite{flippergame} proved two combinatorial characterizations of monadically stable graphs classes --- through the notions of {\em{flip-flatness}} and of the so-called {\em{Flipper game}} --- which mirror two classic characterizations of nowhere denseness --- through {\em{uniform quasi-wideness}}~\cite{Dawar10} and the {\em{Splitter game}}~\cite{GroheKS17}. Particularly, the Splitter game was one of the two key ingredients in the fixed-parameter model-checking algorithm for nowhere dense classes of Grohe et al.~\cite{GroheKS17}. Very recently, Dreier, M\"ahlmann, and Siebertz~\cite{ssmc} showed how to use the Flipper game to perform model checking on monadically stable classes in fixed-parameter time under the assumption that a second key ingredient --- {\em{sparse neighborhood covers}} --- is also present. Here, a {\em{distance-$r$ neighborhood cover}} of a graph $G$ of diameter $d$ and overlap $k$ is a family $\Cover$ of subsets of vertices of $G$, called {\em{clusters}}, such that
\begin{itemize}
 \item for each vertex $u$ of $G$, there is a cluster $C\in \Cover$ that contains the $r$-neighborhood of $u$;
 \item for each cluster $C\in \Cover$, all vertices of $C$ are pairwise at distance at most $d$; and
 \item each vertex of $G$ belongs to at most $k$ clusters of $\Cover$.
\end{itemize}
The method of Dreier et al. works whenever graphs from the considered class $\C$ admit distance-$r$ neighborhood covers of diameter\footnote{We follow the convention that the subscript of the $\Oh(\cdot)$ notation signifies the parameters on which the hidden constants may depend. For instance, the constant hidden by $\Oh_{\C,r,\eps}(\cdot)$ may depend on the class $\C$, integer $r$, and positive real $\eps$.} $\Oh_{\C,r}(1)$ and overlap $\Oh_{\C,r,\eps}(n^\eps)$ for any $\eps>0$ and $r\in \N$, where $n$ is the vertex count. They prove that this is the case for an important subset of monadically stable classes, namely {\em{structurally nowhere dense classes}}, but the general monadically stable case remained open.

\paragraph*{Our contribution.} In this work we give a construction for sparse neighborhood covers in general monadically stable graph classes, and thereby prove that first-order model checking is fixed-parameter tractable on every monadically stable graph class. We also prove a complementary hardness result which shows that in the setting of hereditary classes that are {\em{edge-stable}} --- exclude a half-graph as a semi-induced subgraph\footnote{A half-graph $H_n$ of order $n$ is the bipartite graph with sides $\{a_1,\ldots,a_n\}$ and $\{b_1,\ldots,b_n\}$ where $a_i$ and $b_j$ are adjacent if and only if $i\leq j$. Containing $H_n$ as an semi-induced subgraph means that there are disjoint vertex subsets $A$ and $B$, each of size~$n$, so that the adjacency between $A$ and $B$ is exactly as in $H_n$.} --- monadic stability is precisely the dividing line delimiting tractability and intractability of first-order model checking. The following two statements describe our main results formally.

\begin{restatable}{theorem}{introtractability}\label{thm:intro-tractability}
    Let \(\Cc\) be a monadically stable graph class and $\eps>0$ be a positive real. Then there is an algorithm that, given an $n$-vertex graph \(G \in \C\) and a first-order sentence \(\phi\), decides whether $\phi$ holds in \(G\) in time \(\Oh_{\C,\eps,\phi}(n^{5+\epsilon})\).
\end{restatable}
\begin{restatable}{theorem}{introhardness}\label{thm:intro-hardness}
    Let $\C$ be a hereditary class of graphs that is edge-stable, but not monadically stable. Then first-order model checking on $\C$ is $\mathrm{AW}[*]$-hard, existential first-order model checking on $\C$ is $\mathrm{W}[1]$-hard, and {\sc{Induced Subgraph Isomorphism}} on $\C$ is $\mathrm{W}[1]$-hard with respect to Turing reductions.
\end{restatable}

As proved by Ne\v{s}et\v{r}il et al.~\cite{NesetrilMPRS21}, under the assumption of edge-stability, the notions of monadic stability and monadic NIP coincide. So \cref{thm:intro-tractability,thm:intro-hardness} verify \cref{conj:main} in the setting of edge-stable~classes. \cref{thm:intro-hardness} goes a step further and establishes also hardness of existential first-order model checking and of the {\sc{Induced Subgraph Isomorphism}} problem: Given $H$ and $G$, is $H$ an induced subgraph of $G$? Here, the size of $H$ is the parameter.

The proof of \cref{thm:intro-tractability} proceeds in two steps. The first step is to prove that monadically stable classes have {\em{almost linear neighborhood complexity}}, as explained formally in the theorem below.

\begin{restatable}{theorem}{introneicomp}\label{thm:nei-comp}
    Let $\C$ be a monadically stable graph class and $\eps>0$. Then for every $G\in \C$ and $A\subseteq V(G)$,
    \[|\{N_G[v]\cap A \colon v\in V(G)\}|\leq \Oh_{\Cc,\eps}(|A|^{1+\eps}).\]
\end{restatable}

In the language of the theory of VC dimension, \cref{thm:nei-comp} says that set systems of neighborhoods in graphs from monadically stable classes have near-linear growth rate. The proof of \cref{thm:nei-comp} relies on a reduction to the case of nowhere dense classes, which was known before~\cite{DBLP:conf/icalp/EickmeyerGKKPRS17}. This reduction exploits both sampling techniques in the spirit of $\eps$-nets, and the concept of the {\em{branching index}}, which is a prominent tool in the study of monadically stable classes within model theory.

Once \cref{thm:nei-comp} is established, we exploit another deep result from the theory of VC dimension: orders with low {\em{crossing number}} due to Welzl~\cite{Welzl89}. In a nutshell, Welzl proved that under the premise provided by \cref{thm:nei-comp}, the following is true: for every graph $G\in \C$, the vertices of $G$ can be linearly ordered so that the neighborhood of every vertex can be decomposed into at most $\Oh_{\C,\eps}(n^{\eps})$ intervals in the order. Given such a linear order, a greedy procedure can be used to construct a distance-$1$ neighborhood cover of diameter $4$ and overlap $\Oh_{\C,\eps}(n^{\eps})$. This can be then easily leveraged to also give distance-$r$ neighborhood covers of diameter $4r$ and overlap $\Oh_{\C,r,\eps}(n^{\eps})$, for all $r\in \N$.

Two remarks are in place. First, our method of constructing neighborhood covers
is algorithmic and can be executed in time $\Oh_{\Cc,\eps}(n^{4+\eps})$,
while the model checking algorithm of Dreier et al.~\cite{ssmc}
relied on a slower approximation algorithm based on rounding an LP relaxation.
Consequently, we improve the running time of the previous algorithm~\cite{ssmc} from $\Oh_{\C,\phi}(n^{11})$ 
to $\Oof_{\CC,\phi,\eps}(n^{5+\eps})$.
Second, both the proof of \cref{thm:nei-comp} and the actual construction of
neighborhood covers feature deep ideas from the theory of VC dimension, in
particular the concept of orders with low crossing number due to Welzl. We
anticipate that these tools will play an increasingly important role in the
emerging theory of monadically stable and monadically NIP classes.


For \cref{thm:intro-hardness}, we revisit the work of Gajarsk\'y et al.~\cite{flippergame}, who, within their study of the Flipper game, identified certain combinatorial patterns that must appear in edge-stable classes that are not monadically stable. Starting from those patterns, we repeatedly use Ramsey-type tools to refine their structure, until eventually we arrive at the following~statement.

\begin{restatable}{theorem}{intropatterns}\label{thm:intro-patterns}
    A graph class $\C$ is monadically stable if and only if it is edge-stable and there are no $k,r \geq 2$ such that $\C$ contains as induced subgraphs
    \begin{itemize}
     \item a $k$-flip of every $r$-subdivided clique; or
     \item a $k$-flip of every $r$-web.
    \end{itemize}
\end{restatable}

Here, an $r$-subdivided clique is a graph obtained from a clique by replacing every edge by a path of length $r+1$, while an $r$-web is a graph obtained from an $r$-subdivided clique by turning the neighborhood of every native (principal) vertex into a clique; see \cref{fig:intro-patterns}. Also, a {\em{$k$-flip}}\footnote{We remark that this operation is often called a {\em{$k$-perturbation}} within the community working on vertex-minors.} of a graph $G$ is any graph obtained from $G$ by partitioning the vertices into at most $k$ parts and, for every pair of parts $A,B$ (possibly $A=B$), either leaving the edge relation within $A\times B$ intact or complementing it.

 \begin{figure}[h!]
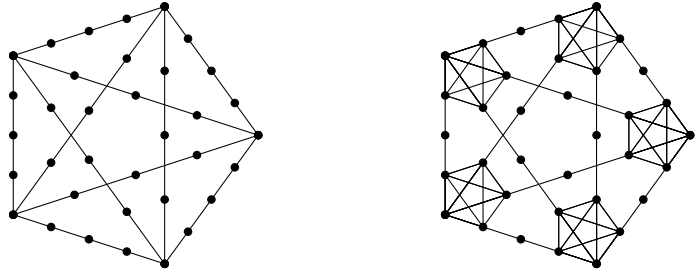

   \centering\small
   \subclique[.9]{3}{5}
   \qquad\qquad\qquad
   \spiderweb[.9]{3}{5}
   \caption{The $3$-subdivided clique of order $5$ and the $3$-web of order $5$.}\label{fig:intro-patterns}
 \end{figure}

It is not hard to see that first-order model checking is $\mathrm{AW}[*]$-hard on the hereditary closure of the class of $r$-subdivided cliques, for any $r\geq 2$, and the same also holds for $r$-webs. However, this does not automatically imply the statement of \cref{thm:intro-hardness}, as the obstructions extracted from \cref{thm:intro-patterns} can be obfuscated by $k$-flips. To bridge this issue, we craft a  very careful way of ``deobfuscating'' the flips by analyzing neighborhoods within $k$-flipped subdivided cliques and webs. This leads to the following statement, which is sufficient for concluding \cref{thm:intro-hardness}: every hereditary edge-stable graph class that is not monadically stable effectively interprets the class of all graphs using only existential formulas.

Again, two remarks are in place. First, \cref{thm:intro-patterns} seems to be of independent interest, as it leads to further model-theoretic corollaries beyond the immediate scope of this work. For instance, in \cref{sec:sop} we demonstrate this by showing that in the context of hereditary classes, Shelah's notion of NSOP (Not the Strict Order Property) collapses to monadic stability.
Second, a subset of the authors developed a characterization of monadically NIP graph classes that generalizes \cref{thm:intro-patterns}, from which they subsequently derived the hardness part of \cref{conj:main}~\cite{CharacterizingNIP}. This latter derivation is based on a different method of ``deobfuscating'' flips than the one used in this work. The method of~\cite{CharacterizingNIP} is simpler, but has the drawback of not relying only on existential formulas, so it does not give hardness of existential model checking. In this work we take the extra mile to derive the stronger hardness results described in~\cref{thm:intro-hardness}.

\paragraph*{Acknowledgements.} The authors wish to thank Jakub Gajarsk\'y, Pierre Ohlmann, Patrice Ossona de Mendez, Wojciech Przybyszewski, and Sebastian Siebertz, for many inspiring discussions around the topic of this work.

%% file: chapters/prelims.tex
\section{Preliminaries}\label{sec:apxPrelims}

\paragraph*{Basic notation and graph theory.}
For a positive integer $p$, we write $[p]\coloneqq \{1,\ldots,p\}$. For a function $f$ and a subset of its domain $A$, $f\vert_A$ denotes the restriction of $f$ to $A$. For a set $X$, $\binom{X}{2}$ denotes the set of all $2$-element subsets of $X$.

All graphs considered in this work are finite, undirected, and simple (loopless and without parallel edges), unless explicitly stated. We use standard graph-theoretic terminology. The ball of radius $r$ around a vertex $u$ in a graph $G$ is the set $\Ball^r_G[u]\coloneqq \{v\in V(G)~|~\dist_G(u,v)\leq r\}$, where $\dist_G(\cdot,\cdot)$ denotes the distance in $G$. The {\em{closed neighborhood}} of $u$ is $N_G[u]\coloneqq \Ball^1_G[u]$, and the {\em{open neighborhood}} is $N_G(u)\coloneqq N_G[u]\setminus \{u\}$. The graph $G$ may be dropped from the notation if it is clear from the context.

For a graph $G$ and $r\in \N$, the {\em{$r$-subdivision}} of $G$ is the graph $G'$ obtained by replacing every edge of $G$ with a path of length $r+1$. Within $G'$, the original vertices of $G$ are called {\em{native}}.

For a graph $G$ and a pair of disjoint vertex subsets $A$ and $B$, the subgraph {\em{semi-induced}} by $A$ and $B$,
denoted $G[A,B]$, is the bipartite graph with sides $A$ and $B$ that contains all edges of $G$ with one endpoint in $A$ and second in $B$. If the subgraph semi-induced by $A$ and $B$ is a biclique, then we say that $A$ and $B$ are {\em{complete}} to each other, and if it is edgeless, then $A$ and $B$ are {\em{anti-complete}} to each other. If $A$ and $B$ are complete or anti-complete to each other, then they are {\em{homogeneous}} to each other, and otherwise they are {\em{inhomogeneous}}.

As mentioned in \cref{sec:intro}, a class of graphs $\C$ is {\em{edge-stable}} if there exists some $n\in \N$ such that no graph in $\C$ contains $H_n$ --- the half-graph of order $n$ --- as a semi-induced subgraph. Here, $H_n$ is the bipartite graph with vertices $\{a_i,b_i\colon 1\leq i\leq n\}$ and edges $\{a_ib_j\colon 1\leq i\leq j\leq n\}$.

For a graph $G$ and $k\in \N$, a {\em{$k$-flip}} of $G$ is any graph $G'$ that can be obtained from $G$ as follows:
\begin{itemize}
 \item Take any (not necessarily proper) coloring $\flipCol\colon V(G)\to [k]$ of vertices of $G$ with $k$ colors.
 \item Take any symmetric relation $R\subseteq [k]^2$ and define $G'$ to be the graph on vertex set $V(G)$ such that for any distinct $u,v\in V(G)$,
 \[uv\in E(G')\qquad\textrm{if and only if}\qquad uv\in E(G)\ \mathrm{XOR}\ (\flipCol(u),\flipCol(v))\in R.\]
\end{itemize}
In other words, for all pairs of colors $(i,j)\in R$, the adjacency relation within $\flipCol^{-1}(i)\times \flipCol^{-1}(j)$ is complemented ({\em{flipped}}). In particular, if $i=j$, the adjacency relation within the induced subgraph $G[\flipCol^{-1}(i)]$ is complemented. Note also that a single graph $G$ has multiple $k$-flips, as the construction of $G'$ depends on the choice of $\flipCol$ and $R$. When we want to signify that a $k$-flip $G'$ is constructed using a particular choice of $\flipCol$ and $R$, we use the notation $G'=G\oplus_{\flipCol} R$.

\paragraph*{Interpretations and transductions.} We assume familiarity with basic terminology related to first-order logic. All considered formulas are first-order over the vocabulary of graphs, which consists of a single symmetric and irreflexive relation $E(\cdot,\cdot)$ that signifies adjacency. A formula is {\em{quantifier-free}} if it contains no quantifier, and {\em{existential}} if it contains no universal quantifiers and negation is applied only to quantifier-free subformulas. For a graph $G$ and a formula $\phi(x)$ with one free variable $x$, we use the shorthand $\phi(G) \coloneqq \{v \in V(G)~\colon~G\models \phi(v)\}$.
We use two basic mechanisms for definable graph transformations: \emph{interpretations} and \emph{transductions}.

\begin{definition}
    Let $\delta(x)$ and $\phi(x,y)$ be formulas, where $\phi$ is symmetric and irreflexive, that is, $G\models \neg \phi(u,u)\wedge (\phi(u,v)\leftrightarrow \phi(v,u))$ for all graphs $G$ and vertices $u,v$.
    The \emph{interpretation} $I_{\delta,\phi}$ is defined to be the operation that maps an input graph $G$ to the output graph $H \coloneqq I_{\delta,\phi}(G)$ with vertex and edge sets defined by \[V(H) \coloneqq \left\{v \in V(G) \colon G \models \delta(v)\right\}\quad \textrm{and}\quad  E(H)\coloneqq \left\{uv \in \binom{V(H)}{2}\colon G \models \phi(u,v) \right\}.\] Moreover, we say that an interpretation $I_{\delta,\phi}$ is \emph{existential} if the formulas $\delta$ and $\phi$ are existential.
    For a graph class \(\Cc\), let \(I_{\delta,\phi}(\Cc) := \setof{I_{\delta,\phi}(G)}{G \in \Cc}\).
\end{definition}

Note that in the definition above we have taken $\delta$ to have a single free variable (and consequently $\phi$ to have two free variables). This is usually referred to as a \emph{$1$-dimensional} interpretation, as opposed to general interpretations where the vertex set of the interpreted graph $I_{\delta,\phi}(G)$ may consist of tuples of vertices of~$G$. All interpretations considered in this paper are $1$-dimensional. Nonetheless, we shall sometimes refer to these as \emph{interpretations on singletons} to make this fact more~explicit.

A {\em{colored graph}} is a graph together with a finite number of unary predicates on its vertices. The definition of an interpretation naturally lifts to the setting when the input graph $G$ is colored; then $\delta$ and $\phi$ may also speak about colors of vertices. For a finite set of unary predicates $\Cols$ and a class $\C$, by $\C^\Cols$ we denote the class of all {\em{$\Cols$-colorings}} of graphs from $\C$, that is, graphs from $\C$ expanded by interpreting the predicates from $\Cols$ in any way.

\begin{definition}
 A {\em{transduction}} consists of a finite set of unary predicates $\Cols$ and an ($1$-dimensional) interpretation $I$ from $\Cols$-colored graphs to graphs. For a transduction $\Tc=(\Cols,I)$ and a graph class $\C$, we define the image of $\Tc$ on $\C$ as $\Tc(\C)\coloneqq I(\C^{\Cols})$.
 For a graph \(G\), we denote \(\Tc(G) := \Tc(\{G\})\).
 A graph class $\Dd$ can be {\em{transduced}} from $\C$ if there exists a transduction $\Tc$ such that $\Dd\subseteq \Tc(\C)$.
\end{definition}

We may now formally define monadic stability and monadic NIP.

\begin{definition}\label{def:monstable}
 A graph class $\C$ is {\em{monadically NIP}} if $\C$ does not transduce the class of all graphs. $\C$~is moreover {\em{monadically stable}} if $\C$ does not transduce the class of all half-graphs.
\end{definition}

Since the composition of two transductions is again a transduction, both monadic stability and monadic NIP persist under taking transductions. For example, observe that for every fixed $k$ and graph class $\C$, the class $\C_k$ consisting of all $k$-flips of graphs from $\C$ can be transduced from $\C$. Consequently, if $\C$ is monadically stable, then so is $\C_k$.

%% file: chapters/tractability.tex
\section{Tractability results for stable classes}

\input{chapters/tractability/neighborhood-complexity.tex}
\input{chapters/tractability/neighborhood-covers.tex}
\input{chapters/tractability/model-checking.tex}

%% file: chapters/tractability/neighborhood-complexity.tex
\subsection{Almost-linear neighborhood complexity}\label{sec:lin-nbd}
In this section, we prove \Cref{thm:nei-comp}.
Given a class $\CC$, 
 define the \emph{neighborhood complexity} of $\CC$ as the function $\nu_\CC\from \N\to\N$ 
such that 
\[\nu_\CC(n)\coloneqq \sup_{G\in \CC,A\subset V(G),|A|=n}|\setof{N[v]\cap A}{v\in V(G)}|.\]
Note that for every graph class $\CC$ we have $\nu_\CC(n)\le 2^n$ for all $n\in \N$.
It is an immediate consequence of the Sauer-Shelah lemma~\cite{sauer1972density,shelah1972combinatorial} that for every class of graphs of bounded VC dimension (in particular, every monadically NIP or monadically stable class) there is some constant $c$ such that $\nu_\CC(n)\le \Oh(n^c)$ for all $n\in \N$.
\cref{thm:nei-comp}
states that every monadically stable graph class $\CC$ has
almost linear neighborhood complexity, that is,
$\nu_\CC(n)\le \Oof_{\CC,\varepsilon}(n^{1+\varepsilon})$ for all $\varepsilon>0$.
This result is a generalization of an analogous result of Eickmeyer et al.~\cite{DBLP:conf/icalp/EickmeyerGKKPRS17} for nowhere dense classes, stated below.

\begin{theorem}[{\cite{DBLP:conf/icalp/EickmeyerGKKPRS17}}]\label{thm:nbd-complexity-sparse}
    Fix a  nowhere dense graph class $\CC$ and $\varepsilon>0$.
    Then for all $G\in \CC$ and $A\subset V(G)$, \[|\setof{N[v]\cap A}{v\in V(G)}|\le \Oof_{\CC,\varepsilon}(|A|^{1+\varepsilon}).\]
\end{theorem}
A similar result
holds for all \emph{structurally nowhere dense classes}, that is, classes that can be transduced from a nowhere dense class~\cite{PilipczukST18a}. Structurally nowhere dense classes are contained in monadically stable classes, and it is conjectured that the two notions coincide~\cite{POM21}. 
 However, \cref{thm:nei-comp} is  incomparable with the statement from \cite{PilipczukST18a}, as the latter also allows defining neighborhoods using a formula $\phi(\bar x,\bar y)$ involving tuples of free variables. We expand on this in \cref{sec:conclusions}.

In order to prove \cref{thm:nei-comp},
we will gradually simplify a monadically stable class --- while preserving monadic stability, and  without decreasing its neighborhood complexity too much ---
until we arrive at a $K_{t,t}$-free graph class.
The next theorem states that monadically stable, $K_{t,t}$-free classes are nowhere dense, so we will be able to conclude using \cref{thm:nbd-complexity-sparse}.
\cref{thm:dvorak} below follows from a result of Dvo\v{r}\'ak \cite{Dvorak18} (see \cite[Corollary 2.3]{DBLP:journals/corr/abs-1909-01564}).

\begin{theorem}[follows from~\cite{Dvorak18}]\label{thm:dvorak}
    Let $\CC$ be a monadically stable class of graphs,
    and suppose that $\CC$ avoids some biclique $K_{t,t}$ as a subgraph.
    Then $\CC$ is nowhere dense.
\end{theorem}

\newcommand{\cmpN}{\overline{N}}

Our simplification process will decrease a parameter we call \emph{branching index}, which was introduced by Shelah in model theory and is sometimes referred to as \emph{``Shelah's 2-rank''}. For a bipartite graph $G=(A,B,E)$ and vertex $a\in A$,
we denote $\cmpN_G(a)\coloneqq B-N_G(a)$.

\newcommand{\br}{\mathrm{br}}
\begin{definition}
    Let $G=(A,B,E)$ be a bipartite graph.
    The \emph{branching index} of a set $U\subset B$, denoted $\br_G(U)$, is defined as
    \[\br_G(U)\coloneqq
    \begin{cases}
-1&\text{if $U=\emptyset$,}\\
1+\max_{a\in A}\min(\br_G(N(a)\cap U),\br_G(\cmpN(a)\cap U))&\text{if $U\neq\emptyset$.}
    \end{cases}
    \]
\end{definition}
Note that for $U\subset B$ we have that  $\br_G(U)=0$ if and only if 
$U$ is nonempty and all vertices in $A$ have the same neighborhoods in $U$.
For higher values of the branching index, the following perspective might be helpful.  

Say that $U\subset B$ is \emph{split} into sets $P,Q$ 
by a vertex $a\in A$ 
if $P=U\cap N_G(a)$ and $Q=U\cap \cmpN_G(a)$ are both nonempty.
Say that $U$ \emph{can be split} into $P$ and $Q$ if $P,Q$ are nonempty and there is some $a\in A$  which splits $U$ into $P$ and $Q$.
Informally, the value $\br_G(U)$ tells us for how many steps
we can repeatedly split $U$ into two, four, eight, etc. sets, where in  each step, we are required to split each set produced in the previous step into two parts.

The key result concerning the branching index is the following fact due to Shelah: the branching index is universally bounded in every edge-stable graph class
\cite[Lemma 6.7.9]{hodges_1993} 
(the converse also holds). In particular,
it is bounded in all monadically stable graph classes.
\begin{theorem}[see {\cite[Lemma 6.7.9]{hodges_1993}}]\label{thm:branching-index}
    Let $\CC$ be an edge-stable class of bipartite graphs.
    Then there is a number $d\in\N$ such that $\br_G(B)\le d$ for all $G = (A,B,E)\in \CC$.
\end{theorem}

We will use the following statement about definability of the branching index, which can be easily proved by induction on $d$.
\begin{lemma}\label{lem:define-br}
    For every $d\in\N$ there is a first-order sentence $\beta_d$ over the signature consisting of a binary relation symbol $E$ and unary relation symbols $A,B$,
    such that, given a bipartite graph $G=(A,B,E)$,
    the structure $G$ (where $A,B,E$ are interpreted as the appropriate relations)
    satisfies $\beta_d$ if and only if $\br_G(B)\le d$.
\end{lemma}

The following proposition, together with \Cref{thm:nbd-complexity-sparse,thm:dvorak}, 
will later easily yield \cref{thm:nei-comp}.

\begin{restatable}{proposition}{propNbd}\label{prop:nbd-compl-main}
    Fix $d\in\N$.
    There is a transduction $\Tc_d$ with the following properties. 
    Given a bipartite graph $G=(A,B,E)$ with $n:=|A|\ge 2$ such that 
    no vertices in $B$ have equal neighborhoods
    and $\br_G(B)\le d$,
    there 
    is a bipartite graph $G'=(A',B',E')\in \Tc_d(G)$
    with $A'\subset A$ and $B'\subset B$
    such that:
\begin{enumerate}[label=(A.{\arabic*})]
    \item\label{it:size'} $|B'|\ge \frac{|B|}{(300\ln n)^d}$,
    
    \item\label{it:one-neighbor'} 
    every 
     vertex $b\in B'$ 
    has at most $d$ neighbors in $A'$ in the graph $G'$,
    
    \item\label{it:classes'} all vertices in $B'$ have distinct neighborhoods in $A'$ in the graph $G'$.
\end{enumerate}
\end{restatable}

The following technical sampling lemma will be used as an ingredient in the proof of \cref{prop:nbd-compl-main}.

\begin{lemma}\label{lem:sample}
        Let $G=(A,B,E)$ be a bipartite graph such that \(|A| \ge 2\) and every vertex in $B$ has some neighbor in $A$.
        Then there are sets $X\subset A$ and $B'\subset B$ with
        \(|B'|\ge \frac{|B|}{150\ln |A|}\)
        such that every vertex $b\in B'$ has exactly one neighbor in $X$.
\end{lemma}
\begin{proof}Fix a real $\alpha>1$ to be specified later and denote $n\coloneqq |A|$.
	Consider all intervals of the form \([\alpha^i,\alpha^{i+1})\) with \(0 \le i \le \log_{\alpha}n\).
	For each \(b \in B\), the degree of $b$ belongs to exactly one such interval.
    Therefore, there is some $i$ as above and a set $B_0\subset B$ with 
     \[B_0\ge |B|/(1+\log_{\alpha}n)\]
     such that all vertices in $B_0$ 
    have degree in the interval  \([\alpha^i,\alpha^{i+1})\). Set $d=\alpha^i$. Thus, all vertices in $B_0$ have degree between $d$ and $\alpha d$.

	Pick \(X \subseteq A\) by including each vertex of $A$ uniformly at random with probability \(1/d\).
	Consider \(b \in B_0\). Since \(d \le |N(b)| \le \alpha d\), the expected size of \(X \cap N(b)\) is between \(1\) and \(\alpha\).
	By Markov's inequality,
	\[
		\Prob\bigl[|X \cap N(b)| \ge 2\bigr] \le \Exp\bigl[|X \cap N(b)|\bigr]/2 \le \alpha/2.
	\]
	On the other hand,
	\[
		\Prob\bigl[|X \cap N(b)| = 0\bigr] \le (1-1/d)^{d} \le 1/e.
	\]
	This means that
	\[
		\Prob\bigl[|X \cap N(b)| = 1\bigr] = 1 - \Prob\bigl[|X \cap N(b)| = 0\bigr] - \Prob\bigl[|X \cap N(b)| \ge 2\bigr] \ge 1-1/e - \alpha/2\eqqcolon\beta.
    \]
    We choose $\alpha$ so that $\beta\coloneqq 1-1/e-\alpha/2\ge \frac{1}{140\ln\alpha}$, for example, $\alpha\coloneqq 1.1$.
	The expected number of vertices \(b \in B_0\) such that \(|X \cap N(b)|=1\) is thus at least \(\beta|B_0|\).
	There exists an assignment to \(X\) reaching the expected value, so let us fix $X$ according to this assignment. We then set $B'$ to be those elements of $B_0$ with exactly one neighbor in $X$.
    As \(n \ge 2\), we can bound \(\beta + 140 \ln n \le 150 \ln n\).
    Therefore, 
    \[|B'|\ge \beta \cdot |B_0|
    \ge \beta \cdot \frac{|B|}{1 + \log_\alpha n }
    \ge \frac{|B|}{\beta + 140\ln n}
    \ge \frac{|B|}{150\ln n}.\]
    This concludes the proof.
\end{proof}


The next lemma is the main engine of the proof of \cref{prop:nbd-compl-main}, and hence of \cref{thm:nei-comp}.
The central definitions of this lemma are also depicted in \Cref{fig:mainlemma}.

\begin{figure}[h]
    \centering
    \includegraphics[width=0.70 \linewidth]{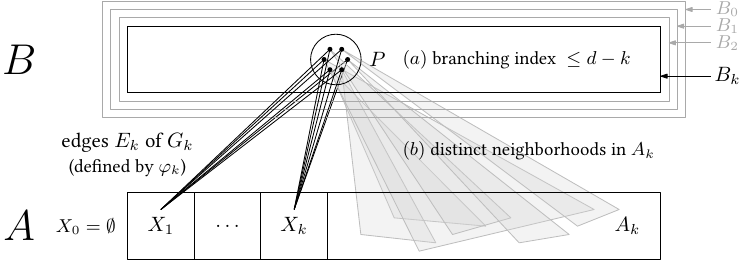}
    \caption{Illustration of central definitions in \Cref{lem:main-nbd-compl}}\label{fig:mainlemma}
\end{figure}

\begin{restatable}{lemma}{mainNbdComplexity}
    \label{lem:main-nbd-compl}
    Fix $d\in\N$.
    For every $0\le k\le d$ there is a formula 
    $\phi_k(x,y)$
    in the signature consisting of a binary relation symbol $E$ and unary relation symbols $A,B,X_0,\ldots,X_k,B_0,\ldots,B_k$, such that the following holds.
    Given a bipartite graph $G=(A,B,E)$ with $n\coloneqq |A|\ge 2$ such that no two vertices of $B$ have equal neighborhoods and $\br_G(B)\le d$,
    there exist sets $B=B_0 \supseteq \ldots \supseteq B_k$,
 pairwise disjoint sets $X_0,\ldots,X_k\subset A$,
 and a bipartite graph $G_k=(A,B,E_k)$
    such that:
\begin{enumerate}[label=(B.{\arabic*})]
    \item\label{it:size} $|B_k|\ge \frac{|B|}{(300\ln n)^k}$;
    
    \item\label{it:edges-by-formula} $E_k$ is the set of pairs $ab$ such that $a\in A,b\in B$ and $(G,A,B,X_0,\ldots,X_k,B_0,\ldots,B_k)\models\phi_k(a,b)$;
    
    \item\label{it:one-neighbor} every 
     $b\in B_k$ 
    has at most $k$ neighbors in  $X_0\cup\cdots \cup X_k$ in the graph $G_k$; and
    
    \item\label{it:classes} for every $P\subset B_k$ such that 
    all $b\in P$ have the same neighborhood in $X_0\cup \cdots \cup X_k$ in the graph $G_k$,
    \begin{enumerate}[label=(\alph*)]
        \item\label{it:classes-b} $\br_{G_k}(P)\le d-k$, and
        \item\label{it:classes-a} no two distinct vertices in $P$ have equal neighborhoods in $A-(X_0\cup \cdots\cup X_k)$ in the graph $G_k$.
    \end{enumerate}
\end{enumerate}    
\end{restatable}

\begin{proof}
    We proceed by induction on $k$.
    For $k=0$ the formula $\phi_0(x,y)\coloneqq E(x,y)$ satisfies the required conditions. Namely, for a given graph $G=(A,B,E)$ we define $B_0\coloneqq B$, $G_0\coloneqq G$, $X_0\coloneqq \emptyset$, and the required conditions trivially hold.

    Assuming the statement holds for some value $k\ge 0$, we prove it for $k+1 \le d$. Let $\phi_k(x,y)$ be as in the statement. 
    Consider a bipartite graph $G=(A,B,E)$
    and let $B=B_0\supseteq\ldots\supseteq B_k$,
    $X_0,\ldots,X_k\subset A$,
    and $G_k$ be given by the induction assumption.
    Denote $A_k\coloneqq A-(X_0\cup\cdots\cup X_k)$.

    By a \emph{$k$-class} we mean an inclusion-wise maximal set $P\subset B_k$ of vertices in $b\in B_k$ which have equal neighborhoods in $X_0\cup\cdots\cup X_k$ in the graph $G_k$.
    By assumption, $\br_{G_k}(P)\le d-k$ for every $k$-class $P$.
    Therefore, for every vertex $a\in A$ and $k$-class $P$,
    we have
    \begin{equation}\label{eq:br-index}
        \min\left(\br_{G_k}(N_{G_k}(a)\cap P),\br_{G_k}(\cmpN_{G_k}(a)\cap P)\right)< d-k.
    \end{equation}

Define a relation $E_0\subset A_k\times B_k\subset A\times B$ as
\begin{eqnarray*}
    E_0 & \coloneqq & \bigcup\,\bigl\{\set{a}\times P \bigm| a\in A_k,\text{$P$ is a $k$-class with $\br_{G_k}(N_{G_k}(a)\cap P)=d-k$}\bigr\}\\
   & = &  \bigl\{ab\bigm| a\in A_k\textrm{ and } b\in P\text{ for some $k$-class $P$ with $\br_{G_k}(N_{G_k}(a)\cap P)=d-k$}\bigl\}.
\end{eqnarray*}
Let $G_{k+1}\coloneqq (A,B,E(G_k)\symdiff E_0)$,
where $\symdiff$ denotes the symmetric difference.
So for a pair $ab$ with $a\in A$ and $b\in B$, we have that
$ab\in E(G_{k+1})$ if and only if $ab$
belongs to exactly one of the sets $E(G_k)$ and $E_0$.
The relation $E_{k+1}$ is definable by a first-order formula, as stated in the next claim.

\begin{claim}
    \label{cl:construct-formula}
    There is a first-order formula $\phi_{k+1}(x,y)$ 
    over the signature $\set{E,A,B,X_0,\ldots,X_k,B_0,\ldots,B_k}$, which is independent of $G$, such that
    for all $a\in A$ and $b\in B_k$
    we have $(G,A,B,X_0,\ldots,X_k,B_0,\ldots,B_k)\models \phi_{k+1}(a,b)$ if and only if ${ab\in E_{k+1}}$.
\end{claim}
\begin{claimproof}
    Following the definition of the relation~$E_0$, we  construct a formula $\psi(x,y)$ defining the relation~$E_0$, by combining the formula $\phi_k$ obtained by the inductive assumption, and the formula $\beta_{d-k}$ obtained in \cref{lem:define-br}. Then $\phi_{k+1}(x,y)$ is defined as the XOR of the formulas $\psi$ and $\phi_k$. We omit the easy details.
\end{claimproof}

Next, we note that the construction of $G_{k+1}$ from $G_k$ did not affect the neighborhoods in $X_0\cup \cdots \cup X_k$, nor the branching indices of subsets of $k$-classes.

\begin{claim}\label{cl:same-nbd}
    If $b\in B$, then the neighborhood 
    of $b$ in $X_0\cup\cdots\cup X_k$ is the 
    same when considered in $G_k$, and when considered in $G_{k+1}$.
\end{claim}
\begin{claimproof}  
We have $E(G_{k+1})=E(G_k)\symdiff E_0$, with $E_0\subset A_k\times B_k$ and $A_k$ disjoint from $X_0\cup\cdots\cup X_k$.
\end{claimproof}

\begin{claim}\label{cl:same-br}
    Let $P\subset B_k$ be a $k$-class and $Q\subset P$.
    Then \(\br_{G_{k+1}}(Q)= \br_{G_{k}}(Q)\). 
\end{claim}
\begin{claimproof}
    We show that
    \begin{equation}\label{eq:wydra}\set{N_{G_k}(a)\cap Q,\cmpN_{G_{k}}(a)\cap Q}=\set{N_{G_{k+1}}(a)\cap Q,\cmpN_{G_{k+1}}(a)\cap Q}\qquad \textrm{for every }a\in A.
    \end{equation}
    So, the two parts into which $a$ splits $Q$ are the same in $G_k$  as in $G_{k+1}$ (where the neighborhood  is possibly swapped with the non-neighborhood). \cref{cl:same-br} then follows from \eqref{eq:wydra} by a straightforward induction.
    
Towards~\eqref{eq:wydra}, fix $a\in A$.
If $a\in X_0\cup\cdots\cup X_k$,
the statement follows from \Cref{cl:same-nbd}.
Now, suppose $a\in A-(X_0\cup\cdots \cup X_k)=A_k$.
Then the set $\set{a}\times P$ is either disjoint from $E_0$, or is contained in  $E_0$.
In the first case, 
we have that 
$N_{G_k}(a)\cap P=N_{G_{k+1}}(a)\cap P$,
while in the latter case, we have that 
$N_{G_k}(a)\cap P=\cmpN_{G_{k+1}}(a)\cap P$. Either way,~\eqref{eq:wydra} follows.
\end{claimproof}

Let $B_-\subset B_k$ be the set of those vertices $b\in B_k$
such that $b$ has no neighbor in $A_k$ in the graph $G_{k+1}$,
and let $B_+\coloneqq B_k-B_-$.
We consider two cases, depending on which of the sets $B_+,B_-$ is larger.

\paragraph{Case 1: $\boldsymbol{|B_+|\ge |B_k|/2}$.}
Apply \cref{lem:sample} to $G_{k+1}[A_k,B_+]$, obtaining sets $B_{k+1}\subseteq B_+\subseteq B_k$ and $X_{k+1}\subseteq A_k$
such that  $|B_{k+1}|\ge \frac{|B_+|}{150\ln n}\ge \frac{|B_k|}{300\ln n}$ and every vertex in $B_k$ has exactly one neighbor in $X_{k+1}$ in $G_{k+1}$.
We check the required properties of $A_{k+1}$ and $B_{k+1}$.
We have $|B_{k+1}|\ge \frac{|B_k|}{300\ln n}\ge \frac{|B|}{(300\ln n)^{k+1}}$ by the induction assumption, so condition~\ref{it:size} holds. Condition~\ref{it:edges-by-formula} holds by \cref{cl:construct-formula}.

To verify condition \ref{it:one-neighbor},
let $b\in B_{k+1}$. 
We show that $b$ has at most $k+1$ neighbors in  $X_0\cup \ldots,X_k\cup X_{k+1}$ in the graph $G_{k+1}$.
By \cref{cl:same-nbd}, the adjacency between $b$ and $X_0,\ldots,X_k$
in the graph $G_{k+1}$ is the same as in the graph $G_k$.
Therefore,  $b$ has at most $k$ neighbors in $X_0\cup \ldots\cup X_k$ in the graph $G_{k+1}$,
as it does so in the graph $G_k$ by assumption.
Furthermore, $b$ has exactly one neighbor in $X_{k+1}$ by construction. This verifies condition~\ref{it:one-neighbor}.

Finally, we verify condition~\ref{it:classes}.
Let $P'\subset B_{k+1}$ be a $(k+1)$-class.
By \Cref{cl:same-nbd} and since every vertex in \(B_{k+1}\) has exactly one neighbor in \(X_{k+1}\),
we can write $P'=N_{G_{k+1}}(a_0)\cap P\cap B_{k+1}$ for some $k$-class $P\subset B_k$ and some $a_0\in X_{k+1}$.
We need to show that (a) $\br_{G_{k+1}}(P')\leq d-k-1$, and that (b) $P'$ does not contain distinct vertices with equal neighborhoods in $A_{k+1}$ in the graph $G_{k+1}$.

We first verify (a), that is, $\br_{G_{k+1}}(P')< d-k$.
For all $b\in P$ we have that 
$(a_0,b)\in E_0$ if and only if $\br_{G_k}(N_{G_k}(a_0)\cap P)=d-k$. 
Suppose first that 
\begin{equation}\label{eq:first-case}
    \br_{G_{k}}(N_{G_k}(a_0)\cap P)<d-k.    
\end{equation}
Then $(a_0,b)\notin E_0$ for all $b\in P$.
 As $E_0=E(G_k)\symdiff E(G_{k+1})$, it follows that the neighborhood of $a_0$ in $P$ is the same
when evaluated in $G_k$ and when evaluated in $G_{k+1}$.
Therefore, \[N_{G_k}(a_0)\cap P\cap B_{k+1}=N_{G_{k+1}}(a_0)\cap P\cap B_{k+1}=P'.\]
In particular $P'\subset N_{G_k}(a_0)\cap P$,
so $\br_{G_{k}}(P')<d-k$ by \eqref{eq:first-case} and the monotonicity of the branching index. Then also $\br_{G_{k+1}}(P')<d-k$ by \cref{cl:same-br}.
This confirms (a) in the considered case.
Now suppose that 
\begin{equation}\label{eq:second-case}
    \br_{G_k}(N_{G_k}(a_0)\cap P)\ge d-k.
\end{equation}
Then $\br_{G_k}(\cmpN_{G_k}(a_0)\cap P)<d-k$
holds by \eqref{eq:br-index}.
Dually to the previous case, we have that $(a_0,b)\in E_0$
for all $b\in P$.
By a reasoning dual to the one above, 
\[\cmpN_{G_k}(a_0)\cap P\cap B_{k+1}=N_{G_{k+1}}(a_0)\cap P\cap B_{k+1}=P'.\]
Again, we conclude that $\br_{G_{k+1}}(P')<d-k$, confirming (a).

We now verify (b), that is, we show that $P'$ does not contain any pair of  distinct vertices with equal neighborhoods in $A_{k+1}$
in the graph $G_{k+1}$.
Let $b,b'\in P'$ be distinct. Then $b,b'\in P$, so by assumption, $b$ and $b'$ have distinct neighborhoods in $A_{k}$ in the graph $G_k$. 
Let $a\in A_k$ be such that 
\[(a,b)\in E(G_k)\Leftrightarrow (a,b')\notin E(G_k).\]
Since $b,b'\in P$ it follows that \[(a,b)\in E_0\Leftrightarrow (a,b')\in E_0.\] As $E(G_{k+1})=E(G_k)\symdiff E_0$,
it follows that $(a,b)\in E(G_{k+1}) \Leftrightarrow (a,b')\notin E(G_{k+1})$.
Since \(a_0\) is the unique neighbor of both \(b\) and \(b'\) in \(X_{k+1}\) in the graph \(G_{k+1}\),
and \(a\) is adjacent in \(G_{k+1}\) either to \(b\) or \(b'\), we conclude \(a \not \in X_{k+1}\).
Therefore, $a$ witnesses that $b$ and $b'$ have distinct neighborhoods in $A_{k+1} = A_k \setminus X_{k+1}$ in the graph $G_{k+1}$.

Thus, we verified condition~\ref{it:classes}, and completed Case 1.

\paragraph{Case 2: $\boldsymbol{|B_-|>|B_{k}|/2}$.}
Let $B_{k+1}\coloneqq B_-\subseteq B_k$ consist of the vertices with no neighbors in $A_k$ in the graph \(G_{k+1}\).
Set $X_{k+1}\coloneqq\emptyset$.
We verify the required conditions for this choice of $B_{k+1}$ and $X_{k+1}$.

Condition~\ref{it:size} holds as
$|B_{k+1}|\ge |B_k|/2\ge \frac {|B|}{2\cdot (300\ln n)^k}\ge \frac{|B|}{(300\ln n)^{k+1}}$,
condition~\ref{it:edges-by-formula} holds by \cref{cl:construct-formula}, and
condition~\ref{it:one-neighbor} holds by \cref{cl:same-nbd}.
To prove condition \ref{it:classes}, we show the following
\begin{claim}\label{cl:singleton-classes}
No two distinct vertices of $B_{k+1}$ have
equal neighborhoods in $X_0\cup\cdots\cup X_{k}$ in $G_{k+1}$.
\end{claim}
\begin{claimproof}
    Suppose that distinct $b,b'\in B_{k+1}$ have equal neighborhoods in $X_0\cup\cdots\cup X_{k}$ in $G_{k+1}$.
    Using \Cref{cl:same-nbd},
    let $P$ be the $k$-class such that $b,b'\in P$.
    By condition \ref{it:classes},\ref{it:classes-a}, we have that $b$ and $b'$ have different neighborhoods
    in $A_k$ in $G_k$.
    Let $a\in A_k$ be adjacent to exactly one of 
    $b,b'$ in $G_{k}$.
    As $b$ and $b'$ belong to the same $k$-class $P$, it follows by definition of \(E_0\) that $ab\in E_0\Leftrightarrow ab'\in E_0$.
    As $E(G_k)=E(G_{k+1})\symdiff E_0$,
    we observe that $a$ is adjacent to exactly one of $b,b'$ in $G_{k+1}$.
    Moreover, as $b,b'\in B_-$, and vertices in $B_-$ have no neighbors in $A_k$ in the graph $G_k$, it follows that $a\notin A_k$, so $a\in X_0\cup\cdots\cup X_{k}$.
    Therefore, $b$ and $b'$ have distinct neighborhoods in $X_0\cup\cdots\cup X_{k}$ in $G_{k+1}$, a contradiction which completes the proof.
\end{claimproof}

As $X_0\cup \cdots\cup X_k=X_0\cup \cdots \cup X_{k+1}$,
condition \ref{it:classes} follows trivially from \cref{cl:singleton-classes}.
This completes Case 2, and the proof of the lemma.
\end{proof}

Next, we proceed to
\cref{prop:nbd-compl-main}, which is obtained by setting $k=d$ in \cref{lem:main-nbd-compl} and studying the consequences of condition \ref{it:classes} in this case.
We repeat the statement.
\propNbd*
\begin{proof}
    Apply \cref{lem:main-nbd-compl} to $k=d$,
    obtaining a formula $\phi_d(x,y)$ 
    involving the edge relation $E$ and some unary predicates. 
    Let \(\Tc_d\) be the transduction that first assigns these unary predicates, then
    applies \(\phi_d(x,y)\), and finally takes an arbitrary subgraph.
    Given a bipartite graph $G=(A,B,E)$,
    let $B_0,\ldots,B_d,X_0,\ldots,X_d$ and $G_d=(A,B,E')$ be as in the statement of \cref{lem:main-nbd-compl}. Set
    $A'\coloneqq X_0\cup\cdots \cup X_d$, $B'\coloneqq B_d$.
    Let $G'\coloneqq G_d[A',B']$ be the bipartite graph induced on $A'$ and $B'$.
    Then $G'\in \Tc_d(G)$. 

    Conditions \ref{it:size'} and \ref{it:one-neighbor'} follow immediately from
    the appropriate conditions in \cref{lem:main-nbd-compl}.
    We verify  condition \ref{it:classes'}.
    Suppose $b,b'\in B'$ have equal neighborhoods in $A'$ in the graph $G'$.
    By condition~\ref{it:classes},\ref{it:classes-b} in \cref{lem:main-nbd-compl} applied to $P=\set{b,b'}$, we have that
    $\br_{G_d}(P)=0$. Hence, $b$ and $b'$ have equal neighborhoods in $G_d$.
    By  condition \ref{it:classes},\ref{it:classes-a}, we get $b=b'$.
    This completes the proof of \cref{prop:nbd-compl-main}.
\end{proof}

The next lemma combines \cref{prop:nbd-compl-main} with \cref{thm:nbd-complexity-sparse} and \cref{thm:dvorak}.

\begin{lemma}\label{lem:nbd-compl}
    Let $\CC$ be a monadically stable class
    of bipartite graphs such that 
    for all $G=(A,B,E)\in \CC$ no vertices in $B$ have equal neighborhoods and let \(\epsilon > 0\).
    Then for every $(A,B,E)\in \CC$ we have that
    $|B|\le \Oof_{\CC,d,\varepsilon}(|A|^{1+\epsilon})$.
\end{lemma}
\begin{proof}
    Let $d$ be as in \cref{thm:branching-index}, so that $\br_G(B)\le d$ for all $G = (A,B,E)\in \Cc$.
    Let $\Tc_d$ be as in \cref{prop:nbd-compl-main}.
    Without loss of generality, we may assume \(|A| \ge 2\) for all $(A,B,E)\in \CC$.
    We associate with every \(G \in \Cc\) a bipartite graph $F(G) \in \Tc_d(G)$ 
    satisfying the conditions listed in \cref{prop:nbd-compl-main}. 
    Let $\DD=\setof{F(G)}{G\in\CC}$.
    Then $\DD$ is monadically stable, as $\DD\subset \Tc_d(\CC)$ and $\CC$ is monadically stable.
    Moreover, the class $\DD$ avoids $K_{d+1,d+1}$ as a subgraph, by condition~\ref{it:one-neighbor'}.
    Therefore, \Cref{thm:dvorak} implies that $\DD$ is nowhere dense.
    By \ref{it:classes'}, for every graph $(A',B',E')\in\DD$, there is no pair of vertices in $B'$ with equal neighborhoods in $A'$.
    Consider \(G = (A,B,E) \in \Cc\) and \(F(G) = (A',B',E') \in \Dd\).
    By \cref{thm:nbd-complexity-sparse} we have for $\delta \coloneqq \epsilon/2$ that
    \[|B'|\le \Oof_{\DD,\delta}\left(|A'|^{1+\delta}\right).\]
    On the other hand, by condition \ref{it:size'} we infer that \[|B'|\ge |B|/(300\ln |A|)^d.\]
    As $|A'|\le |A|$, we obtain
    \[|B|\le |B'| \cdot (300 \ln |A|)^d\le  (300\ln |A|)^d\cdot \Oof_{\DD,\delta}\left(|A|^{1+\delta}\right)\le \Oof_{\CC,d,\epsilon}\left(|A|^{1+\epsilon}\right).\qedhere\]
\end{proof}

\newcommand{\Bb}{\mathscr{B}}

\Cref{thm:nei-comp}, restated for convenience, now follows easily.

\introneicomp*
\begin{proof}
    For a graph $G=(V,E)$,
    let $B(G)$ denote the bipartite graph 
    with sides $V\times \set{1}$ and $V\times \set{2}$,
and edges $(u,1),(v,2)$ for all $u,v$ such that $u=v$ or $uv\in E$.
Let $\Bb$ be the class of all induced subgraphs of graphs $B(G)$, for $G\in \CC$.
Then $\Bb$ monadically stable. One way to argue this is that $\Bb$ is obtained from $\CC$ by a transduction with copying, and it is known that such transductions preserve monadic stability.
It is alternatively not difficult to prove directly that if $\DD$ transduces the class of all half-graphs, then so does~$\CC$.
By \cref{lem:nbd-compl},
for all $(A,B,E)\in \Bb$ such that no two vertices in $B$ have equal neighborhoods,
\[|B|\le \Oof_{\CC,\varepsilon}\left(|A|^{1+\varepsilon}\right).\]

Let $G\in\CC$ and $A\subset V(G)$.
Choose $B\subset V$ 
such that $N[v]\cap A\neq N[v'] \cap A$ for all $v,v'\in B$.
Then $|\setof{N[v]\cap A}{v\in V(G)}|=|B|$.
As $(A,B,\setof{ab}{a\in A,b\in B, a \in N[b]})\in \Bb$,
and no two vertices in \(B\) have equal neighborhoods in this graph,
we conclude that 
\[|\setof{N[v]\cap A }{v\in V(G)}|=|B|\le \Oof_{\CC,\varepsilon}\left(|A|^{1+\varepsilon}\right).\qedhere\]
\end{proof}

%% file: chapters/tractability/neighborhood-covers.tex
\newcommand{\Ss}{\mathcal{S}}

\newcommand{\diam}{\mathrm{diam}}
\newcommand{\overlap}{\mathrm{overlap}}

\subsection{Sparse neighborhood covers}\label{sec:neicov}

For a graph $G$ and vertex subset $X$, the {\em{weak diameter}} of $X$ in $G$ is the maximum distance in $G$ between members of $X$: $\diam_G(X)\coloneqq \max_{u,v\in X} \dist_G(u,v)$.
As discussed in \cref{sec:intro}, the following notion of a {\em{neighborhood cover}} is of key importance to us.

\begin{definition}
 Let $G$ be a graph and $r$ be a positive integer. A family $\Cover$ of subsets of vertices of $G$ is called a {\em{distance-$r$ neighborhood cover}} of $G$ if for every vertex $u$ of $G$ there exists $C\in \Cover$ such that $\Ball_r[u]\subseteq C$.
 The {\em{diameter}} of $\Cover$ is the maximum weak diameter among the sets of $\Cover$, while the {\em{overlap}} of $\Cover$ is the maximum number of sets of $\Cover$ that intersect at a single vertex:
 \[\diam(\Cover)\coloneqq \max_{C\in \Cover}\ \diam_G(C)\qquad\textrm{and}\qquad \overlap(\Cover)\coloneqq \max_{u\in V(G)} |\{C\in \Cover~|~u\in C\}|.\]
 Elements of a neighborhood cover $\Cover$ will often be called {\em{clusters}}.
\end{definition}

The results of Dreier, Mählmann, and Siebertz~\cite{ssmc} require the construction of distance-$r$ neighborhood covers with small overlap, for all $r\in \N$. We now explain how such covers can be obtained by finding distance-$1$ covers in interpretations of the given graph.

For a graph $G$, let the {\em{$r$th power}} of $G$, denoted $G^r$, be the graph on the same vertex set as $G$ where vertices $u,v$ are adjacent if and only if the distance between $u$ and $v$ in $G$ is at most $r$. Note that for every fixed $r$, $G^r$ can be easily interpreted in $G$ using a formula that checks whether the distance between $u$ and $v$ is at most $r$. The next lemma shows that finding a distance-$1$ neighborhood cover in $G^r$ immediately yields a distance-$r$ neighborhood cover in $G$ with the same overlap.

\begin{lemma}\label{lem:distanceCover}
 Let $G$ be a graph, $r$ be a positive integer, and $\Cover$ be a distance-$1$ neighborhood cover of $G^r$ of diameter $d$. Then $\Cover$ is also a distance-$r$ neighborhood cover of $G$ of diameter at most $dr$.
\end{lemma}
\begin{proof}
 That $\Cover$ is a distance-$r$ neighborhood cover of $G$ follows immediately from the observation that for every vertex $u$, $\Ball^{G^r}_1[u]=\Ball^G_r[u]$. That the weak diameter of $\Cover$ is at most $dr$ follows immediately from triangle inequality and the definition of the graph $G^r$.
\end{proof}

In this section we provide a construction of neighborhood covers with small overlap for monadically stable graph classes. Formally, we prove the following result.

\begin{theorem}\label{thm:cover-main}
 Let $\Cc$ be a monadically stable class of graphs. Then for every $r \in \N$ and \(n\)-vertex graph $G\in \Cc$ there exists a distance-$r$ neighborhood cover of $G$ with diameter at most $4r$ and whose overlap can be bounded by $\Oh_{\Cc,r,\eps}(n^\eps)$ for every $\eps>0$. Moreover, there is an algorithm that computes such a neighborhood cover in time complexity that can be bounded by $\Oh_{\Cc,r,\eps}(n^{4+\eps})$, for every $\eps>0$.
\end{theorem}

We note that the algorithm of \cref{thm:cover-main} actually does not depend on the class $\Cc$ or the value of $\eps$: it is a single algorithm that, when supplied with a graph $G\in \Cc$, will always output a neighborhood cover of $G$ with diameter and overlap bounded as asserted.

The main ingredient towards proving \cref{thm:cover-main} will be a tool introduced by Welzl~\cite{Welzl89} in the context of geometric range queries, called 
\emph{spanning paths with low crossing number}, which we will call \emph{Welzl orders}. To introduce them, we need some definitions. We remark that for convenience, our terminology slightly differs from that of Welzl.

Consider a set system $\Ss=(U,\Ff)$, where $U$ is a finite universe and $\Ff$ is a family of subsets of~$\Ff$. The {\em{growth function}} of $\Ss$ is the function
$\pi_\Ss(\cdot)$ 
that assigns each positive integer \(n\) the value
\[\pi_\Ss(n) \coloneqq \max_{A\subseteq U, |A|\leq n}\ |\{X\cap A\colon X\in \Ff\}|.\]
In other words, $\pi_\Ss(n)$ is the largest number of {\em{traces}} that the sets from $\Ff$ leave on a subset $A\subseteq U$ of size~$n$, where the trace left by $X\in \Ff$ on $A$ is $X\cap A$. For example, the Sauer-Shelah Lemma states that if the VC dimension of $\Ss$ is $d$, then $\pi_\Ss(n)\le \Oh(n^d)$. On the other hand, from Theorem~\ref{thm:nei-comp} we immediately obtain the following.

\begin{corollary}\label{cor:growth}
 Let $\Cc$ be a monadically stable class of graphs and $\Dd$ be the class of set systems of closed neighborhoods of graphs in $\Cc$, that is,
 \[\Dd\coloneqq \bigl\{\bigl(V(G),\{N_G[u]\colon u\in V(G)\}\bigr)\colon G\in \Cc\bigr\}.\]
 Then for every $\Ss\in \Dd$, we have $\pi_\Ss(n)\le \Oh_{\Cc,\eps}(n^{1+\eps})$ for all $\eps>0$.
\end{corollary}

Next, for a set system $\Ss=(U,\Ff)$ and a total order $\preceq$ on $U$, we define the {\em{crossing number}} of $\preceq$ as follows. For $X\in \Ff$, the {\em{crossing number}} of $X$ with respect to $\preceq$ is the number of pairs $(u,u')$ of elements of $U$ such that
\begin{itemize}
 \item $u'$ is the immediate successor of $u$ in $\preceq$, and
 \item exactly one of $u$ and $u'$ belongs to $X$.
\end{itemize}
Note that this is equivalent to the following: the crossing number of $X$ is the least $k$ such that $\preceq$ can be partitioned into $k+1$ intervals so that every interval is either contained in or disjoint from $X$. Then the crossing number of $\preceq$ is the maximum crossing number of any $X\in \Ff$ with respect to $\preceq$.

The following statement is the main result of~\cite{Welzl89}.

\begin{theorem}[Theorem 4.2 and Lemma 3.3 of~\cite{Welzl89}, see also Theorem 4.3 of~\cite{ChazelleW89}]\label{thm:welzl-order}
 Suppose $\Ss=(U,\Ff)$ is a set system with $\pi_\Ss(n)\le \Oh(n^d)$, where $d>1$ is a real. Then there exists a total order $\preceq$ on $U$ with crossing number bounded by $\Oh(|U|^{1-1/d}\cdot \log |U|)$. Moreover, there is an algorithm that, given~$\Ss$, computes such an order in time $\Oh(\pi_\Ss(|U|)\cdot (|U|+|\Ff|)^3)$.
\end{theorem}

For the algorithmic statement, see the remark provided below the proof of \cite[Theorem~4.2]{Welzl89}. Also, we note that the algorithm of \cref{thm:welzl-order} does not need to be supplied with the value of $d$: it is a single algorithm that, given $\Ss$, computes a total order $\preceq$, and the guarantee on the crossing number of $\preceq$ follows from the assumption on the growth function of $\Ss$.

Next, we show that, given a total order with a low crossing number, we can construct a neighborhood cover with a small overlap and constant diameter using a relatively easy greedy construction.

\newcommand{\Ii}{\mathcal{I}}

\begin{lemma}\label{lem:rose}
 Suppose $G=(V,E)$ is a graph, $\Ss\coloneqq (V,\{N[u]\colon u\in V\})$ is the set system of closed neighborhoods in $G$, and $\preceq$ is a total order on $V$ with crossing number $k$ (with respect to $\Ss$). Then $G$ admits a distance-$1$ neighborhood cover with diameter at most $4$ and overlap at most $k+1$, and such a neighborhood cover can be computed, given $G$ and $\preceq$, in time $\Oh(|V|^3)$.
\end{lemma}

\begin{proof}
 We need some auxiliary definitions about the order $\preceq$.
 An {\em{interval}} is a set $I\subseteq V$ that is convex in~$\preceq$: $u\preceq v\preceq w$ and $u,w\in I$ entails $v\in I$. A {\em{prefix}} of an interval $I$ is an interval $J$ such that $u,v\in I$, $u \preceq v$ and $v\in J$ entails $u\in J$.
 An interval $I$ is {\em{compact}} if $I\subseteq N[u]$ for some  $u\in V$.
 We perform the following greedy construction of a partition $\Ii$ of $V$ into intervals:
 \begin{itemize}
  \item Start with $\Ii\coloneqq \emptyset$.
  \item As long as $V\setminus \bigcup \Ii\neq \emptyset$, let $I$ be the largest prefix of $V\setminus \bigcup \Ii$ that is compact. Then add $I$ to $\Ii$.
 \end{itemize}
 Thus, $\Ii$ consists of compact nonempty intervals. A straightforward implementation of the procedure presented above computes $\Ii$ in time $\Oh(|V|^3)$.

 We claim that
 $\Cover\coloneqq \{N[I]\colon I\in \Ii\}$
 is a neighborhood cover of $G$ of diameter at most $4$ and overlap at most $k+1$. That $\Cover$ is a neighborhood cover is clear: if $u$ is a vertex and $I\in \Ii$ is such that $u\in I$, then $N[u]\subseteq N[I]$. Also observe that the compactness of every $I\in \Ii$ implies that $N[I]$ has weak diameter at most $4$. We are left with proving the claimed bound on the overlap.

 Fix any vertex $v\in V$. Call a vertex $u\in V$ a {\em{crossing}} for $v$ if $u$ has a successor $u'$ in $\preceq$ and $N[v]$ contains exactly one of the vertices $u$ and $u'$. Since $\Ss$ has crossing number $k$, there are at most $k$ distinct crossings for $v$. We claim the following:
 
 \begin{center}
for every interval $I\in \Ii$ such that $v\in N[I]$,\\ $I$ contains a crossing for $v$ or $I$ contains the $\preceq$-largest element of $V$.
 \end{center}
 Since there can be at most $k$ crossings for $v$ and at most one interval can contain the $\preceq$-largest element,
 this claim will conclude the proof: it implies that $v$ belongs to at most $k+1$ clusters of~$\Cover$.

 To show the claim, first note that $v\in N[I]$ implies that $I\cap N[v]\neq \emptyset$. If also $I\setminus N[v]\neq \emptyset$ then clearly $I$ contains a crossing for $v$, so assume otherwise: $I\subseteq N[v]$. Let $u$ be the largest element of $I$ in the $\preceq$ order. Unless $u$ is actually the $\preceq$-largest element of $V$, $u$ has a successor $u'$ and $u'\in I'\in \Ii$ for some $I'\neq I$. Also, unless $u$ itself is a crossing for $v$, we have $u'\in N[v]$.
 We now observe that $I\cup \{u'\}$ is an interval that is compact, as witnessed by the vertex $v$, and is strictly larger than $I$. This contradicts the construction of $\Ii$: in the round when $I$ was added to $\Ii$, we could have added the larger interval $I\cup \{u'\}$ instead. This concludes the proof of the claim and of the~lemma.
\end{proof}

We may now combine all the gathered tools and prove \cref{thm:cover-main}.

\begin{proof}[Proof of \cref{thm:cover-main}]
 Let $G \in \Cc$ be the input graph.
 By \cref{lem:distanceCover}, it suffices to compute a distance-$1$ neighborhood cover with diameter $4$ for the $r$th power $G^r =(V,E) \in \CC^r$ of $G$, where $\CC^r$ is the class that contains the $r$th power of every graph in $\CC$. As $\CC$ interprets $\CC^r$, the latter class is still monadically stable.
 Let $\Ss\coloneqq (V,\{N[u]\colon u\in V\})$ be the set system of closed neighborhoods in $G^r$. By \cref{cor:growth}, we have $\pi_\Ss(n)\le \Oh_{\Cc,r,\eps}(n^{1+\eps})$ for every $\eps>0$.
 Apply the algorithm of \cref{thm:welzl-order} to $\Ss$, to obtain a total order $\preceq$ on $V$ such that the crossing number of $\preceq$ is bounded by $\Oh_{\Cc,r,\eps}(|V|^{1-\frac{1}{1+\eps}}\cdot \log |V|)\leq \Oh_{\Cc,r,\eps}(|V|^\eps)$ for every $\eps>0$. By \cref{thm:welzl-order}, this application takes time $\Oh_{\Cc,r,\eps}(|V|^{4+\eps})$. It now suffices to apply the algorithm of \cref{lem:rose} to $G^r$ and~$\preceq$.
\end{proof}

We remark that the construction presented above actually provides a neighborhood cover with far stronger structural properties than just a subpolynomial bound on the overlap. Namely, if we define the {\em{incidence graph}} of a neighborhood cover to be the bipartite graph with vertices on one side, clusters on the other side, and adjacency defined through membership (of a vertex to a cluster), then the incidence graphs of the constructed neighborhood covers form a class that is {\em{almost nowhere dense}}, in the sense that edge density in shallow minors is bounded subpolynomially in the vertex count. We view this as a first step towards a proof of a relaxed variant of the so-called {\em{Sparsification Conjecture}}~\cite{POM21}, which postulates that every monadically stable class of graphs can be transduced from a nowhere dense class. The formal statement of this result together with the relevant discussion can be found in \cref{app:almost-nd-covers}.

%% file: chapters/tractability/model-checking.tex
\subsection{Model checking}

The full version of \cite{ssmc} defines the notion of
\emph{flip-closed sparse neighborhood covers} (\cite[Definition 3]{ssmc}),
and shows that one can solve the model checking problem in time \(f(\|\phi\|)\cdot n^{11}\)
on every monadically stable graph class $\C$ admitting such neighborhood covers (\cite[Theorem 5]{ssmc}).
Here, \emph{flip-closed} means that for every fixed $k\in \N$, the class $\C_k$ of all $k$-flips of graphs from $\C$ still admits sparse neighborhood covers (more precisely, distance-$r$ neighborhood covers of diameter $\Oh_{\C,r}(1)$ and overlap $\Oh_{\C,r,\eps}(n^\eps)$, for any $\eps>0$ and $r\in \N$).
Since $\C_k$ is  monadically stable provided $\C$ is, our \Cref{thm:cover-main} implies in particular that \emph{every} monadically stable class
admits flip-closed sparse neighborhood covers, and thus immediately yields, in combination with~\cite{ssmc}, 
an fpt model checking algorithm for monadically stable classes.

However, in~\cite[Theorem 10]{ssmc}, neighborhood covers were approximated using an LP-based rounding technique in time \(\Oh(n^{9.8})\),
while our algorithm of~\cref{thm:cover-main} runs in time \(\Oh_{\Cc,r,\epsilon}(n^{4+\epsilon})\).
In \Cref{apx:mc}, we argue that substituting the neighborhood cover subroutine from~\cite{ssmc} with ours yields the faster running time promised in \cref{thm:intro-tractability}, restated below.

\introtractability*

%% file: chapters/hardness.tex
\section{Hardness results for unstable classes}\label{sec:hardness}

In this section we characterize monadically stable graph classes in terms of forbidden induced subgraphs and subsequently derive algorithmic hardness results. 
Before diving into the details, we will first give an overview of the proof strategy.
The starting point to our characterization is the work of Gajarsk\'y et al.~\cite{flippergame} which gives a particular yet irregular obstructions to monadic stability (\Cref{fig:obstruction}).
We regularize these obstructions via different Ramsey-type theorems. Intuitively, we argue that within a large enough irregular obstruction, we may find a smaller obstruction which is regular, in the sense that adjacency within it depends on certain appropriately chosen color classes. Moreover, edge-stability ensures that the relationship between these color classes is ``orderless'', thus further simplifying the structure of these obstructions and obtaining the concrete patterns of \Cref{thm:intro-patterns} below.

\intropatterns*

The contribution of this characterization is twofold: not only it gives rise to a more natural definition of monadic stability akin to the usual definition of nowhere density, but moreover, the obtained patterns that avoid half-graphs may be shown to effectively interpret the class of all graphs via existential formulas; this essentially yields \Cref{thm:intro-hardness}.

\introhardness*

We briefly discuss our approach to the above. By \Cref{thm:intro-patterns} and hereditariness, it suffices to consider two cases: $\C$ either contains a $k$-flip of every $r$-subdivided clique or a $k$-flip of every $r$-web, for some fixed $k,r \geq 2$. 
The two cases differ in their technical details; nonetheless, the general proof strategy is the same for both, and so we only highlight the case of subdivided cliques.
Up to increasing the value of $r$ by a constant factor, we may ensure by hereditariness that 
$\C$ also contains a $k$-flip of every $r$-subdivided biclique: an $r$-subdivision of a complete bipartite graph.
While a simple graph gadget argument reveals that every class 
that contains all \emph{non-flipped} $r$-subdivided bicliques existentially interprets the class of all graphs, we only know that $\CC$ contains $k$-flip of those.
The main challenge therefore lies in definably ``reversing'' the flips. 
Using Ramsey's theorem for bipartite graphs, the pigeonhole principle, and hereditariness, we conclude that $\CC$ contains for each $r$-subdivided biclique a $k$-flip that is in a certain sense \emph{canonical}.
Moreover, for each canonical flip, the flip partition $\PP$ is minimal in the sense that no two parts can be merged.
Given a canonical $k$-flip of an $r$-subdivided biclique, the key to reversing the flip is to recover for every vertex $v$ its part $\PP(v)$ in $\PP$. From here, a rather naive strategy proceeds by existentially quantifying a representative vertex $z_i$ for each part $P_i \in \PP$, checking the adjacency between $v$ and a realization to $z_i$ to see whether $\PP(v)$ was flipped with $P_i$, and finally determining $\PP(v)$ using minimality of the flip partition. 
However, this approach runs into problems as the realization of each quantified vertex $z_i$ does not necessarily lie within $P_i$. To overcome this, we instead quantify a small set $Z_i$ of representatives for each part, and check whether $v$ is connected to the majority of $Z_i$.
Using the structural properties of subdivided bicliques, 
we then argue that this allows us to approximately recover $\PP(v)$, up to parts that behave symmetrically. This is sufficient to definably tell which pairs of vertices were flipped.

This provides us with a way to recover the unflipped edge relation of a subdivided clique within its $k$-flip. Building up from the edge formula, we obtain for every existential formula $\phi$ an existential formula $\flipp(\phi)$ such that $\phi$ is true in a subdivided clique if and only if $\flipp(\phi)$ is true in its $k$-flip present in $\C$. 



\input{chapters/hardness/ramsey.tex}

\input{chapters/hardness/patterns.tex}
\input{chapters/hardness/model-checking-hardness.tex}

%% file: chapters/hardness/ramsey.tex
\subsection{Auxiliary Ramsey-type results}

Since variants of Ramsey's Theorem are at the core of our proofs, we begin by providing a brief overview of these. Here are the two standard formulations: the classic one and the one tailored to bipartite graphs. By the biclique of {\em{order}} $n$ we mean $K_{n,n}$.

\begin{lemma}[Ramsey's Theorem for Cliques]\label{lem:ramsey}
    There exists a computable function $\Ramsey\colon \N \times \N \to \N$ such that for every $k,n \in \N$ and every edge-coloring of the clique of size $\Ramsey(k,n)$ with $k$ colors we may find a sub-clique of size $n$ which is monochromatic with respect to this coloring.
\end{lemma}

\begin{lemma}[Ramsey's Theorem for Bicliques]\label{lem:bipramsey}
    There exists a computable function $\RamseyBip\colon \N \times \N \to \N$ such that for every $k,n \in \N$ and every edge-coloring of the biclique of order $\RamseyBip(k,n)$ with $k$ colors, we may find a sub-biclique of order $n$ which is monochromatic with respect to this coloring.
\end{lemma}

We will also use a generalization of Ramsey's Theorem where instead of $2$-element subsets, one colors (ordered) $\ell$-tuples of elements. To introduce this variant, we need some notation.

For a pair $(a,b)$ of elements of a linearly ordered set $(A,\le)$,
let $\otp(a,b)\in\set{<,=,>}$ indicate whether $a<b$, $a=b$, or $a>b$ holds.
For $\ell\ge 1$ and an $\ell$-tuple of elements $a_1,\ldots,a_\ell$ of a linearly ordered set $(A,\le)$,
define the \emph{order type} of $(a_1,\ldots,a_\ell)$, denoted $\otp(a_1,\ldots,a_\ell)$, as the tuple $(\otp(a_i,a_j))_{1\le i<j\le \ell}$.

With this notation in place, we can state the general form of Ramsey's Theorem that we will use. It follows easily from the standard formulation of Ramsey's Theorem for colorings of $\ell$-element subsets of an $n$-element set using $k$ colors (generalizing the case $\ell=2$ corresponding to colorings of cliques considered above).

%

\begin{lemma}[General Ramsey's Theorem]
    \label{lem:reramsey}
    Fix $k,\ell,n\in \N$. Then there is some $N\in\N$
    such that for every coloring
    \[\lambda\from [N]^\ell\to [k]\] there is
    a subset $I\subseteq [N]$
    such that $|I|=n$ and for all $(a_1,\ldots,a_\ell)\in I^\ell$,
    we have that $\lambda(a_1,\ldots,a_\ell)$ depends only on
    $\otp(a_1,\ldots,a_\ell)$.
    That is, there is a function $f$
    such that $\lambda(a_1,\ldots,a_\ell)=f(\otp(a_1,\ldots,a_\ell))$,
    for all $(a_1,\ldots,a_\ell)\in I^\ell$.
\end{lemma}


From \cref{lem:reramsey} we derive a variant of Ramsey's Theorem for grids. For this, we again need some terminology.

Fix $n \in \N$. By an {\em{ordered $n$-grid}} we mean the relational structure $\mathbf G_n$ with domain $[n] \times [n]$ and two relations $\leq_1, \leq_2$ satisfying:
\[ (i,j)\leq_1 (i',j') \iff i \leq i';\]
\[ (i,j)\leq_2 (i',j') \iff j \leq j'.\]
For compatibility with the graph notation, the domain of $\mathbf G_n$ will be denoted by $V(\mathbf G_n)$. For $(a,b)=((a_1,a_2),(b_1,b_2))\in V(\mathbf G_n)^2$, the {\em{atomic type}} of $(a,b)$ is defined as
\[ \atp(a,b)\coloneqq (\otp(a_1,b_1),\otp(a_2,b_2)).\]
In other words, $\atp(a,b)\in \{<,=,>\}^2$ encodes the relative order of $a$ and $b$ in $\leq_1$ and $\leq_2$.

A \emph{pair coloring} of $\mathbf G_n$ is a map $\chi\colon V(\mathbf G_n)^2 \to [k]$, where $k \in \N$. We say that a pair coloring $\chi$ is \emph{homogeneous} if for all $(a,b),(c,d) \in V(\mathbf G_n)^2$, we have
\[ \atp(a,b)=\atp(c,d) \implies \chi(a,b)=\chi(c,d).\]
In other words, the color of a pair $(a,b)$ depends only
on $\atp(a,b)$.


\begin{lemma}[Grid Ramsey Theorem]\label{gridramsey}
    There is a computable function $\RamseyGrid\colon\N \times \N \to \N$ such that for every $k,n \in \N$ and every pair coloring of the ordered $\RamseyGrid(k,n)$-grid with $k$ colors there is an induced copy of the ordered $n$-grid which is homogeneous with respect to this coloring.
\end{lemma}

\begin{proof}
    Fix $k,n\in\N$. Let  $\RamseyGrid(k,n)\coloneqq N$
    be the number obtained from Lemma~\ref{lem:reramsey} applied to the values $k$, $\ell=4$ and $2n$. Let $\mathbf G_N$ be the ordered $N$-grid, and let
    $\chi\from V(\mathbf G_N)^2\to [k]$ be a pair coloring of pairs in $\mathbf G_N$ using $k$ colors.
    Define a coloring  $\lambda\from [N]^4\to [k]$
 by setting \[\lambda(a_1,a_2,b_1,b_2)\coloneqq \chi((a_1,a_2),(b_1,b_2))\qquad \text{for $(a_1,a_2,b_1,b_2)\in [N]^4$.}\]
 Thus, $\lambda$ defines a coloring of elements of $[N]^4$, using $k$ colors.

 By Lemma~\ref{lem:reramsey}, we can find a subset $I\subseteq [N]$
     of size $2n$
     such that the restriction of $\lambda$ to $I$ is homogeneous,
     that is,
      $\lambda(a_1,a_2,b_1,b_2)$ depends only on $\otp(a_1,a_2,b_1,b_2)$.

     Let $I_1$ be the first $n$ elements of $I$ and $I_2$ be the last $n$ elements of $I$.

     Consider the substructure $\mathbf A$  of $\mathbf G_N$ induced by
     pairs $(a_1,a_2)$ with $a_1\in I_1$ and $a_2\in I_2$.
     Clearly, $\mathbf A$ is isomorphic to $\mathbf G_n$.
     We show that the restriction of $\chi$
     to $\mathbf A$ is homogeneous.

Note that for every tuple $((a_1,a_2),(b_1,b_2))\in (I_1\times I_2)^2$,
we have that
\[
    \otp(a_1,a_2)=\otp(b_1,b_2)=\otp(a_1,b_2)=\otp(b_1,a_2)={<},
\]
since $a_1,b_1\in I_1$ and $a_2,b_2\in I_2$
and $x < y$ for all $x\in I_1$ and $y\in I_2$.

     Let $((a_1,a_2),(b_1,b_2))\in (I_1\times I_2)^2$
     and $((a'_1,a'_2),(b'_1,b'_2))\in (I_1\times I_2)^2$
     be two pairs such that
     \[\atp((a_1,a_2),(b_1,b_2))=\atp((a'_1,a'_2),(b'_1,b'_2)).\]
     This implies that
     $\otp(a_1,b_1)=\otp(a_1',b_1')$
     and $\otp(a_2,b_2)=\otp(a_2',b_2')$.
     By the previous observation, all the order types
     $\otp(a_1,a_2)$, $\otp(b_1,b_2)$, $\otp(a_1,b_2)$, $\otp(b_1,a_2)$, $\otp(a'_1,a'_2)$, $\otp(b'_1,b'_2)$, $\otp(a'_1,b'_2)$, $\otp(b'_1,a'_2)$ are equal to $<$.

     It follows that $\otp(a_1,a_2,b_1,b_2)=\otp(a'_1,a'_2,b'_1,b'_2)$.
     Therefore, $\lambda(a_1,a_2,b_1,b_2)=\lambda(a'_1,a'_2,b'_1,b'_2)$,
     by homogeneity of $\lambda$ restricted to $I$,
     which implies that
     $\chi((a_1,a_2),(b_1,b_2))=\chi((a'_1,a'_2),(b'_1,b'_2))$,
     as required.
     Hence, $\chi$ is homogeneous on $\mathbf A$.
     \end{proof}

%% file: chapters/hardness/patterns.tex
\subsection{Patterns in monadically stable classes}\label{sec:patterns}

Having established the auxiliary Ramsey tools, we proceed to the characterization of monadically stable classes of graphs in terms of forbidding patterns as flips of induced subgraphs: \cref{thm:intro-patterns}. We first recall the structure of the patterns appearing in \cite{flippergame}. See \cref{fig:obstruction} for a visualization.

\begin{definition}\label{def:patterns}
    For $n,\rho,k \in \N$, we say that a graph $G$ \emph{contains an $(n,\rho,k)$-rocket-pattern} if there is a collection of pairwise disjoint finite sets $A, B_1, B_2,\dots,B_n$ of vertices of $G$ such that $|A|=n$ and there is a $k$-flip $H$ of $G$ in which for all $i, j \in [n]$:
    \begin{enumerate}[label=(R.{\arabic*})]
        \item\label{pt:1} there is a semi-induced matching between $A$ and a subset $C_i$ of $B_i$;
        \item there are no edges between $A$ and $B_i \setminus C_i$;
        \item there are no edges between $B_i \setminus C_i$ and $B_j$ for $i \neq j$; and
        \item\label{pt:4} for every pair $u\neq v$ of vertices of $C_i$ there is a path of length at least $2$ and at most $\rho$ connecting $u$ and $v$, whose internal vertices all belong to $B_i \setminus C_i$.
     \end{enumerate}
\end{definition}

\begin{figure}[h]
    \centering
    \includegraphics[width=0.75 \linewidth]{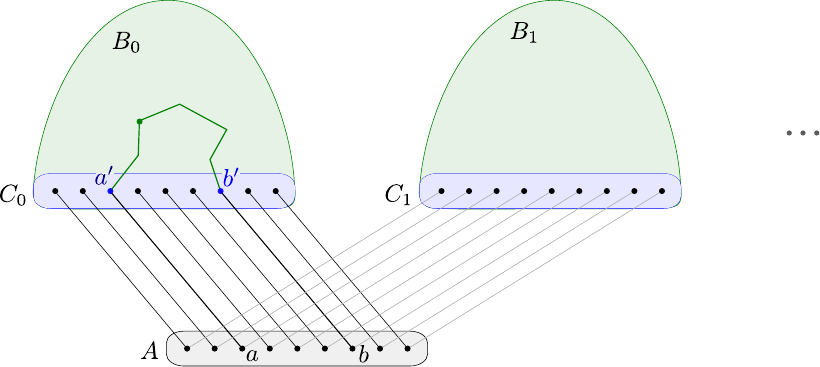}
    \caption{A pattern. Figure reproduced from~\cite{flippergame}.}\label{fig:obstruction}
\end{figure}

\begin{definition}
    A class of graphs $\CC$ \emph{admits rocket-patterns} if there are $\rho,k\in\N$ such that for all $n\in\N$, there is a graph $G_n\in\C$ which contains an $(n,\rho,k)$-rocket-pattern. Conversely, $\CC$ is \emph{rocket-pattern-free}, if $\C$ does not admit rocket-patterns.
\end{definition}

While rocket-patterns were only identified as an obstruction to monadic stability in \cite{flippergame}, the proofs in the same paper illustrate that these in fact characterize monadically stable classes under edge-stability. For the remaining of this section we therefore take the following as a fact, but also provide a short proof in \Cref{ap:patterns}, bridging the gaps between this and the statement of the main theorem in \cite{flippergame}. 

\begin{restatable}{theorem}{rocketpatterns}\label{flipper}
    A class of graphs $\C$ is monadically stable if and only if it is edge stable and rocket-pattern-free.
\end{restatable}
%

We now define one type of patterns that is instrumental to our analysis of monadic stability.

\begin{definition}\label{def:webs}
    For $r \geq 2$, we define the \emph{$r$-web} of \emph{order} $n$, denoted by $W^r_n$, to be the graph obtained by $r$-subdividing the complete graph of size $n$ and creating cliques between the neighbors of all native vertices.
\end{definition}

\begin{figure}[h!]
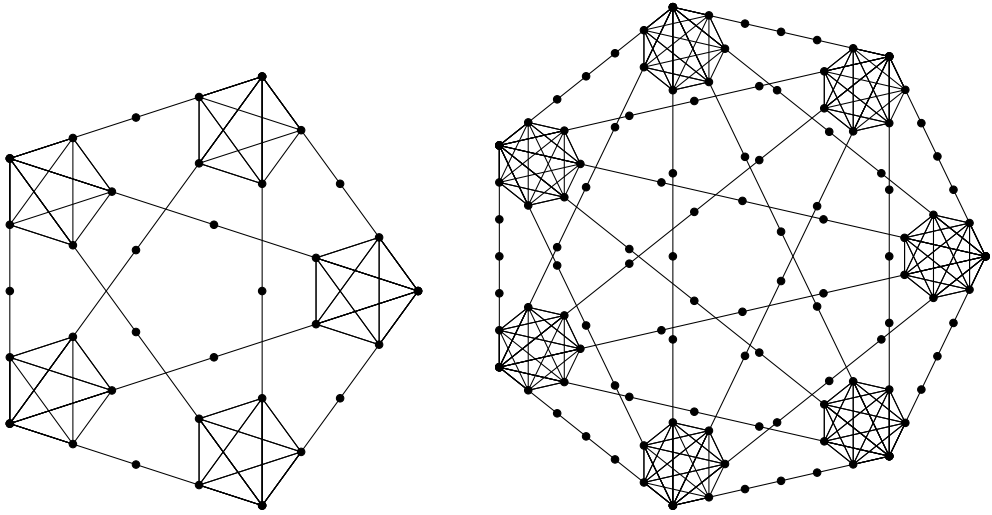

  \centering\small
  \spiderweb{3}{5}
  \qquad 
  \spiderweb[1.7]{5}{7}
  \caption{The graphs $W^3_5$ and $W^5_7$ respectively.}\label{fig:spiderweb}
\end{figure}

To simplify the proofs we shall also consider \emph{rook graphs} in our analysis, see \cref{fig:rook4}.

\begin{definition}
    The \emph{rook graph} of \emph{order} $n$, denoted by $R_n$, is the graph on vertex set $[n]\times [n]$ with the property that for all $i,j,k,\ell \in [n]$ there is an edge between $(i,j)$ and $(k,\ell)$ if and only if $i=k$ or $j=\ell$ (but not both).
\end{definition}
\begin{figure}[htbp]
    \centering
    \includegraphics{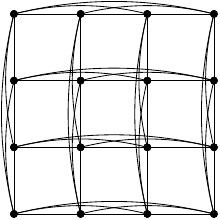}
    \caption{$R_4$: The rook graph of order $4$.}
    \label{fig:rook4}
\end{figure}

It turns out that rook graphs are $1$-flips of particular webs, as described in the lemma below. Hence, they disappear in the characterization of \cref{thm:patterns}.

\begin{lemma}\label{rooks}
    There is a $1$-flip of $R_{n^2}$ that contains $W^2_n$ as an induced subgraph. 
\end{lemma}

\begin{proof}
    Write $[n^2]\times[n^2]$ for the vertices of $R_{n^2}$, and consider the graph $G$ produced by flipping $V\coloneqq \{(1,i)\colon i \in [n]\}\subseteq R_{n^2}$ with itself. Pick an injection $f\colon {[n] \choose 2} \to [n^2]\setminus\{1\}$, and consider the subgraph of $G$ induced on \[G'\coloneqq V\cup \left\{(f(S),i),(f(S),j)\colon \text{ for }S = \{i,j\} \in {[n] \choose 2}\right\}.\]
    We argue that this subgraph is isomorphic to $W_n^2$.
    The situation is depicted in \cref{fig:web-in-rook}.
    The vertices of $V$ represent precisely the native vertices of the web, while for $S=\{i,j\}\subseteq [n]$, the sequence \[(1,i),(f(S),i),(f(S),j),(1,j)\] is a $2$-subdivided edge between $(1,i)$ and $(1,j)$. Moreover, for all $i \in [n]$ the neighbors of $(1,i)$ are all of the form $(x,i)$ for $x \in [n^2]$ and so they form a clique. Since there are no edges between $(f(S),i)$ and $(f(T),j)$ for $S \neq T$ and $i \neq j$ it follows that $G'$ is isomorphic to $W_n^2$.
\end{proof}

\begin{figure}[htbp]
    \centering
    \includegraphics{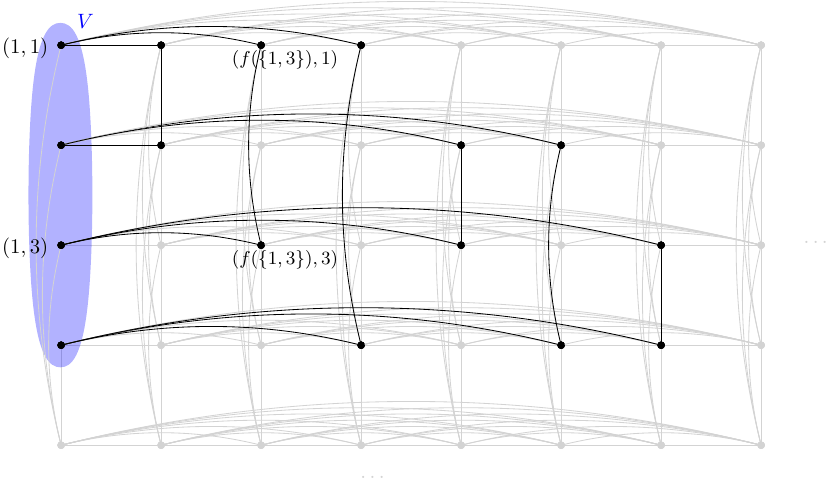}
    
    \caption{Embedding webs in flips of rook graphs. Depicted is $W^2_4$ as the subgraph induced by the black vertices in the $1$-flip of $R_{16}$ where the set~$V$ was flipped with itself.
    $V$ contains the native vertices of $W^2_4$.
    Marked in black are the edges of the $2$-subdivided clique subgraph of $W^2_4$.
    To improve readability, we refrained from marking the edges that form the cliques on the neighborhood of each native vertex in $W^2_4$.
    }
    \label{fig:web-in-rook}
\end{figure}

We now proceed with the main ingredient in our analysis. We argue that, under edge-stability, one may find either a flip of a subdivided clique or a flip of a web in a large enough rocket-pattern. Intuitively, the argument proceeds by regularizing rocket patterns through Ramsey theorems to ensure that adjacency within them depends on an appropriately chosen color. Edge-stability then implies that this color cannot be determined by any type of order, thus giving rise to certain specific cases that all lead to our patterns.

\begin{proposition}\label{lem:rocket}
    Let $\C$ be an edge-stable graph class admitting rocket-patterns. Then there exist $r,k\geq 2$ such that the class $\C_k$ of all $k$-flips of graphs in $\C$ contains either all $r$-subdivided cliques or all $r$-webs as induced~subgraphs.
\end{proposition}

\begin{proof}
    Let $\C$ be as above. By definition, we know that there are $k_1,\rho \in N$ such that for all $n \in \N$ there is a graph $H \in \C_{k_1}$ admitting an $(n,\rho,k_1)$-rocket pattern. We argue that for every $n \in \N$ there is a graph $G_n \in \C_{k_1+3}$ which contains, as an induced subgraph, either the $r_n$-subdivided clique or the $r_n$-web of order $n$, for some $1 \leq r_n\leq \rho+1$. By the pigeonhole principle, the sequence $(G_n)_{n \in \N}$ must necessarily contain an infinite subsequence of arbitrarily large $r$-subdivided cliques, or an infinite subsequence of arbitrarily large $r$-webs, for some fixed $1 \leq r \leq \rho+1$.
    The proposition then follows with $k = k_1 +3$, as the $r$-subdivided clique (respectively $r$-web) of order $n$ contains as induced subgraphs all $r$-subdivided cliques (respectively $r$-webs) of order at most $n$.
    
    Since $\C$ is edge-stable, it follows in particular that $\C_{k_1}$ is edge-stable; hence, there is some $d \in \N$ such that $\C_{k_1}$ does not contain semi-induced half-graphs of order $d$. Clearly, it suffices to show the claim for $n\geq 2d+1$. Assuming so, let \[t\coloneqq \Ramsey(\rho,n),\quad t_0\coloneqq t^2,\quad t_1\coloneqq \RamseyGrid(2,t_0),\quad t_2\coloneqq\Ramsey(2,t_1).\]
    where $\Ramsey(\cdot,\cdot)$ and $\RamseyGrid(\cdot,\cdot)$ denote the Ramsey and grid Ramsey functions from \Cref{lem:ramsey} and \Cref{gridramsey}, respectively.
    Let $H \in \C_{k_1}$ be a $k_1$-flip of a graph in $\C$, and $A, C_i\subseteq B_i$ subsets of $V(H)$ with $|A|=t_2, i \in [t_2]$, satisfying the assumptions in \Cref{def:patterns}. It follows by Ramsey's theorem that $A$ contains a set of size $t_1$ which is either a clique or an independent set; in any case, by performing a single flip if necessary, we assume that it is an independent set. For $m \in [t_1]$, write $A(m)$ for the $m$-th element in this set (ordered arbitrarily). Likewise, $C_i(m)$ is the element of $C_i$ that is matched with $A(m)$ according to \ref{pt:1}.

    Now, define a pair coloring of the $t_1$-grid $\chi\colon V(\mathbf G_{t_1})^2 \to \{0,1\}$ by letting $\chi((i,\alpha),(j,\beta))=1$ if and only if there is an edge between $C_i(\alpha)$ and $C_j(\beta)$. It therefore follows by Ramsey's theorem in grid form (\Cref{gridramsey}) that there is an induced copy of the $t_0$-grid in the $t_1$-grid which is homogeneous with respect to the coloring $\chi$. By relabeling if necessary, we may therefore assume that there is a well-defined map $\phi\colon \{<,=,>\}^2 \to \{0,1\}$ such that $\phi(\mathbf p)=1$ if and only if there is an edge between $C_i(\alpha)$ and $C_j(\beta)$, for every pair $((i,\alpha),(j,\beta))\in ([t_0]\times[t_0])^2$ of atomic type $\mathbf p$.

    The anti-reflexivity of the edge relation implies that $\phi(=,=)=0$, while symmetry ensures that $\phi$ satisfies the following conditions:
    \begin{eqnarray*}
     \phi(<,<)=\phi(>,>); & \qquad\qquad\qquad \phi(=,<)=\phi(=,>);\\
     \phi(<,=)=\phi(>,=); & \qquad\qquad\qquad \phi(<,>)=\phi(>,<).
    \end{eqnarray*}
    We additionally argue that $\phi(<,<)=\phi(<,>)$. By assumption, $2d+1\leq n \leq t_0 \leq t_1$. So, assuming for a contradiction that $\phi(<,<)=1$ and $\phi(<,>)=0$, it follows that the graph semi-induced on $X=\{C_{1}(2i)\colon i\leq d\}\subseteq V(H)$ and $Y=\{C_{2}(2j+1)\colon j\leq d\}\subseteq V(H)$ satisfies
    
    \[(C_{1}(2i),C_{2}(2j+1)) \in E(H) \iff i\leq j,\]
    contradicting that there are no semi-induced half-graphs of order $d$ in $\C_{k_1}$. We analogously obtain a contradiction if $\phi(<,<)=0$ and $\phi(<,>)=1$. Consequently, whether there is an edge between $C_i(\alpha)$ and $C_j(\beta)$ depends only on the equality type of $(i,j)$ and $(\alpha,\beta)$. Up to possibly performing a single flip of the sets $\bigcup_{i \in [t_0]} C_i$, we may additionally assume that $\phi(<,<)=0$.
    Having established
    \begin{eqnarray*}
        &&\phi(=,=)= \phi(<,<)=\phi(>,>) = \phi(<,>)=\phi(>,<) = 0,\\
        &&\phi(<,=)=\phi(>,=),\\
        &&\phi(=,<)=\phi(=,>),
    \end{eqnarray*}
    the behavior of $\phi$ is fully determined by one of four cases. 

    \begin{figure}[h]
        \centering
        \includegraphics[scale = 1]{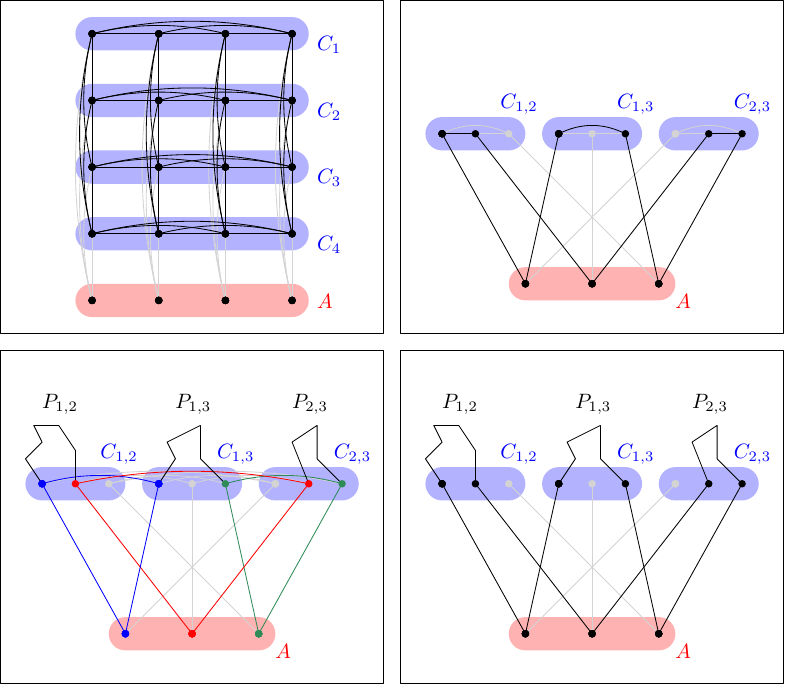}
        \caption{Four cases we encounter in the proof of \cref{lem:rocket}.
        }
        \label{fig:extract-patterns-multi}
    \end{figure}
    
    If $\phi(<,=)=\phi(=,<)=1$ then for all $i,j,k,\ell \in [t_0]$ we have
    \[ (C_i(j),C_k(\ell))\in E(H)\iff (i=k \land j\neq \ell) \lor (i \neq k \land j =\ell). \]
    The situation is depicted in the top-left panel of \cref{fig:extract-patterns-multi}.
    The graph induced on $\bigcup_{i \in [t_0]} C_i$ is therefore the rook graph of order $t_0$ and contains a $1$-flip of the $2$-web of order $t$ by \Cref{rooks}.
    As $t \geq n$, our claim~follows.

    We now handle the remaining three cases.
    By picking an appropriate bijection $[t_0] \to [t]\times[t]$, relabel the sets $C_i$  with $i \in [t_0]$ as $C_{i,j}$ with $i,j \in [t]$. By an application of Ramsey's theorem, there is a set $S$ of size $n$ such that for every $i,j \in S$ with $i<j$, the length of the path between $C_{i,j}(i)$ and $C_{i,j}(j)$ given by \ref{pt:4} is some fixed $r\in \{2,\ldots,\rho\}$. Write $P_{i,j}$ for the vertices appearing in the interior of this path. Thus far we have ensured that:
    \begin{itemize}
        \item the vertices $V\coloneqq \{A(i)\colon i \in S\}$ are independent;
        \item for all $i$ and $j<k<\ell$ belonging to $S$, there is an edge between $A(i)$ and $C_{j,k}(k)$ if and only if $k =i$, and an edge between $A(i)$ and $C_{k,\ell}(k)$ if and only if $k =i$.
        \item for all $i,j \in S$ with $i<j$, the vertices of $A(i),C_{i,j}(i),P_{i,j},C_{i,j}(j),A(j)$ form a path of length~$r+2$;
        \item for $i<j$ and $k<\ell$ belonging to $S$ with $(i,j)\neq(k,\ell)$, there are no edges between the vertices in $P_{i,j}$ and $P_{k,\ell}$;
        \item the adjacency within $U\coloneqq \bigcup \left\{C_{i,j}(i),C_{i,j}(j)\colon i,j\in S, i<j\right\}$ depends on $\phi$.
    \end{itemize}

    If $\phi(<,=)=0$ and $\phi(=,<)=1$ then for all $i<j$ from $S$ there is an edge $(C_{i,j}(i),C_{i,j}(j))$, while there are no edges between sets $C_{i,j}$ and $C_{k,\ell}$ for $(i,j)\neq (k,\ell)$.
    Then graph induced on $V \cup U$ is a $2$-subdivided clique of order $n$; see the top-right panel of \cref{fig:extract-patterns-multi}. If $\phi(<,=)=1$ and $\phi(=,<)=0$ then for all $i \in S$ the vertices $U_i\coloneqq \{C_{i,j}(i)\colon i,j\in S, i<j\}\cup\{C_{k,i}(i)\colon k,i,\in S, k<i\}\subseteq U$ form cliques, while there is no edge between $U_i$ and $U_j$ for $i \neq j$.
    It follows that the graph induced on $V \cup U \cup \bigcup\{P_{i,j}\colon i,j\in S, i<j\}$ is an $(r+1)$-web of order $n$; see the bottom-left panel of \cref{fig:extract-patterns-multi}.
    Finally, if $\phi(<,=)=\phi(=,<)=0$, then $U$ is an independent set.
    Then the graph induced on $V \cup U \cup \bigcup\{P_{i,j}\colon i,j\in S, i<j\}$ is an $(r+1)$-subdivided clique of order $n$; see the bottom-right panel of \cref{fig:extract-patterns-multi}. In either case, our claim~follows.
\end{proof} 

We are now in position to finish the proof of \cref{thm:intro-patterns}. For convenience we prove the following formulation, which is equivalent to \cref{thm:intro-patterns} by negation.


\begin{theorem}\label{thm:patterns}
    A class of graphs $\C$ is monadically stable if and only if $\C$ is edge-stable and for every $k,r \geq 2$ the class $\C_k$ of $k$-flips from $\C$ does not contain arbitrarily large $r$-subdivided cliques or arbitrarily large $r$-webs as induced subgraphs.
\end{theorem}

\begin{proof}
    Trivially, if $\C$ is monadically stable then it is edge-stable, while for every $k \in \N$ the class $\C_k$ is also monadically stable. Since for any $r \geq 2$ the class of $r$-subdivided cliques and the class of $r$-webs are not monadically stable,  these cannot occur as induced subgraphs of members of $\C_k$ for any~$k \in \N$.
    
    Conversely, assume that $\C$ is an edge-stable class of graphs that is not monadically stable. It follows by \Cref{flipper} that $\C$ admits rocket-patterns. Consequently, \Cref{lem:rocket} implies that there are $k,r\geq 2$ such that $\C_k$ contains either arbitrarily large $r$-subdivided cliques or $r$-webs, as required.
\end{proof}

%% file: chapters/hardness/model-checking-hardness.tex
\subsection{Hardness of model checking}

In this section we establish hardness of first-order model checking for hereditary, edge-stable, non-monadically stable graph classes; that is, we prove \cref{thm:intro-hardness}, recalled below.

\introhardness*

The main weight in the proof of \cref{thm:intro-hardness} lies in a construction of an existential interpretation that interprets the class of all graphs in $\C$. To allow transferring complexity hardness results, this interpretation needs to be effective in the following sense.

\begin{definition}
    We say a class of graphs $\CC$ \emph{effectively interprets} a class $\DD$, if there exists an interpretation $I$ and an algorithm which, given an input graph $H\in\DD$, computes in polynomial time an output graph $G \in \CC$ whose size is polynomial in $|V(H)|$, such that $I(G) = H$. We also say that $\CC$ \emph{effectively existentially interprets} $\DD$ if the interpretation $I$ is existential. 
\end{definition}

First-order formulas can be naturally pushed through interpretations.
More precisely, given an interpretation $I \coloneqq I_{\delta,\phi}$ and a formula $\psi(\bar x)$, we define $I(\psi)(\bar x)$ to be the formula obtained by recursively rewriting $\psi$ where we replace each
\begin{itemize}
    \item atomic subformula $E(x,y)$ with $\phi(x,y)$,
    \item existential quantification $\exists z: \alpha(\bar x, z)$ with $\exists z: \delta(z) \wedge \alpha(\bar x, z)$, and
    \item universal quantification $\forall z: \alpha(\bar x, z)$ with $\forall z: \delta(z) \rightarrow \alpha(\bar x, z)$.
\end{itemize}
We have the following standard fact.

\begin{fact}[{see, e.g., \cite[Theorem 4.3.1]{hodges}}]\label{fact:interp}
    For every interpretation $I$, formula $\psi(\bar x)$, graphs $G$ and $H$ satisfying $G = I(H)$, and tuple $\bar a \in V(G)^{|\bar x|}$,
    \[
        G \models \psi(\bar a)\qquad\textrm{if and only if}\qquad H \models I(\psi)(\bar a).
    \]
\end{fact}



The next lemma formalizes the intuition that the hardness of first-order model checking can be pulled through effective interpretations.

\begin{lemma}\label{lem:awhardness_reductions}
    Let $\CC$ be a class of graphs that effectively interprets the class of all graphs. Then the first-order model checking problem is $\mathrm{AW[*]}$-hard on $\CC$. If the interpretation is moreover existential, then existential first-order model checking on $\C$ is $\mathrm{W}[1]$-hard, and {\sc{Induced Subgraph Isomorphism}} on $\C$ is $\mathrm{W}[1]$-hard with respect to Turing reductions.
\end{lemma}
\begin{proof}
    Let $I$ be an effective interpretation that interprets the class of all graphs from $\CC$.
    It is immediate from \cref{fact:interp} that the $\mathrm{AW[*]}$-hard first-order model checking problem on the class of all graphs can be reduced to the first-order model checking problem on $\CC$.

    If $I$ is moreover existential we reduce the $\mathrm{W[1]}$-hard {\sc{Clique}} problem to the existential first-order model checking problem on $\CC$. Given a graph $G$ and a parameter $k$, we can compute in polynomial time a graph $H \in \CC$ such that $I(H) = G$.
    By \cref{fact:interp} we have that $G$ contains a clique of size $k$ if and only if $H \models I(\psi)$,
    where
    \[
        \psi \coloneqq \exists x_1, \ldots, x_k: \bigwedge_{1\leq i<j\leq k} E(x_i,x_j).
    \]
    Note that since $I$ and $\psi$ are existential and $\psi$ does not contain any negations, we can compute an existential sentence equivalent to $I(\psi)$. This finishes the reduction.

    Having shown that the existential first-order model checking problem is $\mathrm{W[1]}$-hard on $\CC$, we Turing reduce it to the {\sc{Induced Subgraph Isomorphism}} problem on $\CC$ using the following standard construction.
    Let $\psi$ be an existential sentence using $k$ quantifiers.
    In time depending only on $k$ we can compute the set $\DD$ of all graphs of size at most $k$ that model $\psi$. Now for every graph $G$, $G\models \psi$ if and only if $G$ contains a graph from $\DD$ as an induced subgraph.
\end{proof}


Thus, \cref{thm:intro-hardness} follows by combining \cref{lem:awhardness_reductions} with the following statement, whose proof spans the remainder of this section.

\begin{theorem}\label{thm:interpret}
    Let $\C$ be a hereditary, edge-stable, non-monadically stable class of graphs. Then $\C$ effectively existentially interprets the class of all graphs.
\end{theorem}

Let us discuss our approach to the proof of \cref{thm:interpret}.
Recall that \Cref{thm:patterns} implies that for any class $\C$ as above there are $k,r\geq 2$ such that $\C$ contains a $k$-flip of every $r$-subdivided clique, or $\C$ contains a $k$-flip of every $r$-web. It therefore suffices to show that each one of the two cases we may existentially interpret the class of all graphs on singletons. Despite the similarity between cliques and webs, we handle each one of these two cases individually as the two arguments rely on different structural properties of these patterns. Nonetheless, in both cases we shall consider the bipartite analogues of the patterns to simplify our arguments conceptually.

\subsubsection{Hardness in webs}

We first handle the case when the considered class contains a $k$-flip of every $r$-web. It will be convenient to conduct the reasoning on bipartite counterparts of webs, which we call {\em{biwebs}}; see~\cref{fig:biweb}

\begin{definition}
    Given $r \geq 2$ and $n,m\in \N$ we define the \emph{$r$-biweb} of \emph{order $(n,m)$}, denoted by $W^r_{n,m}$ as the bipartite graph obtained by $r$-subdividing the complete bipartite graph $K_{n,m}$ and turning the neighborhood of each native vertex into a clique.
\end{definition}

\begin{figure}[h!]
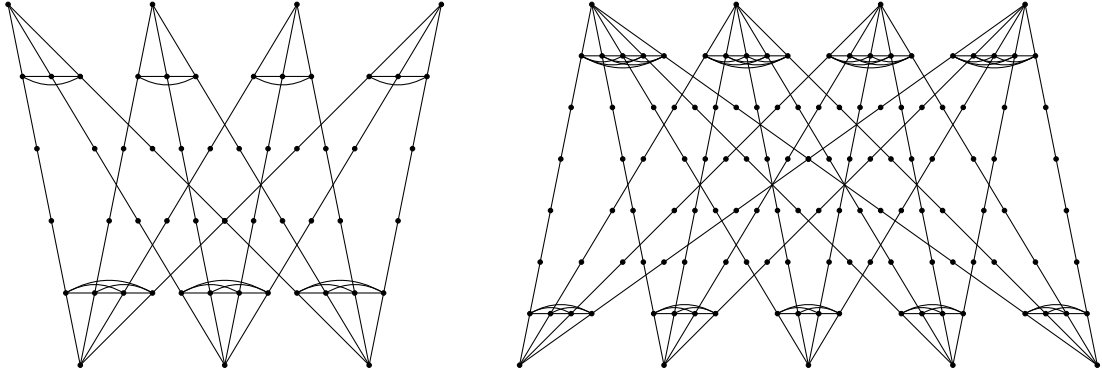

  \centering\small
  \biweb{4}{3}{4}
  \qquad 
  \biweb{6}{5}{4}  
  \caption{The graphs $W^4_{3,4}$ and $W^6_{5,4}$ respectively.}\label{fig:biweb}
\end{figure}

As is the case with subdivided bicliques and cliques, we may find biwebs within webs.

\begin{observation}
    For any $r \geq 2$ and $n \in \N$ the $r$-web $W^r_{2n}$ of order $2n$ contains the $r$-biweb $W^r_{n,n}$ of order $(n,n)$ as an induced subgraph. 
\end{observation}

Moreover, we may pass on to higher subdivision lengths from smaller ones. 

\begin{observation}
    The biweb $W^4_{n,n}$ is an induced subgraph of $W^2_{n^2,n^2}$. 
\end{observation}

With these observations at our disposal, we henceforth assume that $r \geq 3$ to simplify the proceeding arguments. Our main statement is the following.

\begin{theorem}\label{thm:everything-from-biwebs}
    Fix $r \geq 3$ and $k \in \N$.
    Let $\CC$ be a hereditary class of graphs containing a $k$-flip of every $r$-biweb.
    Then $\CC$ effectively interprets the class of all graphs using only existential formulas on singletons $\phi(x,y)$ and $\delta(x)$ for the edge and domain formula.
\end{theorem}

We break down the interpretation into two parts. For fixed $t,r \in \N$, we let
\[ \WW^r_t\coloneqq \{H + W^r_{t,t} \colon \text{$H$ is an induced subgraph of some $r$-biweb}\},\]
where $+$ denotes the disjoint union of two graphs. Observe that if $\C$ is as in \Cref{thm:everything-from-biwebs} then there is some $k\in \N$ such that the class $\C_k$ of all $k$-flips of graphs from $\C$ contains, as induced subgraphs, all of $\WW^r_t$ for $t \in \N$. Our first lemma gives an existential interpretation in the case where there are no flips involved.

\begin{lemma}\label{lem:noflipinterpret}
    For every $r\geq 3$ and $t \geq 5$ the class $\WW^r_t$ effectively existentially interprets the class of all graphs on singletons.
\end{lemma}

\begin{proof}
    Fix $r \geq 3$ and $t\geq 5$. Given an $n$-vertex graph $G$, we first describe the graph $f(G) \in \WW^r_t$ that will perform the interpretation; see \cref{fig:interpretBiweb}. 
    For every $v \in V(G)$, let $K^v$ be the graph obtained by taking $K_4$ and $K_{t+n}$, selecting an arbitrary vertex from $K_4$ and an arbitrary vertex from $K_{t+n}$, which we label $c_v$, and joining them by a path of length $r-1$. We then let $K^V$ be the disjoint union of $K^v$ over all $v \in V(G)$. Then, for every edge $uv \in E(G)$, we select an arbitrary vertex from the $(t+n)$-clique in $K^v$ which is not $c_v$, and an arbitrary vertex from the $(t+n)$-clique in $K^u$ which is not $c_u$, and join them by a path of length $2r-1$; moreover, vertices selected for different edges $uv\in E(G)$ are pairwise different. We write $B_G$ for the resulting graph. It is easy to see that $B_G$ is an induced subgraph of $W^r_{n+2,n+m+t}$ where $m\coloneqq |E(G)|$, and therefore $f(G)\coloneqq B_G + W^r_{t,t} \in \WW^r_t$. 
    Note that we add $2$ to the first dimension of the biweb to ensure that it is at least $3$ and we find $K_4$ as depicted in \cref{fig:interpretBiweb}.
    Moreover, $f(G)$ is clearly computable in polynomial time from~$G$.

    \begin{figure}[htbp]
        \centering
        \includegraphics{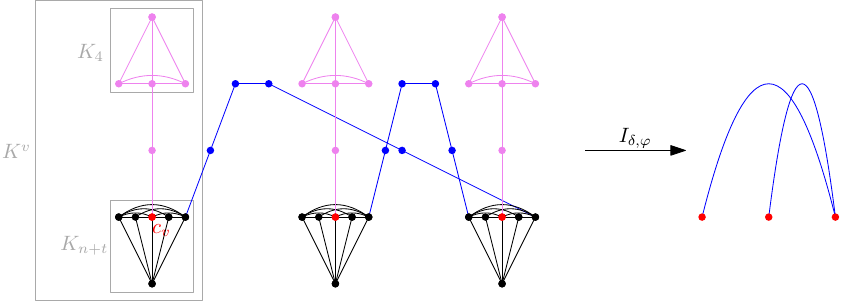}
        \caption{The graph $B_G$ constructed to interpret the graph $G$ on the right. The vertices are layouted to show how they embed into the layers of a $3$-biweb.}
        \label{fig:interpretBiweb}
    \end{figure}
    
    We proceed to define formulas $\delta(x)$ and $\phi(x,y)$ for the domain and edge relation, respectively, of the interpretation. For the former, we first let $\chi(x)$ (for ``kite'') be the existential formula that describes that there is an induced path of length $r-1$ starting at $x$ and leading to a clique of size~$4$. We then let
    \[ \delta(x) \coloneqq \chi(x) \land \deg_{> t}(x), \qquad \phi(x,y)\coloneqq \mathrm{dist}_{\leq 2r+1}(x,y),\]
    where $\deg_{> t}(\cdot)$ and $\mathrm{dist}_{\leq 2r+1}(\cdot,\cdot)$ are existential formulas checking the degree and the distance, respectively.
    Evidently, both $\chi(x)$ and $\phi(x,y)$ are existential formulas. 
    It is now easy to see that for every graph~$G$ we have $V(I_{\delta,\phi}(f(G))) = \{ c_v : v \in V(G) \}$ and $I_{\delta,\phi}(f(G))$ is isomorphic to $G$.
    Hence, $I_{\delta,\phi}$ is an interpretation of the class of all graphs in $\WW^r_t$.
\end{proof}

To obtain an interpretation from $\C$ instead of $\WW^r_t$, we would like to show that we can somehow definably ``undo'' the flip. More precisely, we aim to establish the following.

\begin{theorem}\label{thm:defundoflip}
    Fix $r\geq 3$ and $k\in \N$. Let $\C$ be a hereditary class of graphs containing a $k$-flip of every $r$-biweb. Then there exists some $t \geq 5$ such that for every existential formula $\phi(\bar x)$ there is an existential formula $\flipp(\phi)$, and for every graph $G \in \WW^r_t$ there is a graph $\flipp(G) \in \C$ which is a $k$-flip of $G$, so that
    \begin{equation}
        G \models \phi(\bar v) \qquad \textrm{if and only if} \qquad \flipp(G) \models \flipp(\phi)(\bar v),\tag{$*$}\label{eq:flip}
    \end{equation}
    for all tuples $\bar v \in V(G)^{|\bar x|}$. Moreover, $\flipp(G)$ can be computed from $G$ in polynomial time.
\end{theorem}

With this, \Cref{thm:everything-from-biwebs} becomes an easy consequence. 

\begin{proof}[Proof of \Cref{thm:everything-from-biwebs} from \Cref{lem:noflipinterpret} and \Cref{thm:defundoflip}]
    Let $I=(\delta(x),\phi(x,y))$ be the existential interpretation from \Cref{lem:noflipinterpret}, and let $t\in \N$ and $\flipp(\cdot)$ be the integer and the operator given by \Cref{thm:defundoflip}. Consider the existential interpretation $J=(\flipp(\delta),\flipp(\phi))$. Given an arbitrary graph $H$ there is, by \Cref{thm:defundoflip}, some $G \in \WW^r_t$ such that $H$ is isomorphic to $I(G)$. By \eqref{eq:flip}, we know that $I(G)$ is isomorphic to $J(\flipp(G))$. Since $\flipp(G)\in \C$ can be computed from $H$ in polynomial time, this gives an effective existential  interpretation of the class of all graphs in~$\C$.
\end{proof}

The remainder of this section is therefore devoted to establishing \Cref{thm:defundoflip}. 

\paragraph*{Layers, colors, and flips.}\label{pr:setup}
Fix $r,k$ and $\C$ as in \Cref{thm:defundoflip}. For notational convenience we write $B_n\coloneqq W^r_{n,n}$. So, given an $r$-biweb $B_n$, we partition its vertices into $\ell \coloneqq  r + 2$ \emph{layers} $\LL \coloneqq  \{ L_1, \ldots, L_{\ell} \}$ in the natural way: vertices of layer $L_i$ are those at distance $i-1$ from the native vertices on one side of biweb. Thus,  $L_1\cup L_{\ell}$ are the native vertices on the opposite sides of the biweb.
See the left side of \cref{fig:layers} for a visualization.

\begin{figure}[htbp]
    \centering
    \includegraphics{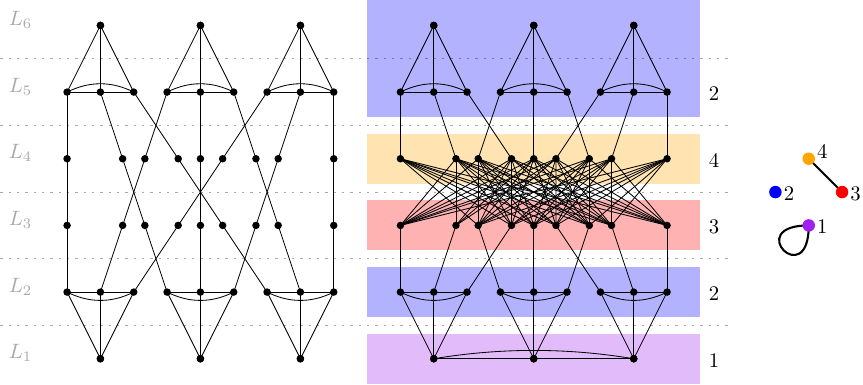}
    \caption{Left: a $4$-biweb partitioned into layers. Middle: a flip of a $4$-biweb partitioned into color classes.
    The exceptional colors in this example being $\Ex = \{ 2 \}$.
    Right: The graph $\str M$ depicting which color classes have been flipped.}
    \label{fig:layers}
\end{figure}

By assumption, for every $n\in \N$, $\C$ contains a $k$-flip of $B_n$. Recall that this means that there is a coloring $\flipCol_n\colon V(B_n)\to [k]$ and a symmetric relation $R_n\subseteq [k]\times [k]$ such that $B_n\oplus_{\flipCol_n} R_n\in \CC$. By Ramsey's theorem for bipartite graphs (see \cref{lem:bipramsey}) and possibly restricting attention to a smaller biweb (recall that $\C$ is hereditary), we may assume that for every length-$(r+1)$ path between native vertices on opposite sides of $B_n$, the sequence of colors assigned by $\flipCol_n$ to the vertices of this path is the same. Also, by pigeonhole principle we may assume that these sequences of colors are equal for all $n\in \N$, and also that all relations $R_n$ for $n\in \N$ are equal. All in all, these assumptions yield the following: there exist a number of colors $c \leq \min\{\ell,k\}$, a surjective mapping $\lc \colon \LL \to [c]$ (for \textbf{l}ayer \textbf{c}olor), and a symmetric relation $R \subseteq [c]^2$ with the following property.
\begin{quote}
    For every $r$-biweb $B_n$, let
    $\flipCol$ be the $c$-coloring of $V(B_n)$ in which each layer $L \in \LL$  is monochromatic and has color $\lc(L)$.
    Note that two different layers might be assigned the same color. See middle of \cref{fig:layers} for a visualization.

    Define $\flipp(B_n) \coloneqq B_n \oplus_\flipCol R$. Then we have that $\mathrm{flip}(B_n)$ is contained in $\CC$.
\end{quote}

Let now $\str M$ be the graph with vertex set $[c]$ and edge set $R$. We consider $\str M$ as a graph possibly with loops: there is a loop at $i\in [c]$ if $(i,i)\in R$. The neighborhood of $i$ is defined as $N_{\str M}(i)\coloneqq \{j\in [c]~|~(i,j)\in R\}$; note that $i\in N_{\str M}(i)$ if there is a loop at $i$. Distinct $i,j\in [c]$ are called {\em{twins}} if $N_{\str M}(i)=N_{\str M}(j)$. Note that if there are twins in $\str M$, then the corresponding colors can be merged in the mapping $\mathrm{lc}$, hence we may assume the following:
\begin{quote}
    The number of colors $c$ is minimal:
     the graph $\str M$ contains no twins.
\end{quote}
Note that this entails the following: for all distinct  $i, j \in [c]$, the symmetric difference
    \[
        \triangle_{\str M}(i,j) \coloneqq  \left( N_\str{M}(i) \setminus N_\str{M}(j)\right) \cup \left( N_\str{M}(j) \setminus N_\str{M}(i)\right)
    \]
    of the neighborhoods of $i$ and $j$ in $\str M$ is always non-empty.
    See right side of \cref{fig:layers} for a visualization.

Having fixed $\lc$, we will from now on assume that any $r$-biweb is implicitly $c$-colored.
For every induced subgraph $H$ of $B_n$ we define $\flipp(H)$ to be the graph induced by $V(H)$ in $\flipp(B_n)$.
Our notion of layers and colors carries over to $\flipp(B_n)$, $H$, and $\flipp(H)$ in the obvious way.
In every $r$-biweb $B_n$, the layers $L_2$ and $L_{\ell-1}$ are the only layers containing cliques.
We will refer to those cliques as \emph{cones}.
Hence, $B_n$ contains $2n$ cones: $L_2$ and $L_{\ell-1}$ each contain $n$ cones.
We will later use the cones to distinguish the \emph{exceptional layers} $L_2$ and $L_{\ell-1}$ from the \emph{non-exceptional layers} $\LL\setminus \{L_2,L_{\ell-1}\}$.
We call the set $\Ex \coloneqq  \{\lc(L_2), \lc(L_{\ell-1})\}$ the \emph{exceptional colors}.
As both exceptional layers can possibly be assigned to the same color, we have~$|\Ex| \in \{ 1,2 \}$.

Our aim is to find some $t \in \N$ such that for every $G\coloneqq H + B_t \in \WW^r_t$ we may definably recover the color of each vertex $v \in V(G)$ by looking at the adjacency in  $\flipp(G)$ between $v$ and certain carefully picked subsets of $\flipp(B_t)$. Since we are not allowed to treat $\flipp(B_t)$ as a tuple of parameters in our formula, we have to existentially quantify over some subset of $G$ that induces a graph isomorphic to $\flipp(B_t)$. However, it is possible that this captures a different copy of $\flipp(B_t)$ within $\flipp(G)$, rather than the intended one. Our aim is to establish that, even in this case, we can recover the color of a vertex by its adjacency to this existentially quantified subgraph.

Therefore, in proving \Cref{thm:defundoflip} we describe quantifier free formulas $\mathrm{check}(\bar z)$ and $E_0(x,y, \bar z)$ which specify the following rewriting procedure.
For every formula $\phi(\bar x)$ define
\[
    \flipp(\phi)(\bar x) \coloneqq  \exists \bar z(\mathrm{check}(\bar z) \wedge \hat\phi(\bar x,\bar z)),
\]
where $\hat \phi$ is obtained from $\phi$ by replacing every occurrence of $E(x,y)$ with $E_0(x,y,\bar z)$.
The goal is to choose $\check(\bar z)$ and $E_0(x,y,\bar z)$ in a way that ensures (\ref{eq:flip}). As we demand $\check(\bar z)$ and $E_0(x,y,\bar z)$ to be quantifier-free, then $\flipp(\phi)$ will be existential, provided that $\phi$ is.

\paragraph*{Defining the interpretation.}

For a sufficiently large $t \in \N$ that remains to be specified, let $(w_1, \ldots, w_q)$ be an arbitrary enumeration of the vertices of $B_t$, where $q\coloneqq |B_t|$.
Letting $\bar z=(z_1,\ldots,z_q)$ be a $q$-tuple of variables, the formula $\check(\bar z)$ will check whether the vertices $\bar z$ induce $\flipp(B_t)$ as a subgraph.
More precisely, we choose $\check(\bar z)$ as the conjunction of atomic and negated atomic formulas capturing the adjacency within $\flipp(B_t)$, so that
for every graph $G$ and tuple $\bar a = (a_1,\ldots,a_q) \in V(G)^q$ we have
\begin{equation}\label{eq:check-iso}
    G \models \check(\bar a)
    \quad
    \Leftrightarrow
    \quad
    \text{$\{a_i \mapsto w_i \colon i \in [c]\}$ is an isomorphism from $G[\bar a]$ to $\flipp(B_t)$.}
\end{equation}

We now carefully select disjoint subsets $Z_1,\ldots, Z_c$ of the variables of $\bar z$. Intuitively, we want the adjacency between a vertex $v$ and the realization of these variables to determine the color of $v$. Pick odd $s\in \N$ so that
\[s\geq 8c\ell.\]
For every $i \in [c]$, let $W_i \coloneqq  \{ w_j \colon z_j \in Z_i \}$ be the vertices from $\flipp(B_t)$ to which valuations of $Z_i$ are mapped by the isomorphism from \eqref{eq:check-iso}. Then we select $Z_1,\ldots,Z_c$ so that for every $i \in [c]$, we have the~following.

\newcommand{\fbound}{\frac{1}{4c} s^2}

\begin{enumerate}[leftmargin= 3em, label={(P.$\arabic*$)}]
        \item\label{itm:w-single-layer} $W_i$ is contained in a single layer $L$ of $\flipp(B_t)$ satisfying $\lc(L) = i$.

        \item\label{itm:w-size} $W_i$ has size $s^2$.
        
        \item\label{itm:w-ex-cones} If $i \in \Ex$, then 
        the $s^2$ vertices from $W_i$ are chosen to be $s$ batches of $s$ vertices, where each batch contained in a different cone.
        In particular $W_i$ induces in $\flipp(B_t)$ either
        \begin{itemize}
            \item $sK_s$: the disjoint union of $s$ cliques of size $s$ each, or 
            \hfill (this happens if $(i,i) \notin R$)
            \item $\overline{sK_s}$: the complement of the above. 
            \hfill (this happens if $(i,i) \in R$)

        \end{itemize}

        \item\label{itm:w-noex-hom} If $i \not\in \Ex$, then 
        $W_i$ induces in $\flipp(B_t)$ either
        \begin{itemize}
            \item an independent set, or 
            \hfill (this happens if $(i,i) \notin R$)
            \item a clique. 
            \hfill (this happens if $(i,i) \in R$)
        \end{itemize}
        
        \item\label{itm:w-pair-hom} 
        For all $j\in [c]\setminus \{i\}$, the bipartite graph semi-induced between $W_i$ and $W_j$ in $\flipp(B_t)$ is either
        \begin{itemize}
            \item anti-complete, or
            \hfill (this happens if $(i,j) \notin R$)
            \item complete.
            \hfill (this happens if $(i,j) \in R$)
        \end{itemize}
    \end{enumerate}

    Let us verify that such a selection is indeed possible, supposing $t$ is large enough.

    \begin{claim}\label{cl:maySelect}
        Assuming $t \geq cs^2$, there are sets $W_1,\dots,W_c$ from $\flipp(B_t)$ satisfying the properties above.
    \end{claim}

    \begin{claimproof}
        As $\flipp(B_t)$ and $B_t$ have the same vertex set, we will describe how to choose the desired vertices in the latter. 
        More precisely, we choose $W_1, \ldots, W_c$ from the graph $c B_{s^2}$, that is, the disjoint union of $c$ $r$-biwebs of order $s^2$ each.
        It is easy to see that $c B_{s^2}$ is an induced subgraph of $B_{cs^2}$, which in turn is contained in $B_t$ under our assumption on $t$.
        For each color $i \in [c]$ we will choose $W_i$ by picking $s^2$ vertices from layer $L(i)$ of the $i$th copy of $B_{s^2}$, where $L(i)$ is an arbitrary layer such that $\lc(L(i)) = i$ and $L(i)$ is exceptional if $i \in \Ex$.
        See \cref{fig:embedding-w} for a visualization.

        \begin{figure}[htbp]
            \centering
            \includegraphics[scale = 0.75]{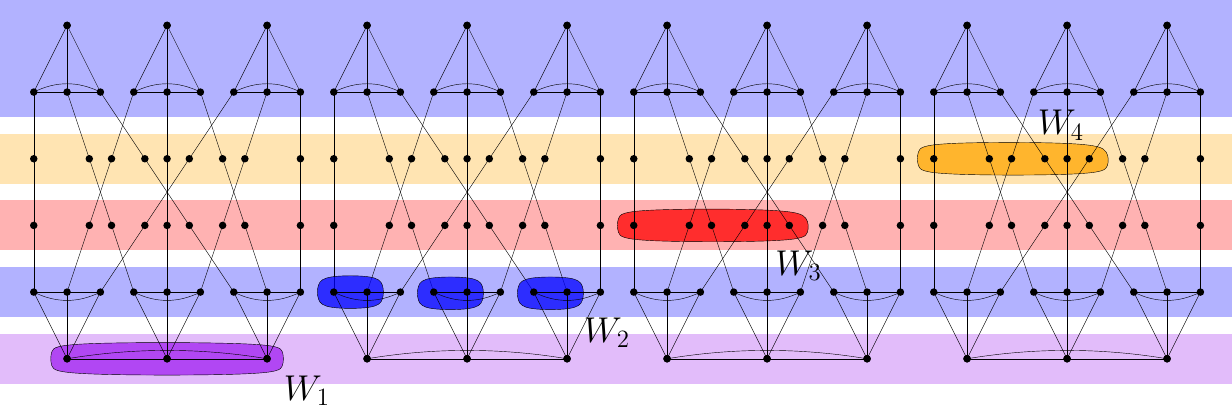}
            \caption{A depiction of how the sets $W_i$ embed into multiple copies of $B_{s^2}$. 
            For the sake of clarity, the sizes of the sets were not chosen faithfully.}
            \label{fig:embedding-w}
        \end{figure}

        It is easy to check that each layer in $B_{s^2}$ contains at least $s^2$ vertices, so this is possible.
        The $W_i$ picked this way satisfy \ref{itm:w-single-layer} and \ref{itm:w-size} by definition.
        To check \ref{itm:w-ex-cones} - \ref{itm:w-pair-hom}, we analyze the adjacencies of the sets in $\flipp(B_t)$.
        As we picked all $W_i$ from disjoint copies, they are pairwise  anti-complete to each other in $B_t$.
        As each $W_i$ is monochromatic, we know that in $\flipp(B_t)$ their pairwise adjacencies remain homogeneous and \ref{itm:w-pair-hom} is satisfied.
        It remains to show \ref{itm:w-noex-hom} and \ref{itm:w-ex-cones}.
        If $i \notin \Ex$ then $W_i$ was picked from a non-exceptional layer from the $i$th copy of $B_{s^2}$.
        This layer is an independent set in $B_t$ and as $W_i$ is monochromatic, \ref{itm:w-noex-hom} follows.
        If $i \in \Ex$ then $W_i$ was picked from an exceptional layer $L(i)$ from the $i$th copy of $B_{s^2}$.
        $L(i)$ induces $s^2K_{s^2}$ in $B_t$, so it is therefore possible to choose $W_i$ to induce $sK_s$ in $B_t$. As $W_i$ is monochromatic,
        \ref{itm:w-ex-cones} follows.
    \end{claimproof}

    Since the variables of $\check(\bar z)$ are in 1-1 correspondence with the vertices of $\flipp(B_t)$, the sets $W_i\subseteq V(B_t)$ uniquely specify the sets $Z_i$ of variables of $\bar z$. Having fixed the variables $Z_1,\dots,Z_c$ among $\bar z$, for each $i \in [c]$ 
    we define $\pc_i(x, \bar z)$ to be the quantifier-free formula checking whether
    \[
        \{j \in [c] \colon \text{$x$ is adjacent to at least half of the vertices of $Z_j$}\}=N_\str{M}(i).
    \]
    This may be easily defined by making a disjunction over all the possible ways that $x$ could be adjacent to the majority of $Z_j$, for all $j \in [c]$ such that $(i,j)\in R$, and non-adjacent to the majority of $Z_{j'}$, for all $j' \in [c]$ such that $(i,j')\notin R$.

    The benefit of this majority-type argument is that, even if $\check(\bar z)$ does not capture the intended copy of $\flipp(B_t)$
    within $\flipp(G)$, we can still determine the color of a vertex, possibly after some automorphism of the flip relation. This observation is captured by the following claim, which is the crux of our~approach.

    \begin{restatable}{claim}{mainclaim}\label{clm:ex-goal}
        For every induced subgraph $G$ of an $r$-biweb and tuple $\bar a$ from $G$ satisfying $\flipp(G) \models \check(\bar a)$, there exists an automorphism $f_\star$ of $\str M$ such that for every vertex $v$ in $G$ and $i \in [c]$,
    \[
        \text{$v$ has color $f_\star(i)$\qquad if and only if\qquad $\flipp(G) \models \pc_i(v,\bar a)$.}
    \]
    \end{restatable}

    We postpone the proof of this claim for now; it will be given \hyperref[pr:ex-goal]{further ahead}. Instead, we proceed with defining $E_0(x,y,\bar z)$ assuming that \cref{clm:ex-goal} is established.

    Towards this, we first define a quantifier-free formula $\flipped(x,y,\bar z)$ which detects whether the connection between $x$ and $y$ was flipped. More concretely, we let

        \[
        \flipped(x,y,\bar z) \coloneqq  \bigvee_{i,j \in [c]}
        \pc_i(x, \bar z) 
        \wedge
        \pc_j(y, \bar z) 
        \wedge
        \bigl((i,j) \in R\bigr).
    \]
    As $c$ and $R$ are fixed, $\flipped(x,y,\bar z)$ is once again a single quantifier-free formula. We argue that this formula does indeed behave in the intended way.

    \begin{claim}\label{clm:flippedworks}
    For every $G \in \WW^r_t$, tuple $\bar a\in V(G)^q$ satisfying $\flipp(G)\models \check(\bar a)$ and vertices $u,v \in V(G)$,
\[
    \flipp(G) \models \flipped(u,v, \bar a)
    \qquad
    \textrm{if and only if}
    \qquad
    \text{$u$ and $v$ are adjacent in exactly one of $G$ and $\flipp(G)$.}
\]
\end{claim}

\begin{claimproof}
By the definition of $\flipped(x,y,\bar z)$ we have
\[
    \flipp(G) \models \flipped(u,v, \bar a)
    \quad
    \Leftrightarrow
    \quad
    \bigvee_{i,j \in [c]} 
    \pc_i(u, \bar a) 
    \wedge
    \pc_j(v, \bar a) 
    \wedge
    \bigl((i,j) \in R\bigr).
\]
Assume $\flipp(G) \models \flipped(u,v, \bar a)$.
Then there exist $i,j \in [c]$ such that
\[
    \flipp(G) \models
    \pc_i(u, \bar z)  
    \wedge
    \pc_j(v, \bar z)  
    \wedge
    \bigl((i,j) \in R\bigr).
\]
By \cref{clm:ex-goal}, there exists an automorphism $f_\star$ on $\str M$ such that $u$ and $v$ have color $f_\star(i)$ and $f_\star(j)$ respectively.
As $f_\star$ is an automorphism, $(i,j) \in R$ implies that $(f_\star(i),f_\star(j))\in R$. Hence, the adjacency between $u$ and $v$ was flipped from $G$ to $\flipp(G)$ as desired.

On the other hand, assume $\flipp(G) \not\models \flipped(u,v, \bar a)$.
Then for all $i,j \in [c]$
\[
    \flipp(G) \not\models
    \pc_i(u, \bar z) 
    \wedge
    \pc_j(v, \bar z) 
    \wedge
    \bigl((i,j) \in R\bigr).
\]
Since $f_\star$ is an automorphism, there exist $i, j \in [c]$ such that $u$ and $v$ have color $f_\star(i)$ and $f_\star(j)$ respectively.
By  \cref{clm:ex-goal}, we have 
\[
    \flipp(G) \models
    \pc_i(u, \bar z)  
    \wedge
    \pc_j(v, \bar z)  
    \quad
    \text{ and consequently }
    \quad
    (i,j)\notin R.
\]
Again using the fact that $f_\star$ is an automorphism, we have $(f_\star(i),f_\star(j))\notin R$.
Then the adjacency between $u$ and $v$ was not flipped from $G$ to $\flipp(G)$ as desired.
\end{claimproof}

It is easy to see that with the formula $\flipped(x,y,\bar z)$ at our disposal, the formula
    \[
        E_0(x,y,\bar z) \coloneqq  E(x,y) \text{ XOR } \flipped(x,y,\bar z)
    \]
has the desired properties. Recall that given a formula $\phi$ in the language of graphs, we write $\hat \phi$ for the formula obtained by replacing every $E(x,y)$ atom by $E_0(x,y,z)$.

\begin{claim}\label{clm:translationworks}
    For all $G \in \WW^r_t$, formulas $\phi(\bar x)$, and tuples $\bar v$ from $V(G)^{|\bar x|}$ the following are equivalent:
    \begin{enumerate}
        \item\label{itm:qf} $G \models \phi(\bar v)$;
        \item\label{itm:uni} $\flipp(G) \models \forall \bar z(\check(\bar z) \to \hat{\phi}(\bar v,\bar z))$;
        \item\label{itm:ex} $\flipp(G) \models \exists \bar z(\check(\bar z) \land \hat{\phi}(\bar v,\bar z)) \quad (=\flipp(\phi)(\bar v))$.
        
    \end{enumerate}
\end{claim}

\begin{proof}
    We handle the case where $\phi(x,y)\coloneqq E(x,y)$; the rest follows by induction on the structure of $\phi$. Fix $u,v \in V(G)$. Suppose that $G \models E(u,v)$ and let $\bar a \in V(G)^{|\bar z|}$ such that $\flipp(G)\models \check(\bar a)$. If the edge $(u,v)$ was flipped from $G$ to $\flipp(G)$ then by construction $\flipp(G)\models \neg E(u,v)\land \flipped(u,v,\bar a)$. Consequently, $\flipp(G) \models E_0(u,v,\bar a)$. If the edge $(u,v)$ was not flipped from $G$ to $\flipp(G)$ then $\flipp(G) \models E(u,v) \land \neg \flipped(u,v,\bar a)$, and so $\flipp(G) \models E_0(u,v,\bar a)$. This shows $\ref{itm:qf} \implies \ref{itm:uni}$.

    Suppose that $\flipp(G) \models \forall \bar z(\check(\bar z) \to E_0(u,v,\bar z))$. Letting $\bar w$ be the enumeration of the vertices in $B_t$, we see that $\flipp(G) \models \check(\bar w) \land E_0(u,v,\bar w)$. Consequently, $\flipp(G) \models \exists \bar z(\check(\bar z) \land E_0(u,v,\bar z))$ and so $\ref{itm:uni} \implies \ref{itm:ex}$.

    Finally, assume that $\flipp(G) \models \exists \bar z(\check(\bar z) \land E_0(u,v,\bar z))$ and let $\bar a \in V(G)^{|\bar z|}$ be the tuple witnessing this. Assume that $\flipp(G)\models E(u,v).$ It follows that $\flipp(G)\models \neg\flipped(u,v,\bar a)$ and so the edge $(u,v)$ was not flipped from $G$ to $\flipp(G)$; hence $G \models E(u,v)$. On the other hand, if $\flipp(G)\models \neg E(u,v)$ then $\flipp(G)\models \flipped(u,v,\bar a)$ and so the edge $(u,v)$ was flipped; it follows that $G \models E(u,v)$. This implies that $3 \implies 1$.
\end{proof}

    Noting that the third formula was precisely defined to be $\flipp(\phi)$, \Cref{thm:defundoflip} follows directly from the above. Note that we set $t\coloneqq cs^2$ so that \Cref{cl:maySelect} holds.

\subsubsection*{Proof of \Cref{clm:ex-goal}}\label{pr:ex-goal}

We now proceed to the proof of \Cref{clm:ex-goal}, which we restate below for convenience. In the following statement and proof, our setup is \hyperref[pr:setup]{precisely the same as before}.

\mainclaim*

We henceforth fix an induced subgraph $G$ of an $r$-web, and a tuple $\bar a=(a_1,\dots,a_q)\in V(G)^q$ such that $\flipp(G)\models\check(\bar a)$. For every $i \in [c]$, let $A_i \coloneqq  \{a_j \colon z_j \in Z_i\} \subseteq V(G)$ be the valuations of the variables in $Z_i$.
As $\flipp(G) \models \check(\bar a)$, there is an isomorphism from $\flipp(G)[\bar a]$ to $\flipp(B_t)$ that maps $A_i$ to $W_i$.
It follows that the properties
\ref{itm:w-size} - \ref{itm:w-pair-hom}, which determine the structure induced by the vertices $W_i$ in $\flipp(B_t)$, carry over to the vertices $A_i$ in $\flipp(G)$. However, a priori we cannot guarantee a property in the spirit of \ref{itm:w-single-layer} for the $A_i$. We know that $\bar a$ is some embedding of $\flipp(B_t)$ in $\flipp(G)$, but we do not know from which layers of $\flipp(G)$ the vertices in $\bar a$ and its $A_i$ subsets were chosen.

The main idea of the proof is that sets $A_i$ are chosen large enough that by pigeonhole principle, they will contain large subsets that will faithfully ``represent'' different colors. Precisely, our analysis will be based on the following definition.

\begin{definition}
 A function $f\colon [c]\to [c]$ is {\em{popular}} if for every $i\in [c]$, $A_i$ contains at least $ \frac{1}{4c} |A_i| = \fbound$ vertices of color $f(i)$.
\end{definition}

%

Note that the pigeonhole principle implies that there is always a popular function $f$, and this would hold even if we relaxed the $\fbound$ bound to $\frac{1}{c}s^2$.
We weaken the bound deliberately for technical reasons. Surprisingly, we will establish that any popular function $f\colon [c]\to [c]$ actually gives an automorphism of the graph $\str M$. Towards this, we make a sequence of observations about the structure of the sets $A_i$ and properties of popular functions.

\begin{observation}\label{obs:forbidden-subgraphs}
    In $G$
    no vertex contains in its neighborhood an induced $2K_2$ or $\overline{2K_2}$.
\end{observation}

\begin{observation}\label{obs:noex-layers-hom}
    Every non-exceptional layer induces an independent set or a clique in $\flipp(G)$.
\end{observation}

\begin{observation}\label{obs:ex-not-hom}
    Let $x \in \Ex$. Every subset of $A_x$ of size $2s$ contains vertices inducing in $\flipp(G)$ 
    \begin{itemize}
        \item $2K_2$, if $(i,i) \notin R$, or
        \item $\overline{2K_2}$, if $(i,i) \in R$.
    \end{itemize}
\end{observation}

The previous two observations combined with the pigeonhole principle give the following.
\begin{observation}\label{clm:few-noex-layers}
    Let $x \in \Ex$.
    $A_x$ contains less than $(\ell-2)\cdot 2s$ vertices from non-exceptional layers.
\end{observation}

We shall also make use of the following claim in our arguments.

\begin{claim}\label{clm:diffsets}
    Let $i,j \in [c]$ be distinct.
    Then for all $4$-element subsets $V_i \subseteq A_i$ and $V_j \subseteq A_j$,
    there is a $4$-element set $V_{ij} \subseteq V(G)$ such that
    \begin{enumerate}
        \item $V_{ij}$ is disjoint from $V_i$ and $V_j$;
        \item $V_{ij}$ is homogeneous to $V_i$ in $\flipp(G)$;
        \item $V_{ij}$ is homogeneous to $V_j$ in $\flipp(G)$;
        \item $V_{ij}$ is inhomogeneous to $V_i \cup V_j$ in $\flipp(G)$; and
        \item $V_{ij}$ is contained in a single layer in $\flipp(G)$.
    \end{enumerate}
\end{claim}

\begin{claimproof}
    We show the existence of a set $V_{ij}$ of size $4\ell$ which satisfies the first four properties.
    By the pigeonhole principle, we then find a set satisfying all properties.
    Choose any $d \in \triangle_{\str M}(i,j)$. (Recall here that $\triangle_{\str M}(i,j)$ is non-empty by the assumption that $\str M$ has no twins.)
    
    First assume $d = i$ and
    $(i,d) \notin R \ni (j, d)$.
    Then either
    \begin{itemize}
        \item $i \in \Ex$ and $A_i$ induces $sK_s$ in $\flipp(G)$, or
        \hfill(by \ref{itm:w-ex-cones}) 

        \item $i \notin \Ex$ and $A_i$ induces an independent set in $\flipp(G)$.
        \hfill(by \ref{itm:w-noex-hom}) 
    \end{itemize}
    In both cases, as $s \geq 8k\ell$ and so in particular $s \geq 4\ell+4$, we find a set $V_{ij}$ in $A_i$ which is disjoint from and non-adjacent to $V_i$.
    By \ref{itm:w-pair-hom}, $v$ is adjacent to all of $A_{j}$ including $V_j$ and $V_{ij}$ has the desired properties.

    If $d=i$ but instead $(i,d) \in R \not \ni (j, d)$, we argue as before exchanging cliques with independent sets and adjacencies with non-adjacencies.
    If $d = j$, we argue by symmetry.
    If $d \notin \{i,j\}$, we argue by \ref{itm:w-pair-hom}.
    This exhausts all cases.
    \end{claimproof}

With all the above, we now proceed to establishing that every popular function is in fact an automorphism of $\str M$.

\begin{claim}\label{clm:ex-bijection}
    For every popular function $f$, $f\vert_\Ex$ is a bijection from $\Ex$ to $\Ex$.
\end{claim}

\begin{claimproof}
    Assume first there is only a single exceptional color: $\Ex = \{x\}$.
    By popularity, $A_x$ contains at least $\frac{1}{4c} s^2$ vertices of color $f(x)$.
    By \cref{clm:few-noex-layers}, there must be a vertex from an exceptional layer with color $f(x)$, since $(\ell-2)\cdot 2s \leq \fbound$.
    It follows that $f(x) \in \Ex$ so $f(x) = x$, and we are done.

    Assume now there are two exceptional colors:
    $\Ex =\{x_1,x_2\}$.
    Following the previous case, we know $f(x_1) \in \Ex$ and $f(x_2) \in \Ex$ and it remains to show $f(x_1) \neq f(x_2)$.
    Assume towards a contradiction that $f(x_1) = f(x_2)$.
    Let $V'_1 \subseteq A_{x_1}$ and $V'_2 \subseteq A_{x_2}$ be the two sets of size $\frac{s^2}{4c}$ which both have color $f(x_1)$.
    By \cref{obs:ex-not-hom}, we find subsets $V_1 \subseteq V_1'$ and $V_2 \subseteq V_2'$ which both either induce $2K_2$ or $\overline{2K_2}$ in $\flipp(G)$; here we use $2s \leq \fbound$.
    \cref{clm:diffsets} yields a vertex $v$ which is homogeneous to $V_1$ and homogeneous to $V_2$ but inhomogeneous to $V_1 \cup V_2$ and not included in any of the sets.
    As $V_1$ and $V_2$ have the same color, no matter the choice of $R$, in the non-flipped graph $G$ the vertex $v$ is adjacent to either all of $V_1$ or all of $V_2$.
    Now \cref{obs:forbidden-subgraphs} yields a contradiction
\end{claimproof}

\begin{claim}\label{clm:noex}
    For every popular function $f$ and all $i \notin \Ex$, we have $f(i) \notin \Ex$.
\end{claim}

\begin{claimproof}
    By \cref{clm:ex-bijection}, for every $x' \in \Ex$, there exists an $x \in \Ex$ with $f(x) = x'$.
    To prove the claim it therefore suffices to show that for all $i \notin \Ex$ and $x \in \Ex$:
    $f(i) \neq f(x)$.
    Assume towards a contradiction that there is a color $d\in [c]$ with $ d = f(i) = f(x)$.
    Since $2s\leq \fbound$, there exists a size four subset $V_x \subseteq A_x$ such that
    \begin{itemize}
        \item $V_x$ either induces $2K_2$ or $\overline{2K_2}$ in $\flipp(G)$; and
        \hfill(by \cref{obs:ex-not-hom})
        \item all vertices in $V_x$ have color $d$.
        \hfill(by popularity)
    \end{itemize}
    Let $V_i$ be a size four subset from a single layer and of color $d$ from $A_i$; such $V_i$ exists due to $4\leq \frac{s^2}{4c\ell}$.
    We obtain $V_{ix}$ by applying \cref{clm:diffsets} to $V_i$ and $V_x$.
    
    Recall that $V_{ix}$ is inhomogeneous to $V_i\cup V_x$, but homogeneous to each of $V_i$ and $V_x$ separately. Since both $V_{ix}$ and $V_i\cup V_x$ are monochromatic, this conclusion also holds in graph $G$. By \cref{obs:forbidden-subgraphs}, $V_{ix}$ cannot be complete to $V_x$ in $G$, for $V_x$ induces $2K_2$ or $\overline{2K_2}$ in $\flipp(G)$, so also in $G$. Hence in $G$, $V_{ix}$ must be anti-complete to $V_x$ and complete to $V_i$.


    By the structure of a biweb, $V_i$ and $V_{ix}$ must form a clique and are located in the same exceptional layer in $G$.
    Recall that there is an isomorphism from $\flipp(G)[\bar a]$ to $\flipp(B_t)$ which maps $A_i$ to $W_i$.
    By $i\notin\Ex$ and \ref{itm:w-single-layer}, $W_i$ is contained in a single non-exceptional layer $L$ in $\flipp(B_t)$.
    Analyzing the structure of the non-flipped biweb $B_t$, it is easy to see that for every vertex $u$ in $W_i$ there exist two distinct vertices $v_1,v_2$ whose only neighbor in the layer $L$ is $u$.
    Lifting this fact through the flip to $\flipp(B_t)$ and through the isomorphism to $\flipp(G)[\bar a]$, we obtain that for every vertex $u \in A_i$, there exist two distinct vertices $v_1,v_2$ with the following property. 
    In $\flipp(G)$, each of $v_1,v_2$ is 
    either 
    \begin{enumerate}[leftmargin= 3em, label={(C.$\arabic*$)}]
        \item\label{itm:case-1-n} adjacent to $u$ and non-adjacent to all other vertices from $V_i \subseteq A_i$, or
        \item\label{itm:case-n-1} non-adjacent to $u$ and adjacent to all other vertices from $V_i \subseteq A_i$.
    \end{enumerate}
    $V_i$ is monochromatic by assumption. As the vertices $v_1$ and $v_2$ are singletons, each of them is trivially monochromatic, too.
    It follows that the adjacency between each of $v_1$ and $v_2$ and $V_i$ in $G$ is the same as in either $\flipp(G)$ or its complement.
    Since also the cases \ref{itm:case-1-n} and \ref{itm:case-n-1} are complementary, they also describe the connection of $v_1$ and $v_2$ towards $V_i$ in the non-flipped graph~$G$.
    Recall that $V_i$ forms a clique in $G$ and fix a vertex $u \in V_i$. Let $v_1, v_2$ be the two distinct vertices satisfying either \ref{itm:case-1-n} or \ref{itm:case-n-1}.
    Depending on which of the two cases applies for each vertex and whether $v_1$ and $v_2$ are connected, $G$ contains one of the three graphs listed in \cref{fig:ex-forbidden-subgraphs}, none of which is contained in $G$ as an induced subgraph.
    We have reached the desired contradiction.
\end{claimproof}

    \begin{figure}[htbp]
        \centering
        \includegraphics[scale = 1]{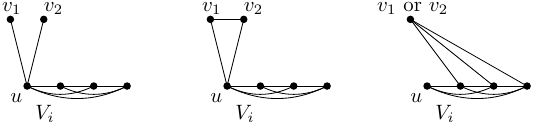}
        \caption{Three induced subgraphs that cannot appear in an $r$-biweb.}
        \label{fig:ex-forbidden-subgraphs}
    \end{figure}

\begin{claim}\label{clm:strong-homomorphism}
    Every popular function $f$ is a strong homomorphism from $\str M$ to $\str M$:
    \[
        \forall i,j \in [c]:\quad (i,j)\in R \Leftrightarrow (f(i),f(j)) \in R.
    \]
\end{claim}

\begin{claimproof}
    We do a case analysis. 

    \medskip\noindent
    Case 1: $i = j$ and $i \in \Ex$.
    As $2s \leq \fbound$, popularity and \cref{obs:ex-not-hom} imply that the color class $f(i)$ contains four vertices that induce ${2K_2}$ in $\flipp(G)$ if $(i,i)\notin R$ or $\overline{2K_2}$ if $(i,i)\in R$.
    Since $\overline{2K_2}$ is not an induced subgraph of $G$, we have that the color class $f(i)$ was flipped if and only if $i$ was, that is, $(f(i),f(i)) \in R \iff (i,i) \in R$.

    \medskip\noindent
    Case 2: $i = j$ and $i \notin \Ex$.
    By popularity, \cref{clm:noex}, and \ref{itm:w-noex-hom}, the color class $f(i)$ contains four vertices from non-exceptional layers, that induce and independent set if $(i,i)\notin R$ or a clique if $(i,i)\in R$ in $\flipp(G)$; here we use that $4\leq \fbound$.
    Since the non-exceptional layers in $G$ contain no clique of size four, we have that the color class $f(i)$ was flipped if and only if $i$ was, that is, $(f(i),f(i)) \in R \iff (i,i) \in R$.

    \medskip\noindent
    Case 3: $i \neq j$ and at least one of $i$ and $j$ is contained in $\Ex$.
    By symmetry, we can assume $i$ is contained in~$\Ex$.
    By popularity and \cref{obs:ex-not-hom}, the color class $f(i)$ contains a four vertex set $V_i$ that induces ${2K_2}$ or $\overline{2K_2}$ in $\flipp(G)$; here we use that $2s \leq \fbound$.
    As $V_i$ is monochromatic the same holds in $G$.
    By popularity and \ref{itm:w-pair-hom}, $A_j$ contains a vertex $v$ with color $f(j)$ that is adjacent to all of $V_i$ in $\flipp(G)$ if $(i,j) \in R$ or adjacent to none of $V_i$ in $\flipp(G)$ if $(i,j) \notin R$.
    We now derive $(i,j)\in R \Leftrightarrow (f(i),f(j)) \in R$ from \cref{obs:forbidden-subgraphs}.

    \medskip\noindent
    Case 4: $i \neq j$ and none of $i$ and $j$ is contained in $\Ex$.
    By popularity and \cref{clm:noex},
    there exist two four element subsets $S_i \subseteq A_i$ and $S_j \subseteq A_j$ with color $f(i)$ and $f(j)$ respectively, none of which contains vertices from exceptional layers; here we use that $4\leq \fbound$.
    By \ref{itm:w-pair-hom}, $S_i$ and $S_j$ are either anti-complete or complete to each other in $\flipp(G)$ depending on whether $(i,j) \in R$.
    Since the non-flipped graph $G$ contains no semi-induced bicliques of order four between vertices of non-exceptional layers, we conclude $(i,j)\in R \Leftrightarrow (f(i),f(j)) \in R$.

    \medskip\noindent
    This exhausts all cases.
\end{claimproof}

\begin{claim}\label{clm:aut}
    Every popular function $f$ is an automorphism on $\str M$.
\end{claim}

\begin{claimproof}
    It remains to show that $f$ is injective.
    Assume towards a contradiction that there are $i\neq j$ both not in $\Ex$ such that $f(i) = f(j)$.
    As $\str M$ contains no twins, we can assume by symmetry that there exists $d\in [c]$ such that $(i,d) \in R$ and $(j,d) \notin R$.
    Now as $f$ is a strong homomorphism, we get $(f(i),f(d)) \in R$ and $(f(i) = f(j),f(d)) \notin R$; a contradiction.
\end{claimproof}

Now that we established that popular functions are automorphisms of $\str M$, we can prove that they actually color the majority of each set $A_i$ with color $f(i)$. In particular, from the claim below it follows that there is a unique popular function.

\begin{claim}\label{clm:f-big}
    For every popular function $f$ and $i\in [c]$, more than $\frac{3}{4}|A_i|$ vertices in $A_i$ have color $f(i)$.
\end{claim}

\begin{claimproof}
    Assume towards a contradiction that $A_i$ contains at least $\frac{1}{4}|A_i|$ vertices not of color $f(i)$.
    By the pigeonhole principle, $A_i$ contains $\fbound$ vertices from a single color $j \neq f(i)$. We can now construct a new function $f'$ as follows:
    \[
        f'(x) \coloneqq
        \begin{cases}
            j & \text{if $x = i$;}\\
            f(x) &\text{otherwise.}
        \end{cases}
    \]
    As $f$ is an automorphism on the finite graph $\str M$, $f'$ is not injective.
    However, $f'$ is a popular function by construction, hence it must be an
    automorphism by \cref{clm:aut}; a contradiction.
\end{claimproof}

In fact, we obtain the following. 

\begin{claim}\label{clm:ex-small-neighborhoods}
    For every popular function $f$, $i \in [c]$, and $v \in V(G)$, there are more than $\frac{1}{2}|A_i|$ vertices of color $f(i)$ in $A_i$ to which $v$ is different and non-adjacent in~$G$.
\end{claim}
 
\begin{claimproof}
    Assume $i \notin \Ex$. 
    By \cref{clm:f-big}, at least $\frac{3}{4}|A_i|$ vertices in $A_i$ have color $f(i)$.
    By \cref{clm:noex}, none of those are from exceptional layers.
    In $G$, every vertex has at most $2$ neighbors in non-exceptional layers.
    As $\frac{3}{4}|A_i| - 2 \geq \frac{1}{2}|A_i|$, the desired bound follows.

    Assume $i \in \Ex$.
    Let $V_i$ be the $f(i)$ colored subset of $A_i$ of size at least $\frac{3}{4}|A_i|$.
    By the structure of a biweb, all but at most two vertices from $N_G(v)$ are contained in the same clique in an exceptional layer.
    Assume now towards a contradiction $|N_G(v) \cap V_i| \geq 2s + 2$.
    Then $v$ is adjacent to $2s$ vertices from $V_i$ which form a clique in a single exceptional layer.
    However by \cref{obs:ex-not-hom}, no subset of size $2s$ from the same layer of $A_i$ can induce a homogeneous subgraph in $G$; a contradiction.
    Now $|N_G(v) \cap V_i| < 2s + 2$, hence there exists more than $\frac{3}{4}|A_i|-(2s+3)\geq \frac{1}{2}|A_i|$ vertices of color $f$ in $A_i$ to which $v$ is different and non-adjacent~in~$G$.
\end{claimproof}

\begin{claim}\label{clm:ex-predict-adj}
    For every popular function $f$ and $i,j \in [c]$ the following holds.
    Let $v$ be a vertex of $G$ with color~$i$.
    Then we have $(i,f(j)) \in R$ if and only if $v$ is adjacent to at least half of the vertices of $A_j$ in $\flipp(G)$.
\end{claim}
\begin{claimproof}
 By \cref{clm:ex-small-neighborhoods}, there is a set $B_j\subseteq A_j$ such that $v\notin B_j$, $|B_j|>\frac{|A_j|}{2}$, all vertices of $B_j$ are of color~$f(j)$, and $v$ is anti-complete to $B_j$ in $G$. If $(i,f(j))\notin R$, then the adjacency between $v$ and $B_j$ remains the same in $\flipp(G)$, hence $v$ is adjacent to less than half of the vertices of $A_j$ in $\flipp(G)$. On the other hand, if $(i,f(j))\notin R$, then the adjacency between $v$ and $B_j$ is flipped in $\flipp(G)$, so $v$ is adjacent to more than half of vertices of $A_j$ in $\flipp(G)$.
\end{claimproof}

%

Having established the above, we finally proceed with a proof of \Cref{clm:ex-goal}. 

\begin{proof}[Proof of \Cref{clm:ex-goal}]
    Fix $G$ and $\bar a \in V(G)^{q}$ satisfying  $\flipp(G) \models \check(\bar a)$ as before. Let $f_\star\colon [c]\to [c]$ be the unique popular function; the existence follows from the pigeonhole principle, and the uniqueness follows from \cref{clm:f-big}. We show that for every vertex $v$ of $G$ and every color $i \in [c]$,
    \[
        v\textrm{ has color }f_\star(i) \qquad\textrm{if and only if}\qquad\flipp(G) \models \pc_i(v,\bar a).
    \]

    So, fix some $v$ with color $f_\star(i)$. Then for every $j \in [c]$ the following are equivalent:
    \begin{itemize}
        \item $v$ is adjacent to at least half of $A_j$ in $\flipp(G)$;
        \item $(f_\star(i),f_\star(j)) \in R$; and
        \hfill(by \cref{clm:ex-predict-adj})
        \item $(i,j) \in R$.
        \hfill(since $f_\star$ is an automorphism)
    \end{itemize}
    Consequently, the set of colors $j \in [c]$ such that $v$ is adjacent to at least half of $A_j$ in $\flipp(G)$ is precisely the neighborhood $N_{\str M}(i)$ of $i$ in $\str M$, and so $\flipp(G) \models \pc_i(v,\bar a)$ as~required.

    On the other hand, assume towards a contradiction that $v$ has color $f_\star(j) \neq f_\star(i)$ for some $j \in [c]$, but $\flipp(G) \models \pc_i(v,\bar a)$.
    As $f_\star$ is an automorphism, we have $i \neq j$. By the argument above we obtain that $\flipp(G) \models \pc_j(v,\bar a)$. In particular, this implies that $i$ and $j$ are twins in $\str M$, contradicting our assumption that $\str M$ has no twins.
\end{proof}


\subsubsection{Hardness in subdivided cliques}

We now handle the case of subdivided cliques. For $r,n \in \N$, by an \emph{$r$-biclique} of \emph{order} $n$ we simply mean the $r$-subdivided complete bipartite graph $K_{n,n}$ with parts of size $n$; this will be denoted by $K_{n,n}^r$. Evidently, for every $r \in \N$ the $r$-biclique of order $n$ is an induced subgraph of the $r$-subdivided clique of size $2n$. Also, we may once again pass to higher subdivision lengths without loss of generality, because $K^r_{n^2,n^2}$ contains $K^{2r+1}_{n,n}$ as an induced subgraph. In this section we establish the following statement, which together with \cref{thm:everything-from-biwebs} proves \cref{thm:interpret}.

\begin{theorem}\label{thm:everything-from-sbcs}
    Fix $r \geq 3$ and $k \in \N$.
    Let $\CC$ be a hereditary class of graphs containing a $k$-flip of every $r$-biclique.
    Then $\CC$ effectively interprets the class of all graphs using only existential formulas on singletons $\phi(x,y)$ and $\delta(x)$ for the edge and domain formula.
\end{theorem}

The proof will closely follow the proof of \cref{thm:everything-from-biwebs} for $r$-biwebs, and so we only highlight the differences. Once again we break the interpretation into two parts. For $r,t \in \N$ we let 

\[ \BB^r_t \coloneqq \{H + K^r_{t,t}\colon \text{$H$ is an induced subgraph of an $r$-biclique}\}.\]

We have the following analogue of \cref{lem:noflipinterpret}.

\begin{lemma}\label{clm:sbc-no-flips-interp}
        For $r\geq 3$ and every $t\geq 5$ the class $\BB^r_t$ effectively existentially interprets the class of all graphs on singletons.
\end{lemma}

\begin{proof}
    We replicate the proof of \Cref{lem:noflipinterpret}, replacing cliques with stars. We write $S_n$ for the star of order~$n$, i.e. the undirected graph consisting of a center and $n-1$ petals adjacent to the center. So, fix $r \geq 3$ and $t\geq 5$. Given a graph $G$, we describe the graph $f(G) \in \BB^r_t$ that will perform the interpretation. For every $v \in V(G)$, let $K^v$ be the graph obtained by taking $S_4$ and $S_{t}$, and joining the center of $S_4$ with the center of $S_t$ by a path of length $r+1$; the center of $S_t$ will be denoted by $c_v$. We then let $K^V$ be the disjoint union of $K^v$ over all $v \in V(G)$. Then, for every edge $(u,v) \in E(G)$, we join $c_u$ and $c_v$ by a path of length $2r+2$. We write $B_G$ for the resulting graph. It is easy to see that $B_G$ is an induced subgraph of $K^r_{n+3,n+m+t}$ where $n\coloneqq |V(G)|, m\coloneqq |E(G)|$, and therefore $f(G)\coloneqq B_G + K^r_{t,t} \in \BB^r_t$. Moreover, $f(G)$ is clearly computable in polynomial time from $G$.

    We proceed to define formulas $\delta(x)$ and $\phi(x,y)$ for the domain and edge relation, respectively, of the interpretation. For the former, we first let $\chi(x)$ be the existential formula that describes that there is an induced path of length $r+1$ starting at $x$ and leading to the center of a star of order $4$. We then let
    \[ \delta(x) \coloneqq  \chi(x) \land \deg_{> t}(x), \qquad \phi(x,y)\coloneqq  \mathrm{dist}_{\leq 2r+2}(x,y),\]
    where again $\deg_{> t}(\cdot)$ and $\mathrm{dist}_{\leq 2r+2}(\cdot,\cdot)$ are existential formulas verifying the degree and the distance, respectively.
    Clearly, both $\delta(x)$ and $\phi(x,y)$ are existentially formulas. 
    It is now easy to see that for every graph~$G$ we have $V(I_{\delta,\phi}(f(G))) = \{ c_v : v \in V(G) \}$ and $I_{\delta,\phi}(f(G))$ is isomorphic to $G$.
    Hence, $I_{\delta,\phi}$ is an interpretation of the class of all graphs in $\BB^r_t$.
\end{proof}

Next, we establish the following analogue of \Cref{thm:defundoflip}.

\begin{theorem}\label{thm:sbc-undo-flips}
    Fix $r\geq 3$, $k\in \N$ and $\C$ a hereditary class of graphs containing a $k$-flip of every $r$-biclique. Then there exists some $t \in \N$ such that for every existential formula $\phi(\bar x)$ there is an existential formula $\flipp(\phi)$, and for every graph $G \in \BB^r_t$ there is a graph $\flipp(G) \in \C$ which is a $k$-flip of $G$ and such that
    \begin{equation}\tag{$*$}\label{eq:flip}
        G \models \phi(\bar v) \iff \flipp(G) \models \flipp(\phi)(\bar v),
    \end{equation}
    for all tuples $\bar v \in V(G)^{|\bar x|}$. Moreover, we can compute $\flipp(G)$ from $G$ in polynomial time. 
\end{theorem}

As before, \Cref{clm:sbc-no-flips-interp} and \Cref{thm:sbc-undo-flips} together imply \Cref{thm:everything-from-sbcs}. We therefore focus on \Cref{thm:sbc-undo-flips}, and henceforth fix $r,k,$ and $\C$ as in the statement of this theorem. 

Here, our notions of colors, layers, and flips carry over from biwebs to subdivided bicliques in the natural manner. More precisely, letting $\LL$ be the set of layers in an $r$-biclique, we once again can use Ramsey's Theorem and pigeonhole principle to assume the following. There is a function $\lc\colon \LL\to [c]$ and a symmetric relation $R \subseteq [c]^2$ such that
for every $r$-biclique $K_{n,m}$ the $c$-coloring $\flipCol$ of $V(K_{n,m})$ in which each layer $L \in \LL$  is monochromatic and has color $\lc(L)$ satisfies $\flipp(K_{n,m}) \coloneqq  K_{n,m} \oplus_\flipCol R \in \CC$. Also, we may again assume that $c$ is minimal, which  implies that the graph $\str M\coloneqq ([c],R)$ contains no twins.
As before, we write $\Ex$ for the set of exceptional colors: $\Ex \coloneqq \{ i \in [c] \colon \lc(L_2)=i \text{ or } \lc(L_{\ell -1}) =i\}$.

We will use the same definition 
\[\flipp(\phi)(\bar x) \coloneqq  \exists \bar z( \mathrm{check}(\bar z) \wedge \hat\phi(\bar x,\bar z)),\]
as before, where for $\check(\bar z)$ we now demand 
\[        G \models \check(\bar a)
        \quad
        \textrm{if and only if}
        \quad
        \text{$\{a_i \mapsto w_i \colon i \in [m]\}$ is an isomorphism from $G[\bar a]$ to $\flipp(K^r_{t,t})$,}
\]
for a fixed enumeration $(w_1, \ldots, w_q)$ of the vertices of $K^r_{t,t}$, as expected. (Here, $q$ is the size of $K^r_{t,t}$.)
The formulas $\hat \phi$, $E_0(x,y,\bar z)$, and $\flipped(x,y,\bar z)$ are constructed with the variable sets $Z_1, \ldots, Z_c$ selected in the same way, differing only in property \ref{itm:w-ex-cones}. 
As the adjacencies inside the cones of $r$-bicliques and $r$-biwebs differ, property \ref{itm:w-ex-cones} is replaced with the following.

\begin{enumerate}[leftmargin= 3em, label={(P.$3\star$)}]

        \item\label{itm:w-ex-cones2} 
        If $i \in \Ex$, then 
        the $s^2$ vertices from $W_i$ are chosen to be $s$ batches of $s$ vertices, where each batch contained in a different cone. 
        In particular $W_i$ induces in $\flipp(K^r_{t,t})$ either
        \begin{itemize}
            \item an independent set, or 
            \hfill (this happens if $(i,i) \notin R$)
            \item a clique. 
            \hfill (this happens if $(i,i) \in R$)

        \end{itemize}

    \end{enumerate}

It is straightforward to see that the analogue of \cref{cl:maySelect} holds: provided we set $t\geq cs^2$, there is a choice of sets $W_1,\ldots,W_c$ from $\flipp(K^r_{t,t})$ satisfying all the demanded properties.

    Our aim is to once again infer the color of any vertex $v$ from its adjacency, in $\flipp(G)$, towards the vertices evaluated to $Z_i$. That is, we shall establish the following analogue of \Cref{clm:ex-goal}, where the formula $\pc_i(x,\bar z)$ is defined exactly as before.
    
    \begin{claim}\label{clm:bcl-ex-goal}
        For every induced subgraph $G$ of an $r$-biclique and tuple $\bar a$ from $G$ satisfying $\flipp(G) \models \check(\bar a)$, there exists an automorphism $f_\star$ on $\str M$ such that for every vertex $v$ in $G$ and $i \in [c]$, 
    \[
        \text{$v$ has color $f_\star(i)$\qquad if and only if\qquad $\flipp(G) \models \pc_i(v,\bar a)$.}
    \]
    \end{claim}
    
     \Cref{clm:bcl-ex-goal} suffices to conclude \Cref{thm:sbc-undo-flips}, as \Cref{clm:flippedworks}, \Cref{clm:translationworks}, and their respective proofs, entirely translate to the biclique setting. So from now on we focus on proving \cref{clm:bcl-ex-goal}.

    Towards this, we fix $G$ and $\bar a= a_1,\dots,a_n \in V(G)^{|\bar x|}$ such that $G \models \check(\bar a)$, and for every $i \in [c]$ we write $A_i = \{a_j \colon z_j \in Z_i\}\subseteq V(G)$. We again work with the same notion of a {\em{popular function}}: this is a function $f\colon [c]\to [c]$ such that at least $\fbound$ elements of $A_i$ are colored with color $f(i)$.
    We first show an analogue of \cref{clm:aut}: every popular function is an automorphism on $\str M$. The proof here is arguably much simpler than in the case of biwebs; it makes use of the following observation.

    \begin{observation}\label{obs:sbc-no-bc}
        For $r \geq 3$, no graph $G \in \BB^r_t$ contains $K_{2,2}$, the biclique of order $2$, as a semi-induced subgraph.
    \end{observation}
    
    \begin{claim}\label{clm:bcs-automorphism}
        Every popular function $f$ is an automorphism on $\str M$.
    \end{claim}

    \begin{claimproof}
    Fix some popular function $f$. We first argue that $f$ is a strong homomorphism, that is, $(i,j)\in R$ if and only if $(f(i),f(j))\in R$. Here we distinguish two cases.

    Assume first that $i=j$. By popularity, there are two vertices within $A_i$, both of color $f(i)$, which are connected if and only if $(i,i) \in R$. Since every layer in $G$ is independent, we have that the color class $f(i)$ was flipped itself if and only if $i$ was flipped with itself, as required.
         
    Now suppose $i \neq j$. Then popularity implies that there exist two $2$-element subsets $S_i \subseteq A_i$ and $S_j \subseteq A_j$ with colors $f(i)$ and $f(j)$ respectively.
    By \ref{itm:w-pair-hom}, $S_i$ and $S_j$ are either anti-complete or complete to each other in $\flipp(G)$, depending on whether $(i,j) \in R$. Since the non-flipped graph $G$ contains no semi-induced bicliques of order $2$, $S_i$ and $S_j$ are complete to each other if and only if $f(i)$ and $f(j)$ were flipped. Consequently, $(i,j)\in R$ if and only if $(f(i),f(j)) \in R$.

    To show injectivity, take colors $i \neq j$ from $[c]$. Since there are no twins in $\str M$, we can assume by symmetry that there is $d \in [c]$ such that $(i,d) \in R$ and $(j,d) \not\in R$. Since $f$ is a strong homomorphism it follows that $(f(i),f(d)) \in R$ while $(f(j),f(d)) \not\in R$; consequently, it must be that $f(i)\neq f(j)$ as required.
    \end{claimproof}

    We also have an analogue of \cref{clm:f-big}.

    \begin{claim}\label{clm:f-big2}
        For every $i \in [c]$, more than $\frac{1}{2}|A_i|$ vertices in $A_i$ have color $f_\star(i)$.
    \end{claim}
\begin{claimproof}
    Exactly the same as the proof of \Cref{clm:f-big}, using \Cref{clm:bcs-automorphism} in place of \Cref{clm:aut}.
\end{claimproof}

    We proceed by arguing the analogues of \Cref{clm:ex-small-neighborhoods,clm:ex-predict-adj}, which amount to showing that any popular function can be used to decode colors through a majority vote over adjacencies to sets $A_i$. For this, we make the following claim.

    \begin{claim}\label{obs:sbc-ex-not-hom}
        Let $x \in \Ex$. Every size $2s$ subset of $A_x$ contains vertices $v_1,v_2,v_3,v_4$ such that there exists a vertex $u$ such that in $\flipp(G)$, $u$ is adjacent to $v_1, v_2$ and non-adjacent to $v_3,v_4$.
    \end{claim}
    \begin{proof}
        Using \ref{itm:w-ex-cones2}, we have that  
        every size $2s$ subset of $A_x$ contains four vertices $v_1,\ldots, v_4$ which in $\flipp(G)[\bar a]$ are isomorphic to vertices $w_{i_1},\ldots,w_{i_4}$ in $\flipp(K^r_{t,t})$, where 
        \begin{itemize}
            \item $w_{i_1}$ and $w_{i_2}$ are from the same cone,
            \item $w_{i_3}$ and $w_{i_4}$ are from the same cone,
            \item $w_{i_1}$ and $w_{i_3}$ are from different cones.
        \end{itemize}
        In $K^r_{t,t}$, there exists a vertex $w_{i_5}$ (which spans the first cone) that is adjacent to $w_{i_1}$ and $w_{i_2}$ and non-adjacent to $w_{i_3}$ and $w_{i_4}$.
        Then in $\flipp(K^r_{t,t})$, $w_{i_5}$ is adjacent to exactly two vertices from $w_{i_1}, \ldots, w_{i_4}$.
        Using the isomorphism, we have that there is a vertex $u$ which is adjacent in $\flipp(G)$ to exactly two vertices from $v_1,\ldots,v_4$. 
        Up to renaming, this proves the claim.
    \end{proof}

    \begin{claim}\label{clm:sbc-ex-small-neighborhoods}
        For every popular function $f$, $i \in [c]$, $v \in V(G)$, there are more than $\frac{1}{2}|A_i|$ vertices of color $f_\star(i)$ in $A_i$ to which $v$ is different and non-adjacent in~$G$.
    \end{claim}
    \begin{claimproof}
        We distinguish two cases. Assume that $v \in V(G)$ is not a native vertex of $G$, that is, it is neither in layer $L_1$ nor in layer $L_\ell$. It follows that $v$ has at most two neighbors in $G$, and since \Cref{clm:f-big2} implies that at least $\frac{3}{4}|A_i|$ vertices in $A_i$ have color $f_\star(i)$, there are at least $\frac{3}{4}|A_i|-3\geq \frac{1}{2}|A_i|$ vertices in $A_i$ of color $f_\star(i)$ to which $v$ is different and non-adjacent in $G$.
        
        Now, assume that $v$ is in one of the layers $L_1$ or $L_\ell$. Let $V_i$ be the $f_\star(i)$ colored subset of $A_i$ of size at least $\frac{3}{4}|A_i|$ that exists by \Cref{clm:f-big2}. Assume $|N_G(v) \cap V_i| \geq 2s$. Since $v$ is adjacent to a vertex of $V_i$ in $G$, this implies that $i \in \Ex$.
        Then by \cref{obs:sbc-ex-not-hom} there exist a four element subset $V_i' \subseteq N_G(v) \cap V_i$ and a vertex $u$ that is adjacent to exactly two vertices from $V_i'$ in $\flipp(G)$. In particular, this means that $u \neq v$. Since $V_i'$ is monochromatic, the same holds in $G$, so $\{u,v\}$ and two vertices from $V_i'$ semi-induce $K_{2,2}$ in $G$; a contradiction to \cref{obs:sbc-no-bc}. It follows that $|N_G(v) \cap V_i| < 2s$, hence there are more than $\frac{3}{4}|A_i|-(2s+1)\geq \frac{1}{2}|A_i|$ vertices of color $f_\star(i)$ in $A_i$ to which $v$ is different and non-adjacent in~$G$.
    \end{claimproof}

\begin{claim}\label{clm:sbc-ex-predict-adj}
    For every popular function $f$ and $i,j \in [c]$ the following holds.
    Let $v$ be a vertex of $G$ with color~$i$.
    Then we have $(i,f(j)) \in R$ if and only if $v$ is adjacent to at least half of the vertices of $A_j$ in $\flipp(G)$.
\end{claim}
\begin{claimproof}
    Exactly the same as the proof of \Cref{clm:ex-predict-adj}, using \Cref{clm:sbc-ex-small-neighborhoods} in place of \Cref{clm:ex-small-neighborhoods}.
\end{claimproof}

 With \cref{clm:sbc-ex-predict-adj} established, the proof of \Cref{clm:bcl-ex-goal} follows directly in analogy to the proof of \Cref{clm:ex-goal}.

%% file: chapters/conclusions.tex
\section{Conclusions}\label{sec:conclusions}

In this work we prove that first-order model checking is fixed-parameter tractable on every monadically stable class of graphs. Moreover, monadic stability turns out to be the dividing line between tractability and intractability on edge-stable class of graphs, which verifies \cref{conj:main} on such classes. With the hardness part of \cref{conj:main} established in~\cite{CharacterizingNIP}, it remains to prove the tractability part of \cref{conj:main}.

This, however, seems to be a major challenge. The model checking algorithm for monadically stable classes completed in this work relies on two fundamental components: the Flipper game and sparse neighborhood covers. The analogues of both these components are so far unclear in the general monadically NIP setting. The largest unknown is to find a suitable analogue of the Flipper game, or any other kind of a decomposition mechanism that would guide a model checking algorithm. As far as sparse neighborhood covers are concerned, the construction presented in \cref{sec:neicov} for monadic stability can be applied to any class property that ensures almost linear neighborhood complexity and is preserved by taking powers. Therefore, to construct sparse neighborhood covers in monadically NIP classes, it suffices to resolve the following conjecture.

\begin{conjecture}\label{conj:neicomp-NIP}
     Let $\C$ be a monadically NIP graph class and $\eps>0$. Then for every $G\in \C$ and $A\subseteq V(G)$,
    \[|\{N_G(u)\cap A \colon u\in V(G)\}|\leq \Oh_{\Cc,\eps}(|A|^{1+\eps}).\]
\end{conjecture}

While resolving \cref{conj:neicomp-NIP} would provide a tangible progress towards \cref{conj:main}, our result about almost linear neighborhood complexity in monadically stable classes, \cref{thm:nei-comp}, still leaves a room for improvement. Namely, it is known~\cite{PilipczukST18a} that in nowhere dense classes, the result can be generalized to the setting of {\em{$\varphi$-types}} over a set of parameters $A$, where $\varphi(\bar x,\bar y)$ is any first order formula operating over pairs of {\em{tuples}} of vertices, rather than just on pairs of vertices. Our proof of \cref{thm:nei-comp} does {\em{not}} generalize to this setting, hence the following conjecture remains open.

\begin{conjecture}\label{conj:types-stable}
     Let $\C$ be a monadically stable graph class, $\varphi(\bar x,\bar y)$ be a first-order formula, and $\eps>0$. Then for every $G\in \C$ and $A\subseteq V(G)$,
    \[\left|\left\{\bigl\{\bar a\in A^{|\bar x|}\colon G\models \varphi(\bar a,\bar b)\bigr\}~\colon~\bar b\in V(G)^{|\bar y|}\right\}\right|\leq \Oh_{\Cc,\eps,\varphi}(|A|^{|\bar y|+\eps}).\]
\end{conjecture}

%% file: chapters/appendix.tex
\section{Rocket-patterns}\label{ap:patterns}

Here we establish the following result, which was crucially used within \Cref{sec:patterns} to prove \Cref{thm:patterns}.

\rocketpatterns*

Though not explicitly stated as such, this essentially follows by results of \cite{flippergame}. The aim of this section is to bridge the gap between the precise statement in \cite{flippergame} and the formulation of \Cref{flipper}. We first recall the relevant background.

We work both with classes of finite graphs and with infinite graphs. Finite graphs will be typically denoted by $G,H,\ldots$, while we reserve notation $\str M$ for a possibly infinite graph.

We assume reader's familiarity with basic notions of model theory such as structures, vocabularies, models, etc. In particular, we will work with (first-order) theories of classes of structures and with their elementary closures, defined as follows.

\begin{definition}
    Let $\C$ be a class of structures over a relational vocabulary $\mathcal{L}$. We write $\Th(\C)$ for the set of first-order sentences satisfied by all structures in $\C$, that is:
    \[ \Th(\C) = \{ \phi \in \mathrm{Sent}(\mathcal{L}) \colon G \models \phi, \text{ for all } G \in \C\}. \]
    We then define the \emph{elementary closure} of $\C$, denoted by $\overline \C$, as the class of models of $\Th(\C)$. 
\end{definition}

It is clear that any elementary (i.e. first-order definable) property of the class is preserved by taking the elementary closure. For instance, this principle applies to stability, as explained in the lemma below. In the following, a structure $\str M$ is called edge-stable if the class $\{\str M\}$ is edge-stable. 

\begin{lemma}\label{lem:stab}
    Let $\CC$ be an edge-stable graph class. Then every $\str M$ the elementary closure $\overline{\CC}$ is edge-stable.
\end{lemma}
\begin{proof}
    Since $\C$ is edge-stable, there exists some $k \in \N$ such that there is no graph $G \in \C$ and elements $(a_i)_{i \in [k]}, (b_j)_{j \in [k]}$ satisfying $G \models E(\bar a_i,\bar b_j) \iff i \leq j$. Letting $\psi$ be the formula 
    \[\exists x_1 \dots x_k \exists y_1 \dots y_k \left[\bigwedge_{i\leq j} E(x_i,y_j) \land \bigwedge_{i > j} \neg E(x_i,y_j)\right]\]
    we see that all $G \in \C$ satisfy $\neg \psi$. Consequently, $\neg \psi \in \Th(\C)$, and since $\str M \models \Th(\C)$ this implies that $\str M \models \neg \psi$. This trivially implies that $\str M$ is edge-stable. 
\end{proof}

We additionally extend the definition of rocket-patterns from classes of graphs to individual graphs. 

\begin{definition}
    We say that a graph $\str M$ \emph{admits rocket-patterns} if there are $\rho,k \in \N$ such that for every $n \in \N$
    we have that $\str M$ admits an $(n,\rho,k)$-rocket-pattern in the sense of \Cref{def:patterns}. Conversely, $\str M$ is \emph{rocket-pattern-free}, if does not admit rocket-patterns.
\end{definition}

Clearly, graphs admitting rocket-patterns will necessarily be infinite. We next illustrate that admitting rocket-patterns is an elementary property of a structure/class of structures.  

\begin{lemma}\label{lem:elemrockpat}
    Let $\C$ be a rocket-pattern-free graph class
    and $\str M$ be in the elementary closure of $\C$.
    Then $\str M$ is rocket-pattern-free.
\end{lemma}
\begin{proof}
    Suppose that $\str M$ is not rocket-pattern-free. Then there are $\rho,k \in \N$ such that for all $n \in \N$, $\str M$ contains an $(n,\rho,k)$-rocket-pattern. On the other hand, since $\C$ is rocket-pattern-free there is a bound $n_0 \in \N$ such that no $G \in \C$ has an $(n_0,\rho,k)$-rocket-pattern. It follows by the assumption that there is a $k$-flip of $\str M$ and pairwise disjoint sets $A,B_1,\dots,B_{n_0}$ with $|A|=n_0$ satisfying the assumptions of \Cref{def:patterns}. Consider the finite subgraph $G$ of $\str M$ induced on $A \cup B_1 \cup \dots B_{n_0}$, and let $\phi$ be the existential formula specifying that there is an injective strong homomorphism from $G$. Since $\str M \models \Th(\C) \cup \{\phi\}$, it follows that there is some $A \in \C$ such that $A \models \phi$; if not, then every $A \in \C$ would satisfy $\neg \phi$ and consequently $\neg \phi$ would be in $\Th(\C)$, contradicting that $\str M \models \Th(\C) \cup \{\phi\}$. Since $A \models \phi$, it follows that $G$ is isomorphic to an induced subgraph of $A$. In particular, $A$ contains an $(n_0,\rho,k)$-rocket-pattern, contradicting the assumption on $\C$.
\end{proof}

Rocket-patterns were identified as a special type of obstruction to what \cite{flippergame} refers to as \emph{pattern-free} classes, i.e. those classes of graphs that, for every $r\geq 1$, do not quantifier-free transduce the class of $r$-subdivided cliques. In fact, the following is established:

\begin{theorem}[Theorem 1.4 of \cite{flippergame}]\label{thm:main-flipper}
    The following are equivalent for a class $\C$ of graphs:
    \begin{enumerate}
        \item \label{it:ms} $\CC$ is monadically stable.
        \item \label{it:pf} $\C$ is edge-stable and pattern-free.
        \item \label{it:sep} 
        For every $r\in \N$, every $\str M$ in the elementary closure of $\C$ is $r$-separable.
%
    \end{enumerate}
\end{theorem}

The exact definition of being $r$-separable is immaterial for our purposes here, as we will use this property as a black-box. We refer to \cite{flippergame} for details.

The proof of the implication $\ref{it:pf} \implies \ref{it:sep}$ is established in~\cite{flippergame} by also extending the definition of pattern-freeness to infinite graphs. In the following, $\str M$ being pattern-free means that the class of finite induced subgraphs of $\str M$ is pattern-free.

\begin{proposition}[Theorem 7.1 of \cite{flippergame}]\label{fact:patternfree}
    Let $\str M$ be an edge-stable pattern-free graph. Then $\str M$ is $r$-separable for every $r\in\N$.
\end{proposition}

The proof of \Cref{fact:patternfree} proceeds by induction on $r$. The base case of $r=1$ simply requires that $\str M$ is edge-stable. The inductive step instrumentally relies on a lemma (Lemma 7.10 in \cite{flippergame}) which shows that if $\str M$ is pattern-free and edge-stable, then the set of types over $\str M$ realized within a radius-$r$ ball contained in an elementary extension of $\str M$ is finite. This is shown by contradiction. More precisely, it is established that if this set of types is infinite, then $\str M$ admits rocket-patterns, which in turn implies that $\str M$ is not pattern-free. Hence, by bypassing pattern-freeness and working with rocket-patterns directly, we may obtain the following fact through the same induction argument.

  \begin{proposition}\label{lem:rseparable}
    Let $\str M$ be an edge-stable rocket-pattern-free graph. Then $\str M$ is $r$-separable, for every $r\in\N$.
\end{proposition}

With this, \Cref{flipper} becomes an easy consequence of what was previously established.

\begin{proof}[Proof of \Cref{flipper}]
    Let $\C$ be a monadically stable class of graphs. Trivially, $\C$ is also edge-stable. Moreover, $\C$ must be rocket-pattern-free: if not, then \Cref{lem:rocket} implies that there is some $r\geq 1$ such that $\C$ transduces either the class of all $r$-subdivided cliques or the class of all $r$-webs, while neither is monadically stable. 

    Conversely, assume that $\C$ is edge stable and rocket-pattern-free. Let $\str M \models \Th(\C)$ be a graph from the elementary closure of $\C$. It follows by \Cref{lem:stab} and \Cref{lem:elemrockpat} that $\str M$ is edge-stable and rocket-pattern-free. Consequently, \Cref{lem:rseparable} implies that $M$ is $r$-separable for every $r \in \N$. Since this holds for every such $\str M$,  \Cref{thm:main-flipper} implies that $\C$ is monadically stable.
\end{proof}

\section{Model checking algorithm}\label{apx:mc}

In this appendix we improve the running time of the fpt model checking algorithm of~\cite{ssmc}
as stated in the following.

\introtractability*

We consider the algorithm presented presented in Section 5 of the full version of \cite{ssmc}.
Let us go through the necessary changes to Sections 5.1, 5.2, and 5.4 to prove \Cref{thm:intro-tractability}.
There is only a single position where neighborhood covers are computed: 
In the paragraph \emph{Neighborhood Cover Computation} in the proof of \cite[Theorem 7]{ssmc},
a monadically stable graph class \(\Cc\) is fixed, and
the algorithm of \cite[Theorem 10]{ssmc} is run on a graph $G \in \CC_\ell$, 
where $\CC_\ell$ is the hereditary closure of applying at most \(\ell\) flips\footnote{\cite{ssmc} defines flips slightly differently.}
to graphs from $\CC$; so $\CC_\ell$ is again monadically stable.
This yields a distance-\(2^q\) neighborhood cover with diameter \(\sigma(2^q)\) for some function~\(\sigma\)
and overlap bounded by $\Oh_{\Cc,q,\ell,\eps}(n^\eps)$; here, $q$ is the quantifier rank of $\varphi$.
Note that~\cite{ssmc} presents this in slightly different terminology. In particular, the diameter of a neighborhood cover is called \emph{spread}.

We can substitute the call to \cite[Theorem 10]{ssmc} by our \Cref{thm:cover-main} to compute in time 
$\Oh_{\Cc,q,\ell,\eps}(n^{4+\eps})$ a distance-\(2^q\) neighborhood cover of $G$ with diameter at most $\sigma(2^q) \coloneqq 4 \cdot 2^q$
and whose overlap can be also bounded by $\Oh_{\Cc,q,\ell,\eps}(n^\eps)$.
Thus, the neighborhood cover computed by our routine satisfies the same guarantees,
and we can use the faster routine instead.

To propagate these changes,
essentially, every occurrence of a running time bound of \(\Oh(|V(G)|^{9.8})\) in Sections 5.1, 5.2, and 5.4
is replaced by a bound \(\Oh_{\Cc,q,\ell,\epsilon}(|V(G)|^{4+\epsilon})\).
This is possible as \(q,\ell,\epsilon > 0\) and a (sometimes implicit) monadically stable graph class \(\Cc\)
are appropriately quantified at all positions of substitution.
In particular, the time complexity of an \emph{efficient \(\textnormal{MC}(\Cc,q,\ell,c)\)-algorithm} (\cite[Definition~4]{ssmc})
is now reduced to
\[
    f_\textnormal{MC}(q,\ell,c,\epsilon) \cdot |V(G)|^{(1+\epsilon)^{d}} \cdot |V(G_0)|^{4+\epsilon}.
\]

After propagating these substitutions through \cite[Theorem 6, Theorem 7, Lemma 10, Theorem~9]{ssmc},
finally, in the proof of \cite[Theorem~5]{ssmc}, the running time is for every \(\epsilon' > 0\) bounded by
\[
    f_\textnormal{MC}(q,0,|\Sigma|,\epsilon') \cdot |V(G)|^{(1+\epsilon')^{\textnormal{game-depth}(\Cc,\rho)}} \cdot |V(G)|^{4+\epsilon'}.
\]
By choosing \(\epsilon' \coloneqq (1+\epsilon/2)^{1/\textnormal{game-depth}(\Cc,\rho)}-1 > 0\), we get a running time of
\[
    f_\textnormal{MC}(q,0,|\Sigma|,\epsilon') \cdot |V(G)|^{1+\epsilon/2} \cdot |V(G)|^{4+\epsilon/2} \le f_\textnormal{MC}(q,0,|\Sigma|,\epsilon') \cdot |V(G)|^{5+\epsilon}.
\]
This in turn can be bounded by \(\Oh_{\Cc,\eps,\phi}(n^{5+\epsilon})\),
proving \Cref{thm:intro-tractability}.

%% file: chapters/tractability/appendix-sparse-covers.tex
\newcommand{\Zz}{\mathcal{Z}}
\newcommand{\Yy}{\mathcal{Y}}

\section{Almost nowhere dense neighborhood covers}\label{app:almost-nd-covers}

In this appendix we prove that the incidence graphs of the distance-$r$ neighborhood covers constructed in \cref{thm:cover-main} form an almost nowhere dense class. To formulate this result, we need first to establish the relevant terminology.

\subsection{More preliminaries}

\paragraph*{Shallow minors.}
A graph $H$ is a {\em{depth-$r$ minor}} of a graph $G$ if there exists {\em{depth-$r$ minor model}} of $H$ in $G$, that is, a mapping $\eta$ that maps vertices of $H$ to pairwise disjoint subgraphs of $G$ such that:
\begin{itemize}
 \item for each vertex $u$ of $H$, $\eta(u)$ is connected and of radius at most $r$; and
 \item for each edge $uv$ of $H$, in $G$ there is an edge with one endpoint in $\eta(u)$ and the other in $\eta(v)$.
\end{itemize}
Next, $H$ is a {\em{depth-$r$ topological minor}} of $G$ if there exists a {\em{depth-$r$ topological minor model}} of $H$ in $G$, that is, a mapping $\pi$ that maps vertices of $H$ to pairwise different vertices of $G$ and edges of $H$ to paths in $G$ so~that:
\begin{itemize}
 \item for each edge $uv$ of $H$, $\pi(uv)$ is a path of length at most $2r+1$ with endpoints $\pi(u)$ and $\pi(v)$; and
 \item paths $\{\pi(uv)\colon uv\in E(H)\}$ are pairwise disjoint apart from possibly sharing endpoints.
\end{itemize}
It is easy to see that if $H$ is a depth-$r$ topological minor of $G$, then it is also a depth-$r$ minor of $G$. The converse implication does not hold, but it is known that sparsity of shallow topological minors entails sparsity of shallow minors in the following sense. Here, the {\em{edge density}} of a graph $H$ is the quantity~$\frac{|E(H)|}{|V(H)|}$.

\begin{theorem}[see Theorem~3.9 of~\cite{dvorak-thesis}]\label{thm:topminor-equiv}
 For every $r\in \N$ there exists a polynomial $Q(\cdot)$ such that the following holds. Let $G$ be a graph and let $\ell$ be the maximum edge density among the depth-$r$ topological minors of $G$. Then every depth-$r$ minor of $G$ has edge density bounded by $Q(\ell)$.
\end{theorem}

Next, we define the central sparsity notion of this section.

\begin{definition}
 A class of graphs $\Cc$ is {\em{almost nowhere dense}} if for every graph $H$ that is a depth-$r$ minor of a graph $G\in \Cc$, the edge density of $H$ can be bounded by $\Oh_{\Cc,r,\eps}(|V(G)|^{\eps})$, for every $\eps>0$.
\end{definition}

Dvo\v{r}\'ak~\cite[Theorem 3.15]{dvorak-thesis} proved that every nowhere dense class is almost nowhere dense (see also~\cite[Chapter~1, Theorem~3.1]{sparsityNotes}). While formally almost nowhere denseness is a weaker property than nowhere denseness, it still entails a number of important properties of quantitative nature observed in nowhere dense classes, for instance subpolynomial bounds on generalized coloring numbers or centered colorings.
Moreover, every hereditary almost nowhere dense graph class is nowhere dense.
See~\cite{bushes} for a broader discussion.

Note that by \cref{thm:topminor-equiv}, almost nowhere denseness can be equivalently defined in terms of edge density of depth-$r$ topological minors, instead of depth-$r$~minors.

\paragraph*{More on covers.}
For a neighborhood cover $\Cover$ of a graph $G$, we define $\Inc(\Cover)$ to be the {\em{incidence graph}} of $\Cover$: $\Inc(\Cover)$ is a bipartite graph with $V(G)$ on one side and $\Cover$ on the second side, where $u\in V(G)$ is adjacent to $C\in \Cover$ if and only if $u\in C$.

We now formalize the structural properties of the neighborhood covers constructed in the proof of \cref{thm:cover-main} that we will use.
Call a subset of vertices $X$ in a graph $G$ {\em{compact}} if there exists a vertex $u$ of $G$ such that $X\subseteq N[u]$. A partition $\Zz$ of the vertex set of $G$ is called {\em{compact}} if every element of $\Zz$ is compact. Observe that if for a compact partition $\Zz$ we define
\[\Cover(\Zz)\coloneqq \{N[Z]\colon Z\in \Zz\},\]
then $\Cover(\Zz)$ is a distance-$1$ neighborhood cover of $G$ of diameter at most $4$. Neighborhood covers of this form, i.e., $\Cover(\Zz)$ for a compact partition $\Zz$, will be called {\em{compact}}. Note that the overlap of $\Cover(\Zz)$ is
\[\max_{v\in V(G)} |\{Z\in \Zz~|~v\in N[Z]\}|.\]

We observe that the distance-$1$ neighborhood covers constructed in the proof of \cref{thm:cover-main} are compact.
This is the only structural property of the covers from \cref{thm:cover-main} that we will use.

\subsection{The main result and its motivation}\label{app:sparseCoversMain}
With all the notation in place, we may now state formally the main result of this section.

\begin{restatable}{theorem}{almostNdCovers}\label{thm:almost-nd-covers}
 Let $\Cc$ be a monadically stable class of graphs and $r\in \N$ be fixed. For every graph $G\in \Cc$, let $\Cover_{G,r}$ be the distance-$r$ neighborhood cover of $G$ constructed as in the proof of \cref{thm:cover-main}. Then the class of graphs $\{\Inc(\Cover_{G,r})\colon G\in \Cc\}$ is almost nowhere dense.
\end{restatable}

We note that the first step of the proof of \cref{thm:cover-main} was to use \cref{lem:distanceCover} to reduce to the case $r=1$ by considering the monadically stable class $\C^r$. Therefore, when arguing \cref{thm:almost-nd-covers} we may also restrict attention to the case $r=1$. Then the distance-$1$ neighborhood covers $\Cover_G\coloneqq \Cover_{G,1}$ are compact.

For brevity we follow the following convention in the remainder of this section: whenever we speak about neighborhood covers, we mean distance-$1$ neighborhood covers.

\medskip

Before we proceed to the proof of \cref{thm:almost-nd-covers}, let us briefly discuss its significance in a broader perspective.
The so-called {\em{Sparsification Conjecture}}~\cite{POM21} postulates that every monadically stable class of graphs is actually {\em{structurally nowhere dense}}: it can be transduced from a nowhere dense class. We believe that \cref{thm:almost-nd-covers} may be a first step towards a proof of the following relaxed variant of this conjecture:

\begin{conjecture}\label{conj:relaxed-sparsification}
For every monadically stable class of graphs $\Cc$, there is an almost nowhere dense class $\Bb$ of quasi-bushes that encode graphs from $\Cc$.
\end{conjecture}
                                                                                                                                                                                                                                                                                                                                                          Here, {\em{quasi-bushes}} are a specific form of a decomposition for dense graphs introduced by Dreier et al.~\cite{bushes}. Intutitively speaking, a quasi-bush encoding a graph $G$ is a structure $B$ that is essentially a bounded-depth tree with a sparse network of pointers, which encodes the edge relation in $G$ in the sense that $G$ can be transduced from $B$ by a specific, very simple trasduction. If the class $\Bb$ in \cref{conj:relaxed-sparsification} was nowhere dense, instead of almost nowhere dense, then this would settle the Sparsification Conjecture in full generality. Unfortunately,
with our current toolbox we do not see a way of achieving this, but settling \cref{conj:relaxed-sparsification} would still imply a number of significant corollaries for monadically stable classes of graphs. For instance, \cref{conj:relaxed-sparsification} was confirmed in~\cite{bushes} for structurally nowhere dense classes, and from this it followed that such classes have so-called {\em{low shrubdepth coverings}} of size $\Oh_{\Cc,\eps}(n^\eps)$, for any $\eps>0$. Proving \cref{conj:relaxed-sparsification} would yield the same conclusion for all monadically stable classes.

%

Let us now elaborate on a possible approach to proving \cref{conj:relaxed-sparsification} using the Flipper game~\cite{flippergame} combined with our construction of neighborhood covers.
Suppose $\Cc$ is a monadically stable class of graphs and $G\in \Cc$ is an $n$-vertex graph that we would like to transduce from some almost nowhere dense graph. We consider playing the Flipper game on $G$ for radius $r=4$. On the one hand, the results of~\cite{flippergame} provide us with a strategy for Flipper that guarantees termination of the game within a constant number of rounds. On the other hand, for the moves of Connector, we restrict those moves to playing elements of the neighborhood covers provided by \cref{thm:cover-main}. It is not hard to see that in this way, the whole game tree has constant depth and total size bounded by $\Oh_{\Cc,\eps}(n^{1+\eps})$. Moreover, if we supply the game tree with suitable pointers, then the original graph $G$ can be transduced from the obtained quasi-bush~$B$.  \cref{thm:almost-nd-covers} ensures that at least at every step of the construction of $B$, the graph of interaction between the vertices of $G$ and the elements of the cover is already almost nowhere dense. This gives hope that the whole $B$ is actually almost nowhere dense as well, which would imply the desired result. There are, however, multiple technical details to be fixed in this outline, so we leave further investigation to future~work.

\subsection{Transducing a dense subdivision}

The key step in the proof of \cref{thm:almost-nd-covers} is encapsulated in the following lemma, which intuitively says the following: if the incidence graph of a compact neighborhood cover $\Cover$ contains a topological minor of a dense graph, but the neighborhood cover itself has small overlap, then from $\Inc(\Cover)$ one can transduce a subdivision of a dense graph.

\begin{lemma}\label{lem:transduce-dense}
 For every $r\in \N$ there exists a polynomial $P(\cdot)$ and a transduction $\Tc$ (with copying) with the following property. Suppose $G$ is a graph and $\Cover$ is a compact neighborhood cover of $G$ with overlap $k$. Suppose further that $\Inc(\Cover)$ contains a graph of edge density $\ell$ as a depth-$r$ topological minor. Then $\Tc(G)$ contains a $1$-subdivision of a graph of edge density at least $\frac{\ell}{P(k)}-1$.
\end{lemma}

We remark that in \cref{lem:transduce-dense} we consider transductions {\em{with copying}}: The transduction is not applied to the original graph $G$, but to a graph $G^{\bullet c}$ where every vertex is replaced by $c$ non-adjacent copies, for some constant $c$. (Copies of the same vertex are bound by an additional binary predicate.) It is well-known and easy to prove that the the notions of monadic stability and monadic dependence are the same when considering transductions with copying, hence we stick to transductions with copying throughout this~section.

The remainder of this section is devoted to the proof of \cref{lem:transduce-dense}. Let $G$ and $\Cover$ be the given graph and its neighborhood cover, respectively. For brevity, we denote $V\coloneqq V(G)$. Since $\Cover$ is assumed to be compact, there is a compact partition $\Zz$ of $V$ such that $\Cover=\{N[Z]\colon Z\in \Zz\}$. For convenience, we reinterpret $\Inc(\Cover)$ as a bipartite graph with sides $V$ and $\Zz$, where $u\in V$ is adjacent to $Z\in \Zz$ if and only if $u\in N[Z]$. By assumption, the maximum degree among $V$ in $\Inc(\Cover)$ is equal to $k$.

Further, let $\pi$ be the assumed topological minor model of a graph $H$ with edge density $\ell$.
First, we show that we can regularize the model while not losing much on the edge density.
The statement of the lemma is trivial for \(\ell \le P(k)\). By making sure that $P(k) \ge k$, we may therefore assume $\ell > k$.

\begin{claim}
 There is a subgraph $H'$ of $H$ satisfying the following.
 \begin{itemize}
  \item $H'$ is bipartite and has edge density at least $\ell'\coloneqq \frac{\ell}{2r}$;
  \item $\pi(V(H'))\subseteq \Zz$;
  \item all paths in $\pi(E(H'))$ have the same length $2p$, where $p$ is an integer with $1\leq p\leq r$.
 \end{itemize}
\end{claim}
\begin{claimproof}
 We describe the construction of $H'$ as a process that starts with $H'\coloneqq H$ and discards edges and vertices until all the requested properties are achieved.

 First, as long as $H'$ has a vertex of degree at most $k$, remove this vertex. Note that since $\ell$, the edge density of $H$, is bounded below by $k$, applying this operation cannot decrease the edge density. So at the end, we obtain a subgraph $H'$ of edge density at least $\ell$ such that all vertices have degree at least $k+1$. Since the maximum degree in $\Inc(\Cover)$ among the vertices of $V$ is bounded above by $k$, it follows at this point that $\pi(V(H'))\subseteq \Zz$.

 Next, it is well-known that every graph contains a bipartite subgraph with at least half of its edges. By applying this fact to $H'$, we may remove at most half of the edges of $H'$ to make it a bipartite graph. Thus, at this point $H'$ is in addition bipartite and has edge density at least $\frac{\ell}{2}$.

 Finally, since $\Inc(\Cover)$ is bipartite and $\pi(V(H'))\subseteq \Zz$, the length of every path of $\pi(E(H'))$ is even and between $2$ and $2r$. By the pigeonhole principle, there exists an integer $p$ with $1\leq p\leq r$ such that at least $\frac{1}{r}$-fraction of edges of $H'$ are mapped by $\pi$ to paths of length exactly $2p$. Restricting the edge set of $H'$ to those edges yields a subgraph $H'$ with all the desired properties.
\end{claimproof}

From now on we will work on the (still suitably dense) graph $H'$ and its topological minor model $\pi|_{H'}$.
For brevity, denote $F\coloneqq E(H')$. Also, let $A$ and $B$ be the two sides in a bipartition of $H'$.

Since all paths of $\pi(F)$ have length $2p$ and are pairwise disjoint except for sharing endpoints, we may choose disjoint subsets $\Xx_0,\Xx_1,\ldots,\Xx_{p-1},\Xx_p$ of $\Zz$, with $\Xx_0=\pi(A)$ and $\Xx_p=\pi(B)$, and disjoint subsets $Y_1,Y_2,\ldots,Y_p$ of $V$ such that every path of $\pi(F)$ consecutively traverses vertices belonging to sets
\[\Xx_0, Y_1, \Xx_1, Y_2, \Xx_2,\ldots, Y_{p-1},\Xx_{p-1},Y_p,\Xx_p,\]
in this order.

We next clean the paths of $\pi(F)$. This cleaning proceeds in two steps. In the first step, we clean certain conflicts between the elements of $\bigcup_{i=0}^p \Xx_i$, defined as~follows.

For every set $Z\in \Zz$, arbitrarily choose the {\em{center}} of $Z$ to be any vertex $\gamma(Z)\in V$ such that $Z\subseteq N[\gamma(Z)]$. Clearly, such choice exists for every $Z \in \Zz$ as $\Zz$ is compact. Call two sets $Z,Z'\in \Zz$ {\em{in conflict}} if $\gamma(Z)\in N[Z']$ or $\gamma(Z')\in N[Z]$. A subset of $\Zz$ is {\em{conflict-free}} if its elements are pairwise not in conflict.

\begin{claim}\label{cl:conflict-free}
 There are conflict-free subsets
 $\Xx'_0\subseteq \Xx_0,\Xx'_1\subseteq \Xx_1,\ldots,\Xx'_p\subseteq \Xx_p$
 and a set $F'\subseteq F$ such~that
 \[|F'|\geq \frac{|F|}{(2k-1)^{p+1}}\qquad\textrm{and}\qquad V(P)\cap \Zz\subseteq \bigcup_{i=0}^p \Xx'_i\quad \textrm{for every path }P\in \pi(F').\]
\end{claim}
\begin{claimproof}
 Let $D$ be the directed conflict graph, i.e., the directed graph on vertex set $\Zz$ such that for all distinct $Z,Z'\in \Zz$, there is a directed edge $(Z,Z')$ in $D$ if and only if $\gamma(Z)\in N[Z']$. By assumption, a fixed center $\gamma(Z)$ may belong to at most $k-1$ among the sets $N[Z']$ for $Z'\in \Zz\setminus \{Z\}$, hence the maximum outdegree in $D$ is at most $k-1$. A standard degree-counting argument now shows that every subgraph of $D$ contains a vertex of total degree at most $2(k-1)$. Proceeding by induction, we obtain that $D$ has a proper coloring with at most $2k-1$ colors. In other words, there is a coloring $\lambda\colon \Zz\to \{1,\ldots,2k-1\}$ such that for all $Z,Z'\in \Zz$ with $\lambda(Z)=\lambda(Z')$, $Z$ and $Z'$ are not in conflict.

 For every edge $e\in F$, let
 \[\lambda(e)\coloneqq \bigl(\lambda(X_0),\lambda(X_1),\ldots,\lambda(X_p)\bigr)\in \{1,\ldots,2k-1\}^{p+1},\]
 where $X_0,X_1,\ldots,X_p$ are the unique vertices traversed by $\pi(e)$ in the sets $\Xx_0,\Xx_1,\ldots,\Xx_p$, respectively.
 By the pigeonhole principle, there exists a tuple $\bar c=(c_0,c_1,\ldots,c_p)\in \{1,\ldots,2k-1\}^{p+1}$ such that $\bar c=\lambda(e)$ for at least $\frac{|F|}{(2k-1)^{p+1}}$ edges $e\in F$. Hence, setting
 \[F'\coloneqq \lambda^{-1}(\bar c)\qquad\textrm{and}\qquad \Xx'_i\coloneqq \Xx_i\cap \lambda^{-1}(c_i)\quad\textrm{for all }i\in \{0,1,\ldots,p\}\]
 satisfies all the claimed properties.
\end{claimproof}

The second step is to remove unwanted adjacencies between different paths.

\begin{claim}\label{cl:unique}
 There are subsets
 $\Xx''_0\subseteq \Xx'_0,\Xx''_1\subseteq \Xx'_1,\ldots,\Xx''_p\subseteq \Xx'_p$
and a set $F''\subseteq F'$ such~that
 \[|F''|\geq \frac{|F'|}{(4k)^{p+3}}\qquad\textrm{and}\qquad V(P)\cap \Zz\subseteq \bigcup_{i=0}^p \Xx''_i\qquad \textrm{for every path }P\in \pi(F''),\]
 and moreover, for every path $P\in \pi(F'')$ and $j\in \{1,\ldots,p\}$, we have the following: if $X_{j-1},y_j,X_j$ are the unique vertices traversed by $P$ in $\Xx''_{j-1},Y_j,\Xx''_j$, respectively, then $X_{j-1}$ is the unique neighbor of $y_j$ in $\Xx''_{j-1}$ and $X_j$ is the unique neighbor of $y_j$ in $\Xx''_j$.
\end{claim}
\begin{claimproof}
 In the following, for a path $P\in \pi(F')$ and $i\in \{0,1,\ldots,p\}$, by $X_i(P)$ we denote the unique element of $V(P)\cap \Xx'_i$. Similarly, for $j\in \{1,\ldots,p\}$, $y_j(P)$ is the unique element of $V(P)\cap Y_j$.

 We first focus on defining $\Xx''_0$. Observe that for every $P\in \pi(F')$, $y_1(P)$ has between $1$ and $k$ neighbors in $\Xx'_0$. Hence, there is $F^\circ\subseteq F'$ such that $|F^\circ|\geq \frac{|F'|}{k}$ and for all $P\in \pi(F^\circ)$, $y_1(P)$ has the same number of neighbors in $\Xx'_0$; let this be some $d\in \{1,\ldots,k\}$.

 Consider the following random experiment. Let $\Yy\subseteq \Xx'_0$ be a subset of $\Xx'_0$ constructed by including every $X\in \Xx'_0$ independently with probability $\frac{1}{d}$. Call an edge $f\in F^\circ$ {\em{good}} with respect to $\Yy$ if $X_0(P)\in \Yy$ and $X_0(P)$ is the unique neighbor of $y_1(P)$ in $\Yy$, where $P=\pi(f)$. Observe that for every fixed $f\in F^\circ$, we have
 \[\Prob[f \textrm{ is good w.r.t. }\Yy]=\frac{1}{d}\cdot \left(1-\frac{1}{d}\right)^{d-1}\geq \frac{1}{ed}\geq \frac{1}{ek};\]
 here, the $\frac{1}{d}$ factor is the probability that $X_0(P)$ is included in $\Yy$ while the $\left(1-\frac{1}{d}\right)^{d-1}$ factor is the probability that all the other neighbors of $y_1(P)$ in $\Xx''_0$ are not included in $\Yy$. By the linearity of expectation, we have
 \[\Exp \bigl[ \left|\{f\in F^\circ~|~f\textrm{ is good w.r.t. }\Yy\}\right| \bigr] \geq \frac{|F^\circ|}{ek}.\]
 It follows that we may choose $\Xx''_0\subseteq \Xx'_0$ and $F_1\subseteq F^\circ$ such that
 \[|F_1|\geq \frac{|F^\circ|}{ek}\geq \frac{|F'|}{ek^2},\]
 and for every $f\in F_1$, $X_0(\pi(f))$ is the unique neighbor of $y_1(\pi(f))$ in $\Xx''_0$.

 By performing an analogous construction, we may find $\Xx''_p\subseteq \Xx'_p$ and $F_2\subseteq F_1$ such that
 \[|F_2|\geq \frac{|F_1|}{ek^2}\geq \frac{|F'|}{e^2 k^4}\geq \frac{|F'|}{8k^4}\]
 and for every $f\in F_2$, we in addition have that $X_p(\pi(f))$ is the unique neighbor of $y_p(\pi(f))$ in $\Xx''_p$.

 We are left with suitably choosing $\Xx''_i$ for $i\in \{1,\ldots,p-1\}$.
 For every $i\in \{1,\ldots,p-1\}$, define the directed conflict graph $D_i$ on vertex set $F_2$ as follows: for distinct $f,f'\in F_2$ we put a directed edge $(f,f')$ in $D_i$ if and only if $y_i(P)\in N[X_i(P')]$ or $y_{i+1}(P)\in N[X_i(P')]$, where $P=\pi(f)$ and $P'=\pi(f')$. Again, both $y_i(P)$ and $y_{i+1}(P)$ can belong to at most $k-1$ sets of the form $N[X_i(P')]$ for $P'\in \pi(F_2)$, $P'\neq P$, hence the outdegrees in $D_i$ are bounded by $2k-2$. As in the proof of \cref{cl:conflict-free}, we conclude that $D_i$ can be colored with $4k-3$ colors, that is, there is a coloring $\lambda_i\colon F_2\to \{1,\ldots,4k-3\}$ such that for all $f,f'\in F_2$ with $\lambda_i(f)=\lambda_i(f')$, we have $y_i(\pi(f)),y_{i+1}(\pi(f))\notin N[X_i(\pi(f'))]$ and $y_i(\pi(f')),y_{i+1}(\pi(f'))\notin N[X_i(\pi(f))]$.

 Similarly as in the proof of \cref{cl:conflict-free}, for each $f\in F_2$ define
 \[\lambda(f)\coloneqq \bigl(\lambda_1(f),\lambda_2(f),\ldots,\lambda_{p-1}(f)\bigr)\in \{1,\ldots,4k-3\}^{p-1},\]
 and observe that there exists $\bar c=(c_1,\ldots,c_{p-1})\in \{1,\ldots,4k-3\}^{p-1}$ such that $\bar c=\lambda(f)$ for at least $\frac{|F_2|}{(4k-3)^{p-1}}$ edges of $F_2$. We may now set
 \[F''\coloneqq \lambda^{-1}(\bar c)\qquad\textrm{and}\qquad \Xx''_i\coloneqq \Xx_i'\cap \bigcup_{f\in F''} V(\pi(f))\quad \textrm{for all }i\in \{1,\ldots,p-1\}.\]
 It follows directly from the construction that $F''$ and $\Xx''_0,\Xx''_1,\ldots,\Xx''_p$ satisfy the required property about the uniqueness of neighbors. To conclude the proof, note that
 \[|F''|\geq \frac{|F_2|}{(4k-3)^{p-1}}\geq \frac{|F'|}{(4k-3)^{p-1} \cdot 8k^{4}}\geq \frac{|F'|}{(4k)^{p+3}}.\qedhere\]
\end{claimproof}

With the cleaning steps performed, we may conclude the proof of \cref{lem:transduce-dense} by presenting the transduction $\Tc$.

Let $H''$ be the subgraph of $H'$ obtained by restricting the edge set to $F''$. Then
\[\frac{|E(H'')|}{|V(H'')|}=\frac{|F''|}{|V(H')|}\geq \frac{1}{|V(H')|}\cdot \frac{|F'|}{(4k)^{p+3}}\geq \frac{1}{|V(H')|}\cdot\frac{|F|}{(2k-1)^{p+1}(4k)^{p+3}}\geq  \frac{\ell}{2r(4k)^{2r+4}}.\]
So we set
\[P(k)\coloneqq 2r(4k)^{2r+4}.\] Further, let $H^\bullet$ be the $1$-subdivision of $H''$. It now remains to construct a transduction $\Tc$ with the property that $H^\bullet\in \Tc(G)$. We explain $\Tc$ by presenting how the transduction constructs $H^\bullet$ through an existentially guessed precoloring and interpretation; the formalization of this procedure in the language of transductions is straightforward and left to the reader.

Let $H^\star$ be the $(2p-1)$-subdivision of $H''$. It is easy to see that $H^\bullet$ can be transduced from $H^\star$, so it suffices to show how to transduce $H^\star$ from $G$.
Note that $\pi$ naturally induces a subgraph embedding of $H^\star$ in $\Inc(\Cover)$, so we identify $H^\star$ with this subgraph of $\Inc(\Cover)$. By restricting sets $\Xx''_i$ and $Y_j$ if necessary, we may assume that the vertex set of $H^\star$ is $\bigcup_{i=0}^p \Xx''_i\cup \bigcup_{j=1}^p Y_j$; in other words, every element of the sets $\Xx''_i$ and $Y_j$ participates in~$H^\star$.  For $i\in \{0,1,\ldots,p\}$, we define the following subsets of $V$:
\[\Gamma_i\coloneqq \{\gamma(Z)\colon Z\in \Xx''_i\}\qquad \textrm{and} \qquad R_i\coloneqq \bigcup_{Z\in \Xx''_i} Z.\]

With this notation in place, we explain the consecutive steps taken by $\Tc$.
\begin{itemize}
 \item Guess $p\in \{1,\ldots,r\}$.
 \item Guess sets $\Gamma_i$ and $R_i$, for $i\in \{0,1,\ldots,p\}$.
 \item Guess sets $Y_j$, for $j\in \{1,\ldots,p\}$.
 \item For every vertex $w$ of $H^\star$, interpret $w$ in a vertex of $G$ as follows:
 \begin{itemize}
 \item If $w\in Y_j$ for some $j\in \{1,\ldots,p\}$, interpret $w$ in itself.
 \item If $w=Z$ for some $Z\in \Xx''_i$, $i\in \{0,1,\ldots,p\}$, interpret $w$ in $\gamma(Z)\in \Gamma_i$.
 \end{itemize}
 Note that since sets $\Xx''_i$ are conflict-free, we have $\gamma(Z)\neq \gamma(Z')$ for all $Z,Z'$ belonging to the same set $\Xx''_i$. However, the center sets $\Gamma_i$ may overlap for different $i$, and they may also overlap with sets $Y_j$ (which, recall, are pairwise disjoint). Therefore, in the scheme above, one vertex of $G$ can be used to interpret up to $p+2$ vertices of $H^\star$. So formally the transduction $\Tc$ copies the vertex set of $G$ $p+2$ times, and uses a different layer of copies for each of the sets $\bigcup_{j=1}^p Y_j,\Xx''_0,\Xx''_1,\ldots,\Xx''_p$.
 \item Interpret the edge relation of $H^\star$ as follows: for every $j\in \{1,\ldots,p\}$, every vertex $y\in Y_j$ should be made adjacent to a set $Z\in \Xx''_j$ if and only if the set $N[y]\cap N[\gamma(Z)]$ contains a vertex of $R_j$; and similarly $y$ should be made adjacent to a set $Z\in \Xx''_{j-1}$ if and only if the set $N[y]\cap N[\gamma(Z)]$ contains a vertex of $R_{j-1}$. That this condition correctly interprets the edge relation of $H^\star$ follows from the properties achieved in \cref{cl:conflict-free,cl:unique}: sets $\Xx_i''$ are conflict-free and each $y\in Y_j$ has unique neighbors in $\Xx_{j-1}''$ and $\Xx_j''$.
\end{itemize}
This concludes the definition of the transduction $\Tc$ and the proof of \cref{lem:transduce-dense}.

\subsection{Wrapping up the proof}

We may now use \cref{lem:transduce-dense} to finish the proof of \cref{thm:almost-nd-covers}, which we restate for convenience.

\almostNdCovers*
\begin{proof}
As discussed in the beginning of \Cref{app:sparseCoversMain}, it is sufficient to consider the case \(r=1\).
 Fix $d\in \N$. By \cref{thm:topminor-equiv}, it suffices to prove the following: If $H$ is a depth-$d$ topological minor of $\Inc(\Cover_G)$, where $\Cover_G$ is the neighborhood cover constructed as in the proof of \cref{thm:cover-main} for an $n$-vertex graph $G\in \Cc$, then the edge density of $H$ is bounded by $\Oh_{\Cc,d,\eps}(n^\eps)$, for every $\eps>0$.
 Let us then fix $G$, $\Cover_G$, $H$, and $\eps$ as above. Let $\ell$ be the edge density of $H$; we want to prove that $\ell\leq \Oh_{\Cc,d,\eps}(n^\eps)$.

 Let $P(\cdot)$ and $\Tc$ be the polynomial and the transduction provided by \cref{lem:transduce-dense} for depth $d$. Choose $\delta>0$ small enough so that $P(n^\delta)\leq \Oh(n^{\eps/2})$. Denoting by $k$ the overlap of $\Cover_G$, we have by \Cref{thm:cover-main},
 \[k\leq \Oh_{\Cc,\delta}\left(n^\delta\right)=\Oh_{\Cc,d,\eps}\left(n^\delta\right).\]

 Recalling that $\Cover_G$ is compact, we conclude from \cref{lem:transduce-dense} that $\Tc(G)$ contains a $1$-subdivision of a graph $J$ of edge density at least
 \[\frac{\ell}{P(k)}-1\geq \Omega_{\Cc,d,\eps}\left(\frac{\ell}{n^{\eps/2}}\right).\]

 Let the $1$-subdivision of $J$ be $J^\bullet$.
 Let $\Dd$ be the class of all graphs in $\Tc(\Cc)$ that do not contain $K_{3,3}$ as a subgraph. Since the $1$-subdivision of any graph has this property, we conclude that $J^\bullet\in \Dd$. On the other hand, $\Dd$ can be transduced from $\Cc$ by definition, hence $\Dd$ is monadically stable. As all graphs in $\Dd$ exclude $K_{3,3}$ as a subgraph, it follows from \cref{thm:dvorak} that $\Dd$ is actually nowhere dense. Now it follows from a result of Dvo\v{r}\'ak~\cite[Theorem~3.15]{dvorak-thesis} that the edge densities of bounded-depth minors of graphs from $\Dd$ are bounded subpolynomially. In particular, since $J$ is a depth-$1$ minor of $J^\bullet\in \Dd$, the edge density of $J$ is upper bounded by $\Oh_{\Dd,\eps'}(|V(J)|^{\eps'})\leq \Oh_{\Cc,d,\eps'}(n^{\eps'})$, for every $\eps'>0$.

 By selecting $\eps'=\eps/2$ and comparing the established lower and upper bounds on the edge density of~$J$, we conclude that
 \[\Omega_{\Cc,d,\eps}\left(\frac{\ell}{n^{\eps/2}}\right) \leq \Oh_{\Cc,d,\eps}\left(n^{\eps/2}\right).\]
 This implies that
 \[\ell\leq \Oh_{\Cc,d,\eps}\left(n^\eps\right),\]
 as required.
\end{proof}

%% file: chapters/hardness/sop.tex
\section{Collapse of the Strict Order Property}\label{sec:sop}

NIP offers a natural generalization to stability that extends much of its tameness while permitting certain orders. One may then inquire what is the qualitative bridge between these two notions. In his seminal work on classification theory \cite{shelah1990classification}, Shelah defined NSOP (Not the Strict Order Property) as a dividing line extending stability that plays precisely this role. Despite its somewhat more artificial nature, NSOP has since been shown to capture a wide class of interesting and model-theoretically tame combinatorial structures. These include the \emph{Rado graph} \cite{rado}, the \emph{Henson graphs} \cite{henson}, as well as all stable graphs. It is therefore meaningful to study to what extent it gives rise to interesting combinatorics in the context of classes of finite structures. Our results in this section are negative in the sense that we show that, while there are interesting examples of non-hereditary unstable classes of graphs without the strict order property, stability and NSOP collapse for hereditary classes. 

We now turn to formal definitions. While we have so far resorted to definitions of monadic NIP and monadic stability in terms of $1$-dimensional transductions, we here recall the original definitions of Shelah, relativized to graph classes. 

\begin{definition}
Let $\C$ a class of graphs and $\phi(\bar x,\bar y)$ a first-order formula. We say that $\phi$ has the:
\begin{enumerate}
    \item The \emph{Order Property} over $\C$ if for all $n\in\N$ there is some $G_n \in \C$ and sequences $(\bar a_i)_{i \in [n]}$  and $(\bar b_i)_{i \in [n]}$ of tuples from $G_n$ such that:
    \[
    G_n\models \phi(\bar a_i,\bar b_j) \text{ if and only if } i<j.
    \]
    \item The \emph{Independence Property} over $\C$ if for all $n\in\N$ there is some $G_n\in\C$ and sequences $(\bar a_i)_{i \in [n]}$  and $(\bar b_I)_{I\subseteq [n]}$ of tuples from $G_n$ such that:
    \[
    G_n\models \phi(\bar a_i,\bar b_I) \text{ if and only if } i\in I.
    \]
    \item The \emph{Strict Order Property} over $\C$ if for all $n\in\N$ there is some $G_n\in\C$ and a sequence $(\bar a_i)_{i<n}$ of tuples from $G_n$ such that:
    \[
    G_n\models (\exists\bar{x})(\lnot\phi(\bar{x};\bar{a}_i)\land\phi(\bar{x};\bar{a}_j)) \text{ if and only if } i<j.
    \]
\end{enumerate}
We say that $\C$ is \emph{stable} if no formula has the order property over $\C$. We say that $\C$ is \emph{NIP} (No Independence Property) if no formula has the independence property over $\C$. Similarly, we say that $\C$ is \emph{NSOP} (No Strict Order Property) if no formula has the strict order property over $\C$.
\end{definition}

In particular, it follows by \cite{baldwin1985second} and \cite{braunfeld2022existential} that a hereditary class $\C$ is stable (resp. NIP) if and only if it is monadically stable (resp. monadically NIP), in the sense of \Cref{def:monstable}.

Shelah's outlook was that instability of a first-order theory must either come from \emph{randomness} (IP) or \emph{order} (SOP). Because of this, NSOP may be viewed as being ``perpendicular'' to NIP. More precisely, the following was established by Shelah, which we relativize to the context of classes of structures once again.

\begin{fact}[\cite{shelah1990classification}]\label{factshelah}
    A class $\C$ is stable if and only if it is NIP and NSOP. 
\end{fact}

In the context of non-hereditary classes of graphs, NSOP captures several well-behaved classes which are not NIP. 

\begin{example*}[Paley graphs]
Let $q = p^n$ be a prime power with $q \equiv 1$ (mod $4)$. We define $P_q$ to be the graph with vertices the elements of the finite field $\mathbb{F}_q$ of order $q$, and edges $xy$ for all $x\neq y$ such that $(x-y)$ is a square in the field. The choice of $q$ ensures that $(x-y)$ is a square if and only if $(y - x)$ is, and so this relation is indeed symmetric. Write $\mathbf{Pal}$ for the class of these graphs.

By a result of Bollobás and Thomason \cite{bollobasthomason}, it follows that the first-order sentences that are asymptotically true in Paley graphs are precisely those that hold in the Rado graph. In other words, any $\str M \models \Th(\mathbf{Pal})$ is either finite, or it is elementary equivalent to the Rado graph. The latter being a prototypical example of a NSOP structure (it is even simple, see \cite{tent_ziegler}) implies that $\mathbf{Pal}$ is NSOP.

On the other hand, the Rado graph contains all finite graphs as induced subgraphs. Consequently, any finite graph occurs as an induced subgraph of some Paley graph. In particular, we may find arbitrarily large induced half-graphs, implying that $\mathbf{Pal}$ is even edge-unstable. 
\end{example*}

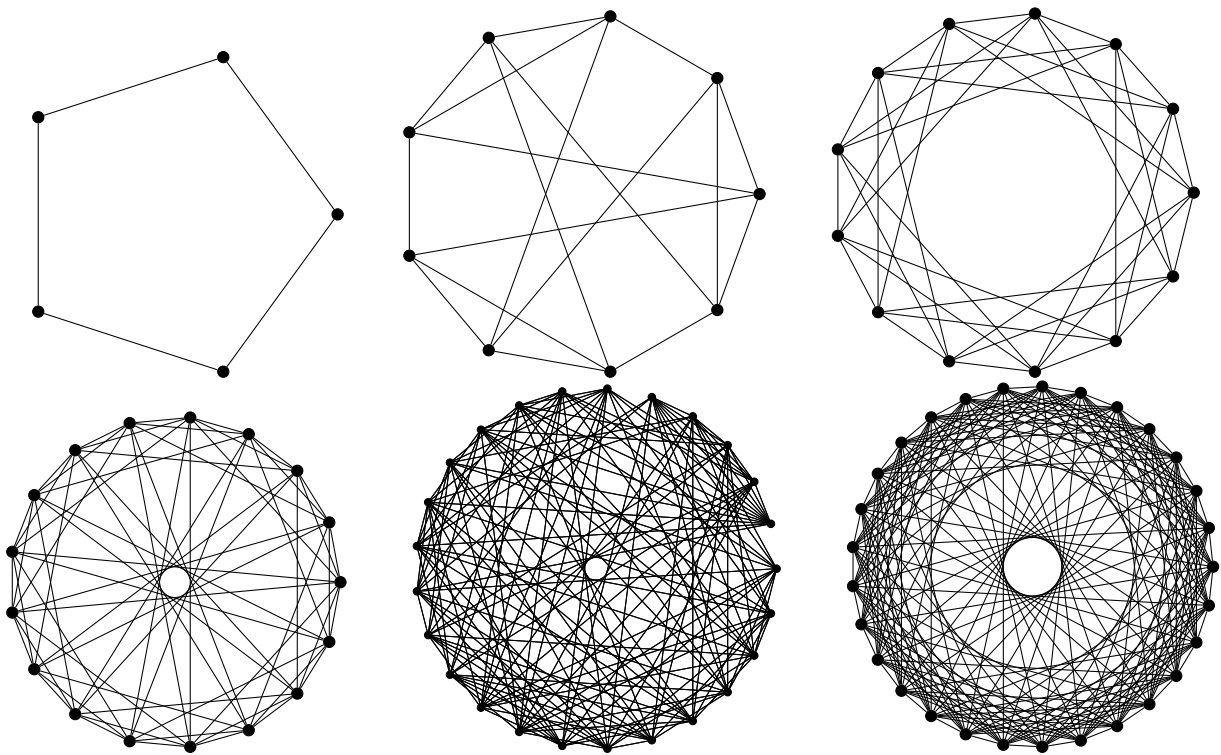
\begin{figure}[h!]
  \centering\small

  \begin{tikzpicture}[scale=1.1]
        \foreach \i in {0,...,4}
  {
    \node[draw,circle,fill=black,inner sep=1.5pt] (N\i) at ({360/5 * \i}:2) {};
  }
    \draw (N0) -- (N1);
    \draw (N1) -- (N2);
    \draw (N2) -- (N3);
    \draw (N3) -- (N4);
    \draw (N4) -- (N0);
  \end{tikzpicture}
  \qquad
  \input{chapters/hardness/paley9}
  \qquad
  \input{chapters/hardness/paley13}
  \qquad
  \input{chapters/hardness/paley17}
  \qquad
  \input{chapters/hardness/paley25}
  \qquad
  \input{chapters/hardness/paley29}
  \caption{The Paley graphs of size 5, 9, 13, 17, 25, and 29 respectively.}

\end{figure}

On the other hand, our results established thus far imply the collapse of (monadic) stability and NSOP in hereditary classes.

\begin{theorem}
    Let $\C$ be a hereditary class of graphs. Then $\C$ is monadically stable if and only if $\C$ is NSOP. Moreover, if $\C$ fails either of the two then this is witnessed by existential formulas. 
\end{theorem}

\begin{proof}
    Clearly, if $\C$ is monadically stable then it is stable, and consequently it is NSOP by \Cref{factshelah}. On the other hand, if $\C$ is hereditary and not monadically stable, then by \Cref{thm:interpret} it is either not edge-stable, or it existentially interprets the class of all graphs. If it is not edge-stable then $\C$ contains arbitrarily large semi-induced half-graphs. In particular, we may assume by Ramsey's theorem that $\C$ contains arbitrarily large semi-induced half-graphs whose sides are either cliques or independent sets. It is then easily seen that the formula $(x=y)\vee E(x,y)$ has the strict order property over $\C$. On the other hand, if $\C$ existentially interprets the class of graphs, then there are existential formulas $\delta(x)$ and $\phi(x,y)$ such that for any graph $G$ we may find some graph $f(G) \in \C$ satisfying $I_{\delta,\phi}(f(G)) \equiv G$. In particular, taking $G:=H_n$ to be the half-graph of size $n$, and $a_1,\dots,a_n$ be the elements of $f(H_n)$ corresponding to the right side of $H_n$, it follows that
    \[ f(H_n) \models \forall x (\delta(x) \land \phi(x,a_i) \to \delta(x) \land \phi(x,a_j)) \text{ if and only if } i \leq j.\]
    Consequently, the existential formula $\delta(x) \land \phi(x,y)$ has the strict order property over $\C$.
\end{proof}

%% file: chapters/hardness/paley9.tex
\begin{tikzpicture}[scale=1.2]


  \foreach \i in {0,...,8}
  {
    \node[draw,circle,fill=black,inner sep=1.5pt] (N\i) at ({360/9 * \i}:2) {};
  }
  \draw (N0) -- (N4);
  \draw (N0) -- (N5);
  \draw (N0) -- (N1);
  \draw (N0) -- (N8);
  \draw (N1) -- (N2);
  \draw (N1) -- (N8);
  \draw (N1) -- (N6);
  \draw (N2) -- (N4);
  \draw (N2) -- (N6);
  \draw (N2) -- (N3);
  \draw (N3) -- (N4);
  \draw (N3) -- (N7);
  \draw (N3) -- (N8);
  \draw (N4) -- (N5);
  \draw (N5) -- (N6);
  \draw (N5) -- (N7);
  \draw (N6) -- (N7);
  \draw (N7) -- (N8);

  \end{tikzpicture}

%% file: chapters/hardness/paley13.tex
\begin{tikzpicture}[scale=1.2]


  \foreach \i in {0,...,12}
  {
    \node[draw,circle,fill=black,inner sep=1.5pt] (N\i) at ({360/13 * \i}:2) {};
  }
  \foreach \i in {0,...,12}
  {
    \foreach \j in {\i,...,12}
    {
      \pgfmathtruncatemacro\diff{mod(\j-\i,13)}
      \foreach \k in {1, 4, 9, 3, 12, 10}
      {
      \ifnum\diff=\k
        \draw (N\i) -- (N\j);
      \fi
      }
    }
  }
\end{tikzpicture}

%% file: chapters/hardness/paley17.tex
\begin{tikzpicture}[scale=1.1]


  \foreach \i in {0,...,16}
  {
    \node[draw,circle,fill=black,inner sep=1.5pt] (N\i) at ({360/17 * \i}:2) {};
  }
  \foreach \i in {0,...,16}
  {
    \foreach \j in {\i,...,16}
    {
      \pgfmathtruncatemacro\diff{mod(\j-\i,17)}
      \foreach \k in  {1,2,4,8,9,13,15,16}
      {
      \ifnum\diff=\k
        \draw (N\i) -- (N\j);
      \fi
      }
    }
  }
\end{tikzpicture}

%% file: chapters/hardness/paley25.tex
\begin{tikzpicture}[scale=1.2]


  \foreach \i in {1,...,25}
  {
    \node[draw,circle,fill=black,inner sep=1pt] (N\i) at ({360/25 * \i}:2) {};
  }

  \foreach \k in {{1,2},{1,3},{1,4},{1,5},{1,6},{1,7},{1,8},{1,9},{1,10},{1,11},{1,12},{1,13},{2,3},{2,4},{2,5},{2,6},{2,7},{2,14},{2,15},{2,16},{2,17},{2,18},{2,19},{3,4},{3,5},{3,8},{3,9},{3,14},{3,15},{3,20},{3,21},{3,22},{3,23},{4,5},{4,10},{4,11},{4,16},{4,17},{4,20},{4,21},{4,24},{4,25},{5,12},{5,13},{5,18},{5,19},{5,22},{5,23},{5,24},{5,25},{6,8},{6,9},{6,10},{6,12},{6,14},{6,16},{6,17},{6,18},{6,22},{6,24},{7,8},{7,9},{7,11},{7,13},{7,15},{7,16},{7,17},{7,19},{7,23},{7,25},{8,10},{8,12},{8,15},{8,16},{8,20},{8,22},{8,23},{8,25},{9,11},{9,13},{9,14},{9,17},{9,21},{9,22},{9,23},{9,24},{10,12},{10,13},{10,14},{10,17},{10,19},{10,20},{10,21},{10,25},{11,12},{11,13},{11,15},{11,16},{11,18},{11,20},{11,21},{11,24},{12,15},{12,18},{12,19},{12,21},{12,23},{12,24},{13,14},{13,18},{13,19},{13,20},{13,22},{13,25},{14,16},{14,18},{14,19},{14,20},{14,21},{14,23},{15,17},{15,18},{15,19},{15,20},{15,21},{15,22},{16,18},{16,20},{16,23},{16,24},{16,25},{17,19},{17,21},{17,22},{17,24},{17,25},{18,21},{18,22},{18,25},{19,20},{19,23},{19,24},{20,22},{20,24},{21,23},{21,25},{22,24},{22,25},{23,24},{23,25}}
           { \foreach \x in \k \foreach \y in \k
           \draw (N\x.center) -- (N\y.center);
        }

\end{tikzpicture}

%% file: chapters/hardness/paley29.tex
\begin{tikzpicture}[scale=1.2]


  \foreach \i in {0,...,28}
  {
    \node[draw,circle,fill=black,inner sep=1.5pt] (N\i) at ({360/29 * \i}:2) {};
  }
  \foreach \i in {0,...,28}
  {
    \foreach \j in {\i,...,28}
    {
      \pgfmathtruncatemacro\diff{mod(\j-\i,29)}
      \foreach \k in  {1, 4, 5, 6, 7, 9, 13, 16, 20, 22, 23, 24, 25, 28}
      {
      \ifnum\diff=\k
        \draw (N\i) -- (N\j);
      \fi
      }
    }
  }
\end{tikzpicture}

%% file: main.bbl
\newcommand{\etalchar}[1]{$^{#1}$}
\begin{thebibliography}{NdMRS19}

\bibitem[AA14]{adler2014interpreting}
Hans Adler and Isolde Adler.
\newblock Interpreting nowhere dense graph classes as a classical notion of
  model theory.
\newblock {\em European Journal of Combinatorics}, 36:322--330, 2014.

\bibitem[BGO{\etalchar{+}}22]{BonnetGMSTT22}
{\'{E}}douard Bonnet, Ugo Giocanti, Patrice {Ossona de Mendez}, Pierre Simon,
  St{\'{e}}phan Thomass{\'{e}}, and Szymon Toru\'nczyk.
\newblock Twin-width {IV:} ordered graphs and matrices.
\newblock In {\em 54th Annual {ACM} {SIGACT} Symposium on Theory of Computing,
  {STOC} '22}, pages 924--937. {ACM}, 2022.

\bibitem[BKTW22]{BonnetKTW22}
{\'{E}}douard Bonnet, Eun~Jung Kim, St{\'{e}}phan Thomass{\'{e}}, and
  R{\'{e}}mi Watrigant.
\newblock Twin-width {I:} {T}ractable {FO} {M}odel {C}hecking.
\newblock {\em J. {ACM}}, 69(1):3:1--3:46, 2022.

\bibitem[BL22]{braunfeld2022existential}
Samuel Braunfeld and Michael~C. Laskowski.
\newblock Existential characterizations of monadic {NIP}.
\newblock {\em arXiv preprint arXiv:2209.05120}, 2022.

\bibitem[BS85]{baldwin1985second}
John~T Baldwin and Saharon Shelah.
\newblock Second-order quantifiers and the complexity of theories.
\newblock {\em Notre Dame Journal of Formal Logic}, 26(3):229--303, 1985.

\bibitem[BT81]{bollobasthomason}
Béla Bollobás and Andrew Thomason.
\newblock Graphs which contain all small graphs.
\newblock {\em European Journal of Combinatorics}, 2(1):13--15, 1981.

\bibitem[CMR00]{courcelle2000linear}
Bruno Courcelle, Johann~A Makowsky, and Udi Rotics.
\newblock Linear time solvable optimization problems on graphs of bounded
  clique-width.
\newblock {\em Theory of Computing Systems}, 33(2):125--150, 2000.

\bibitem[CW89]{ChazelleW89}
Bernard Chazelle and Emo Welzl.
\newblock Quasi-optimal range searching in space of finite {VC}-dimension.
\newblock {\em Discret. Comput. Geom.}, 4:467--489, 1989.

\bibitem[Daw10]{Dawar10}
Anuj Dawar.
\newblock Homomorphism preservation on quasi-wide classes.
\newblock {\em J. Comput. Syst. Sci.}, 76(5):324--332, 2010.

\bibitem[DGK{\etalchar{+}}22]{bushes}
Jan Dreier, Jakub Gajarsk\'{y}, Sandra Kiefer, Micha\l{} Pilipczuk, and Szymon
  Toru\'{n}czyk.
\newblock Treelike decompositions for transductions of sparse graphs.
\newblock In {\em 37th Annual ACM/IEEE Symposium on Logic in Computer Science,
  LICS 2022}. ACM, 2022.

\bibitem[DMS23]{ssmc}
Jan Dreier, Nikolas M\"{a}hlmann, and Sebastian Siebertz.
\newblock First-order model checking on structurally sparse graph classes.
\newblock In {\em 55th Annual ACM Symposium on Theory of Computing, STOC 2023},
  pages 567--580. ACM, 2023.
\newblock Full version available at arXiv:2302.03527.

\bibitem[DMST23]{DreierMST23}
Jan Dreier, Nikolas M{\"{a}}hlmann, Sebastian Siebertz, and Szymon Toru\'nczyk.
\newblock Indiscernibles and flatness in monadically stable and monadically
  {NIP} classes.
\newblock In {\em 50th International Colloquium on Automata, Languages, and
  Programming, {ICALP} 2023}, volume 261 of {\em LIPIcs}, pages 125:1--125:18.
  Schloss Dagstuhl --- Leibniz-Zentrum f{\"{u}}r Informatik, 2023.

\bibitem[DMT23]{CharacterizingNIP}
Jan Dreier, Nikolas M\"ahlmann, and Szymon Toru\'nczyk.
\newblock Combinatorial characterizations of monadically {NIP} graph classes.
\newblock {\em Manuscript, forthcoming}, 2023.

\bibitem[Dvo07]{dvorak-thesis}
Zden\v{e}k Dvo\v{r}\'ak.
\newblock {\em Asymptotical structure of combinatorial objects}.
\newblock PhD thesis, Charles University, Prague, 2007.

\bibitem[Dvo18]{Dvorak18}
Zden\v{e}k Dvo\v{r}{\'{a}}k.
\newblock Induced subdivisions and bounded expansion.
\newblock {\em European Journal of Combinatorics}, 69:143--148, 2018.

\bibitem[EGK{\etalchar{+}}17]{DBLP:conf/icalp/EickmeyerGKKPRS17}
Kord Eickmeyer, Archontia~C. Giannopoulou, Stephan Kreutzer, O{-}joung Kwon,
  Micha\l{} Pilipczuk, Roman Rabinovich, and Sebastian Siebertz.
\newblock Neighborhood complexity and kernelization for nowhere dense classes
  of graphs.
\newblock In {\em 44th International Colloquium on Automata, Languages, and
  Programming, {ICALP} 2017}, volume~80 of {\em LIPIcs}, pages 63:1--63:14.
  Schloss Dagstuhl --- Leibniz-Zentrum f{\"{u}}r Informatik, 2017.

\bibitem[GHO{\etalchar{+}}20]{GajarskyHOLR20}
Jakub Gajarsk{\'{y}}, Petr Hlin\v{e}n{\'{y}}, Jan Obdr\v{z}{\'{a}}lek, Daniel
  Lokshtanov, and M.~S. Ramanujan.
\newblock A new perspective on {FO} model checking of dense graph classes.
\newblock {\em {ACM} Trans. Comput. Log.}, 21(4):28:1--28:23, 2020.

\bibitem[GKS17]{GroheKS17}
Martin Grohe, Stephan Kreutzer, and Sebastian Siebertz.
\newblock Deciding first-order properties of nowhere dense graphs.
\newblock {\em J. {ACM}}, 64(3):17:1--17:32, 2017.

\bibitem[GMM{\etalchar{+}}23]{flippergame}
Jakub Gajarsk{\'{y}}, Nikolas M{\"{a}}hlmann, Rose McCarty, Pierre Ohlmann,
  Micha\l{} Pilipczuk, Wojciech Przybyszewski, Sebastian Siebertz, Marek
  Soko\l{}owski, and Szymon Toru\'nczyk.
\newblock Flipper games for monadically stable graph classes.
\newblock In {\em 50th International Colloquium on Automata, Languages, and
  Programming, {ICALP} 2023}, volume 261 of {\em LIPIcs}, pages 128:1--128:16.
  Schloss Dagstuhl --- Leibniz-Zentrum f{\"{u}}r Informatik, 2023.

\bibitem[Hen71]{henson}
C.~Ward Henson.
\newblock {A family of countable homogeneous graphs.}
\newblock {\em Pacific Journal of Mathematics}, 38(1):69 -- 83, 1971.

\bibitem[Hod93]{hodges_1993}
Wilfrid Hodges.
\newblock {\em Model Theory}.
\newblock Encyclopedia of Mathematics and its Applications. Cambridge
  University Press, 1993.

\bibitem[Hod97]{hodges}
Wilfrid Hodges.
\newblock {\em A Shorter Model Theory}.
\newblock Cambridge University Press, 1997.

\bibitem[NdMRS19]{DBLP:journals/corr/abs-1909-01564}
Jaroslav Ne{\v{s}}et{\v{r}}il, Patrice~Ossona de~Mendez, Roman Rabinovich, and
  Sebastian Siebertz.
\newblock Classes of graphs with low complexity: the case of classes with
  bounded linear rankwidth.
\newblock {\em arXiv preprint 1909.01564}, 2019.

\bibitem[NOP{\etalchar{+}}21]{NesetrilMPRS21}
Jaroslav Ne\v{s}et\v{r}il, Patrice {Ossona de Mendez}, Micha\l{} Pilipczuk,
  Roman Rabinovich, and Sebastian Siebertz.
\newblock Rankwidth meets stability.
\newblock In {\em 2021 {ACM-SIAM} Symposium on Discrete Algorithms, {SODA}
  2021}, pages 2014--2033. {SIAM}, 2021.

\bibitem[{Oss}21]{POM21}
Patrice {Ossona de Mendez}.
\newblock First-order transductions of graphs (invited talk).
\newblock In {\em 38th International Symposium on Theoretical Aspects of
  Computer Science, {STACS} 2021}, volume 187 of {\em LIPIcs}, pages 2:1--2:7.
  Schloss Dagstuhl --- Leibniz-Zentrum f{\"{u}}r Informatik, 2021.

\bibitem[PPS20]{sparsityNotes}
Marcin Pilipczuk, Micha\l{} Pilipczuk, and Sebastian Siebertz.
\newblock Lecture notes for the course ``{S}parsity'' given at {F}aculty of
  {M}athematics, {I}nformatics, and {M}echanics of the {U}niversity of
  {W}arsaw, Winter semesters 2017/18 and 2019/20.
\newblock Available online at
  \url{https://www.mimuw.edu.pl/~mp248287/sparsity2}.

\bibitem[PST18]{PilipczukST18a}
Micha\l{} Pilipczuk, Sebastian Siebertz, and Szymon Toru\'nczyk.
\newblock On the number of types in sparse graphs.
\newblock In {\em 33rd Annual {ACM/IEEE} Symposium on Logic in Computer
  Science, {LICS} 2018}, pages 799--808. {ACM}, 2018.

\bibitem[PZ78]{podewski1978stable}
Klaus-Peter Podewski and Martin Ziegler.
\newblock Stable graphs.
\newblock {\em Fundamenta Mathematicae}, 100(2):101--107, 1978.

\bibitem[Rad64]{rado}
Richard Rado.
\newblock Universal graphs and universal functions.
\newblock {\em Acta Arithmetica}, 9(4):331--340, 1964.

\bibitem[Sau72]{sauer1972density}
Norbert Sauer.
\newblock On the density of families of sets.
\newblock {\em Journal of Combinatorial Theory, Series A}, 13(1):145--147,
  1972.

\bibitem[She72]{shelah1972combinatorial}
Saharon Shelah.
\newblock A combinatorial problem; stability and order for models and theories
  in infinitary languages.
\newblock {\em Pacific Journal of Mathematics}, 41(1):247--261, 1972.

\bibitem[She90]{shelah1990classification}
Saharon Shelah.
\newblock {\em Classification theory: and the number of non-isomorphic models}.
\newblock Elsevier, 1990.

\bibitem[TZ12]{tent_ziegler}
Katrin Tent and Martin Ziegler.
\newblock {\em A Course in Model Theory}.
\newblock Lecture Notes in Logic. Cambridge University Press, 2012.

\bibitem[War16]{Warwick16}
Open problems from the workshop on algorithms, logic and structure,
  {U}niversity of {W}arwick, 2016, 2016.
\newblock Available online at
  \url{https://warwick.ac.uk/fac/sci/maths/people/staff/daniel_kral/alglogstr/openproblems.pdf},
  access date Nov 9, 2023.

\bibitem[Wel88]{Welzl89}
Emo Welzl.
\newblock Partition trees for triangle counting and other range searching
  problems.
\newblock In {\em Fourth Annual Symposium on Computational Geometry, SoCG
  1988}, pages 23--33. {ACM}, 1988.

\end{thebibliography}
